\documentclass[acmsmall,screen,nonacm,review=false,timestamp=false]{acmart}

\usepackage{afterpage}
\usepackage{amsmath}
\usepackage[inline,shortlabels]{enumitem}
\usepackage{booktabs}
\usepackage{graphicx}
\usepackage{xspace}
\usepackage[T1]{fontenc}
\usepackage{listings}
\usepackage[scaled=0.75]{beramono}
\usepackage[textsize=tiny]{todonotes}
\usepackage{stmaryrd}
\usepackage{caption}
\usepackage{subcaption}

\newif\iflong
\longfalse

\newcommand{\ie}{\emph{i.e.}\xspace}

\newcommand{\eg}{\emph{e.g.}\xspace}

\newcommand{\emphbf}[1]{\emph{\textbf{#1}\xspace}}
\newcommand{\mypara}[1]{\smallskip\noindent\emphbf{#1.}\xspace}

\newcommand{\Egraphs}{E-graphs\xspace}
\newcommand{\egraph}{e-graph\xspace}
\newcommand{\egraphs}{e-graphs\xspace}
\newcommand{\eclass}{e-class\xspace}
\newcommand{\eclasses}{e-classes\xspace}

\newif\ifdraft
\drafttrue

\ifdraft
\newcommand\authorrnote[3]{\textcolor{#1}{({#2}: {#3})}\xspace}
\newcommand{\TODO}[1]{{\color{orange!80!black}[\textsl{#1}]}}
\else
\newcommand\authorrnote[2]{}
\newcommand{\TODO}[1]{}
\fi

\newcommand{\tname}[1]{\textsc{#1}\xspace}
\newcommand{\babble}{\tname{babble}}
\newcommand{\tool}{\tname{babble}}
\newcommand{\dc}{\tname{DreamCoder}}
\newcommand{\stitch}{\tname{Stitch}}

\newcommand{\T}[1]{\mbox{\lstinline^#1^}}

\definecolor[named]{darkblue}{cmyk}{1,0.58,0,0.21}
\definecolor[named]{darkred}{cmyk}{0,1,0.87,0.21}

\newcommand{\semeq}{\equiv}
\newcommand{\many}[1]{\overline{#1}}
\newcommand{\au}{\mathsf{AU}}
\newcommand{\gau}{\mathsf{AU}}
\newcommand{\fv}{\mathsf{vars}}
\newcommand{\sz}{\mathsf{size}}
\newcommand{\sk}{\mathsf{skeleton}}
\newcommand{\arity}{\mathsf{arity}}
\newcommand{\termsof}[1]{\mathcal{T}(#1)}
\newcommand{\patsof}[2]{\mathcal{T}(#1,#2)}
\newcommand{\absof}[2]{\hat{\mathcal{T}}(#1,#2)}

\newcommand{\rng}[1]{\mathsf{range}(#1)}
\newcommand{\subst}[2]{{#1} \mapsto {#2}}
\newcommand{\sapp}[2]{{#1}({#2})}
\newcommand{\sub}{\sqsubseteq}

\newcommand{\join}{\sqcup}
\newcommand{\bnd}{\ \mathbf{\shortrightarrow}\ }
\newcommand{\match}[2]{\mathsf{match}(#1,#2)}
\newcommand{\matches}[2]{\mathsf{matches}(#1,#2)}
\newcommand{\rwsteps}[1]{\rightarrow^{#1}_1}
\newcommand{\rewrites}[1]{\rightarrow^{#1}}
\newcommand{\betared}{\rwsteps{\beta}}
\newcommand{\evals}{\rewrites{\beta}}

\newcommand{\cs}{\mathsf{costset}} % costset
\newcommand{\subterms}{\mathsf{subterms}}
\newcommand{\occ}{\mathsf{occurs}}

\newcommand{\pl}[1]{\mathcal{P}(#1)}
\newcommand{\plau}[1]{\mathcal{P}^2(#1)}
\newcommand{\denot}[1]{\llbracket #1 \rrbracket}
\newcommand{\eqg}{\semeq^{\mathcal{G}}}
\newcommand{\dominant}{\mathsf{dominant}}
\newcommand{\extract}{\mathsf{extract}}
\newcommand{\gramdef}{\ensuremath{\mathrel{\mathord{:}\mathord{:=}}}}
\newcommand{\contains}[2]{{#1} \prec {#2}}
\newcommand{\cost}{\mathsf{cost}}
\DeclareMathOperator*{\argmin}{arg\,min}
\newcommand{\ecau}[3]{\gau(#1 \vdash #2,#3)}
\newcommand{\intro}{\sz}
\newcommand{\emptyseq}{\epsilon}
\newcommand{\ctx}{\mathcal{C}}
\newcommand{\ruleset}{\mathcal{R}}
\newcommand{\patset}{\mathcal{P}}
\newcommand{\libset}{\mathcal{L}}

\newcommand{\sfunify}{\ensuremath{\textcolor{darkblue}{\mathsf{reduce}}}}
\newcommand{\sfprune}{\ensuremath{\textcolor{darkred}{\mathsf{prune}}}}
\newcommand{\sftopk}{\ensuremath{\mathsf{top\_k}}}

\newcommand{\cogsci}{\tname{2d cad}}
\newcommand{\nuts}{Nuts-and-bolts\xspace}

\begin{document}

%% Title information
\title{\tool: Learning Better Abstractions with E-Graphs and Anti-Unification}
\renewcommand{\shorttitle}{Learning Better Abstractions with E-Graphs and Anti-Unification}

%% Author with single affiliation.
\author{David Cao}
\authornote{Equal contribution}
\affiliation{
  \institution{UC San Diego}
 \country{USA}
}
\email{dmcao@ucsd.edu}

\author{Rose Kunkel}
\authornotemark[1]
\affiliation{
  \institution{UC San Diego}
 \country{USA}
}
\email{rkunkel@eng.ucsd.edu}

\author{Chandrakana Nandi}
\affiliation{
 \institution{Certora, Inc.}
 \country{USA}
}
\email{chandra@certora.com}

\author{Max Willsey}
\affiliation{
  \institution{University of Washington}
 \country{USA}
}
\email{mwillsey@cs.washington.edu}

\author{Zachary Tatlock}
\affiliation{
  \institution{University of Washington}
 \country{USA}
}
\email{ztatlock@cs.washington.edu}

\author{Nadia Polikarpova}
\affiliation{
  \institution{UC San Diego}            %% \institution is required
 \country{USA}
}
\email{npolikarpova@eng.ucsd.edu}          %% \email is recommended

\renewcommand{\shortauthors}{Cao, D., Kunkel, R., Nandi, C., Willsey, M., Tatlock Z., Polikarpova, N.}

\begin{abstract}
\emph{Library learning}
  compresses a given corpus of programs
  by extracting common structure from
  the corpus into reusable library functions.
Prior work on library learning suffers from two limitations
 that prevent it from scaling to larger, more complex inputs.
First, it explores too many candidate library functions
  that are not useful for compression. 
Second, it is not robust to syntactic variation in
 the input.

We propose \emph{library learning modulo theory} (LLMT),
 a new library learning algorithm
 that additionally takes as input
 an equational theory for a given problem domain.
LLMT uses e-graphs and equality saturation
 to compactly represent the space of programs equivalent modulo the theory,
 and uses a novel \emph{e-graph anti-unification} technique
 to find common patterns in the corpus more directly and efficiently.

We implemented LLMT in a tool named \tool.
Our evaluation shows that \tool
 achieves better compression orders of magnitude faster
 than the state of the art.
We also provide a qualitative evaluation
 showing that \tool learns reusable functions
 on inputs previously out of reach for library learning.

\end{abstract}

\begin{CCSXML}
  <ccs2012>
     <concept>
         <concept_id>10011007.10011006.10011008.10011009.10011012</concept_id>
         <concept_desc>Software and its engineering~Functional languages</concept_desc>
         <concept_significance>500</concept_significance>
     </concept>
     <concept>
         <concept_id>10011007.10011074.10011092.10011782</concept_id>
         <concept_desc>Software and its engineering~Automatic programming</concept_desc>
         <concept_significance>500</concept_significance>
      </concept>
   </ccs2012>
\end{CCSXML}
  
\ccsdesc[500]{Software and its engineering~Functional languages}
\ccsdesc[500]{Software and its engineering~Automatic programming}
  
%% Keywords
%% comma separated list
\keywords{library learning, e-graphs, anti-unification}

\maketitle

\section{Introduction}
\label{sec:intro}

Abstraction is the key to managing software complexity.
Experienced programmers routinely extract common functionality
into libraries of reusable abstractions
to express their intent more clearly and concisely.
What if this process of extracting useful abstractions from code could be automated?
\emph{Library learning} seeks to answer this question
  with techniques to compress
  a given corpus of programs
  by extracting common structure into reusable library functions.
Library learning
  has many potential applications
  from refactoring and decompilation~\cite{szalinski,shapemod},
  to modeling human cognition~\cite{cogsci-smiley,cogsci-dataset},
  and speeding up program synthesis by
  specializing the target language to
  a chosen problem domain~\cite{dreamcoder}.

Consider the simple library learning task in \autoref{fig:polygons}.
On the left,
 \autoref{fig:polygons}a shows a corpus of three programs in a
 \cogsci DSL from \citet{cogsci-dataset}.
Each program corresponds to a picture
 composed of regular polygons,
 each of which is made of multiple rotated line segments.
On the right,
 \autoref{fig:polygons}b shows a learned library
 with a single function (named \T{f0})
 that abstracts away the construction of scaled regular polygons.
The three input programs can then be refactored
 into a more concise form using the learned \T{f0}.
% In this case,
%  abstracting regular polygons is intuitively ``correct'',
%  but in general, this is a qualitative choice.
Whether \T{f0} is the ``best'' abstraction for this corpus is generally hard to quantify.
For this paper,
 we follow \dc~\cite{dreamcoder} and use \emph{compression} as a
 metric for library learning, \ie,
 the goal is to reduce the total size of the corpus in AST nodes
 (from 208 to 72 \autoref{fig:polygons}).
Importantly, the total size of the corpus includes the size of the library:
this prevents library learning from generating too many overly-specific functions,
and instead biases it towards more general and reusable abstractions. % that can be reused multiple times.

\begin{figure}
    \centering
    \includegraphics[width=\textwidth]{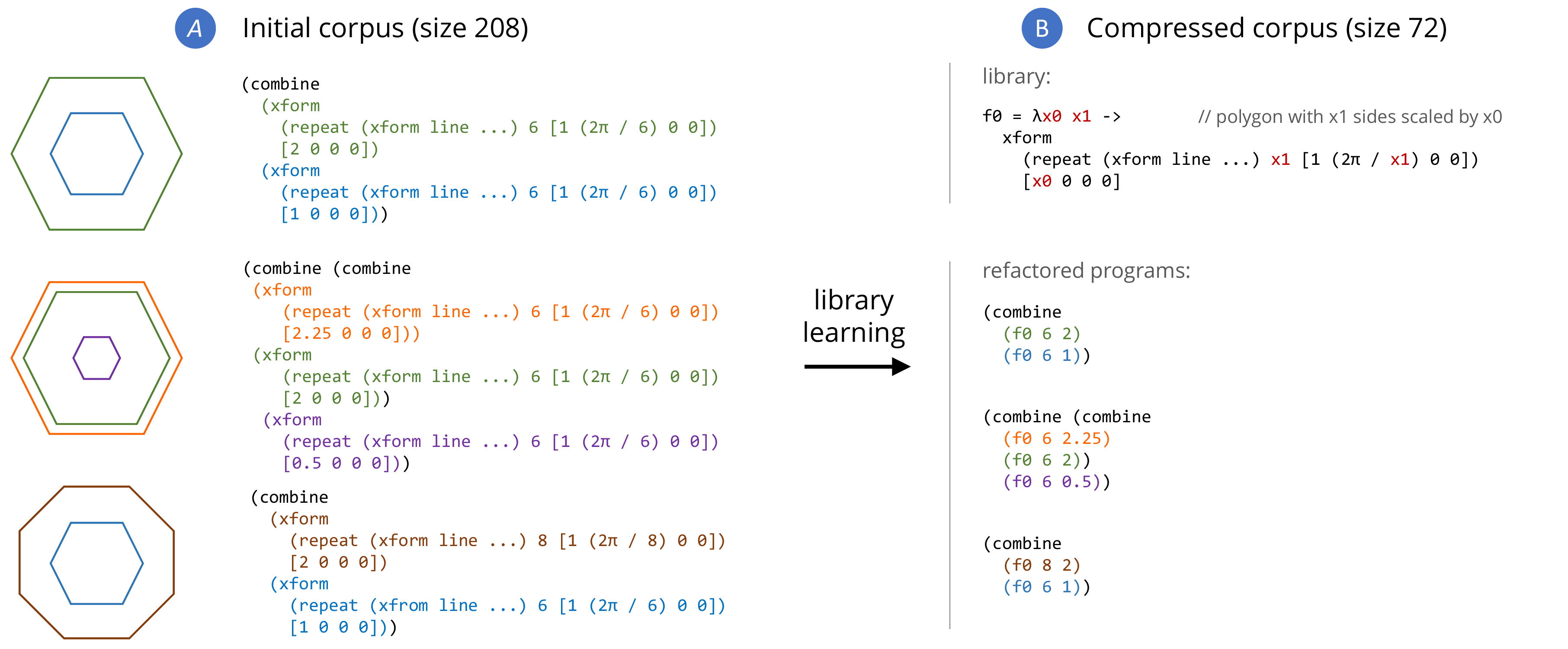}
    \caption{Example of library learning.
             Initial corpus A contains three graphical programs from the ``nuts \& bolts'' dataset of \citet{cogsci-dataset}.
             Corpus B is the output of library learning with a single learned function for a scaled polygon,
             and the original programs refactored using this function.}\label{fig:polygons}
\end{figure}

Library learning can be phrased as a program synthesis problem
 structured in two phases:
%  \textit{Generate} -- generating candidate abstractions, and then
%  \textit{Select} -- selecting those abstractions that produce the best (smallest)
\emph{generating} candidate abstractions, and then
\emph{selecting} those abstractions that produce the best (smallest)
     refactored corpus.
The state-of-the-art technique, implemented in \dc~\cite{dreamcoder},
  suffers from two primary limitations % in the \textit{Generate} phase
  that hinder scaling library learning to larger and more realistic inputs.
\begin{itemize}
    \item
    Candidate generation is not \emphbf{precise}:
    \dc generates many candidate abstractions that cannot be useful,
    slowing down the selection phase and the algorithm as a whole.
    % The result is that many candidates are sent to the selection phase
    %  that should never be selected, slowing down the algorithm as a whole.
    \item
    % Candidate generation is not \emphbf{robust} to syntactic variation.
    % Candidate generation is not \emphbf{robust}:
    The technique is purely syntactic and hence not \emphbf{robust} to superficial variation.
    % \dc proposes library function candidates purely syntactically
    % and hence is brittle to superficial variation.
    For example,
     a human programmer would immediately know that the terms $2 + 1$ and $1 + 3$
     can be refactored using the abstraction $\lambda x \bnd x + 1$,
     because addition commutes;
     a purely syntactic library learning approach cannot generate this abstraction.
    % For example,
    %  a purely syntactic algorithm cannot
    %  generate an abstraction for the programs $2 + 1$ and $1 + 3$.
    % However,
    %  since we know that addition commutes,
    %  $\lambda x \bnd x + 1$ applies in both programs.
    % \item
    % \mw{something about beam extraction}
\end{itemize}
In this paper we propose \emph{library learning modulo (equational) theories} (LLMT)---%
a new library learning algorithm
that addresses both of the above limitations.
% from the \textit{Generate} phase of library learning.

\mypara{Precise Candidate Generation via Anti-Unification}
To make candidate generation more precise,
%To make the proposal stage of library learning more precise,
 LLMT leverages two key observations:
\begin{itemize}
    \item
    Useful abstractions must be used in the corpus at least twice.
    For example, in a corpus of two programs $2 + 1$ and $3 + 1$,
    there is no need to consider $\lambda x \bnd 3 + x$
    because it can only be used in one place,
    and hence would only increase the size of the corpus.
    %% Nadia: there's nothing wrong with using this example,
    %% but the asymmetry with the second example annoys me.
    % For example, there is no need to consider $\lambda x \bnd \T{trans}\ x\ [0.5\ 0\ 0\ 0]$
    %  in our running example,
    %  because it can only be used in one place;
    %  it would only increase the size of the program.
    \item
    Abstractions should be ``as concrete as possible'' for a given corpus.
    For example,
     in the same corpus with $2 + 1$ and $3 + 1$,
     the abstraction $\lambda x \bnd x + 1$
     is superior to the more general
     $\lambda x\ y \bnd x + y$,
     since both apply to the same two terms,
     but applying the latter is more expensive (it requires two arguments).
    %  since the argument for $y$ in the latter would always be $1$.
    % In our example, the chosen \texttt{f0} of
    %   $\lambda x \bnd \T{repeat}\ (\T{trans line}\ \ldots)\ x\ [1\ (2\pi / x)\ 0\ 0]$
    %   is superior
    %   the more general
    %   $\lambda x\ y \bnd \T{repeat}\ (\T{trans line}\ \ldots)\ x\ y$,
    %   because the argument for $y$ in the latter abstraction will always be exactly
    %   $[1\ (2\pi / x)\ 0\ 0]$.
\end{itemize}
In other words, a useful abstraction corresponds to the least general \emph{pattern}
that matches some pair of subterms from the original corpus;
such a pattern can be computed via \emph{anti-unification} (AU)~\cite{plotkin1970lattice}.
For example, anti-unifying $2 + 1$ and $3 + 1$
  yields the pattern $X + 1$,
  and the desired candidate library function $\lambda x \bnd x + 1$
  can be derived by abstracting over the pattern variable $X$.
  % \footnote{Hereafter, we use upper-case letters for pattern variables and lower-case letters for program variables.}
%
Similarly, in \autoref{fig:polygons}, the abstraction \T{f0} can be derived
 by anti-unifying, for example, the blue and the brown subterms of the corpus.
%
% The remaining challenge is to efficiently compute the set of anti-unifiers
% between \emph{all} pairs of subterms in the corpus;
% LLMT achieves this via a new bottom-up dynamic programming algorithm.

\mypara{Robustness via E-Graphs}
To make candidate generation more robust,
%To address the robustness issue,
LLMT additionally takes as input a \emph{domain-specific equational theory} % (expressed as a set of rewrite rules)
and uses it to find programs that are semantically equivalent to the original corpus,
but share the most syntactic structure.
For example, in the domain of arithmetic, we expect the theory to contain the equation
$X + Y  \semeq  Y + X$,
which states that addition is commutative.
Given the corpus with terms $2 + 1$ and $1 + 3$,
this rule can be used to rewrite the second term in to $3 + 1$,
enabling the discovery of the common abstraction $\lambda x \bnd x + 1$.

% For example, for the graphical DSL from our running example,
% the equational theory might contain a rule that adds a unit transformation to any shape $s$:
% $$
% s  \Rightarrow  \T{trans}\ s\ [1\ 0\ 0\ 0]
% $$
% %
% This rule can then be used to transform the two unit hexagons in corpus C from \autoref{fig:llmt}
% into the form they have in the original corpus in \autoref{fig:polygons},
% which enabled the discovery of the scaled polygon abstraction $f_0$
% and using it to compress the corpus to size 72.

The main challenge with this approach is to search over the large space of programs
that are semantically equivalent to the original corpus.
To this end, LLMT relies on the \emph{e-graph} data structure
and the \emph{equality saturation} technique~\cite{tate2009equality,willsey2021egg}
to compute and represent the space of semantically equivalent programs.
To enable efficient library learning over this space,
% we extend our AU-based candidate generation algorithm to work over e-graphs.
we propose a new candidate generation algorithm
that efficiently computes the set of all anti-unifiers
between pairs of sub-terms represented by an e-graph,
using dynamic programming.

Finally, \emph{selecting} the optimal library
%   (the \textit{Select} phase of library learning)
  in this setting reduces to the problem of
  \emph{extracting} the smallest term out of an e-graph
  in the presence of common sub-expressions.
This problem is extremely computationally intensive in its general form,
and existing approaches have limited scalability~\cite{tensat}.
Instead we propose \emph{targeted subexpression elimination}:
a novel e-graph extraction algorithm
that uses domain-specific knowledge to reduce the search space
and readily supports approximation via beam search to trade off accuracy and efficiency.
% Finally, to select the optimal library,
% we propose \emph{beam extraction}:
% a new algorithm for extracting the best term from an e-graph,
% which is based on beam search and can trade off accuracy and efficiency.

\begin{figure}
    \centering
    \includegraphics[width=.8\textwidth]{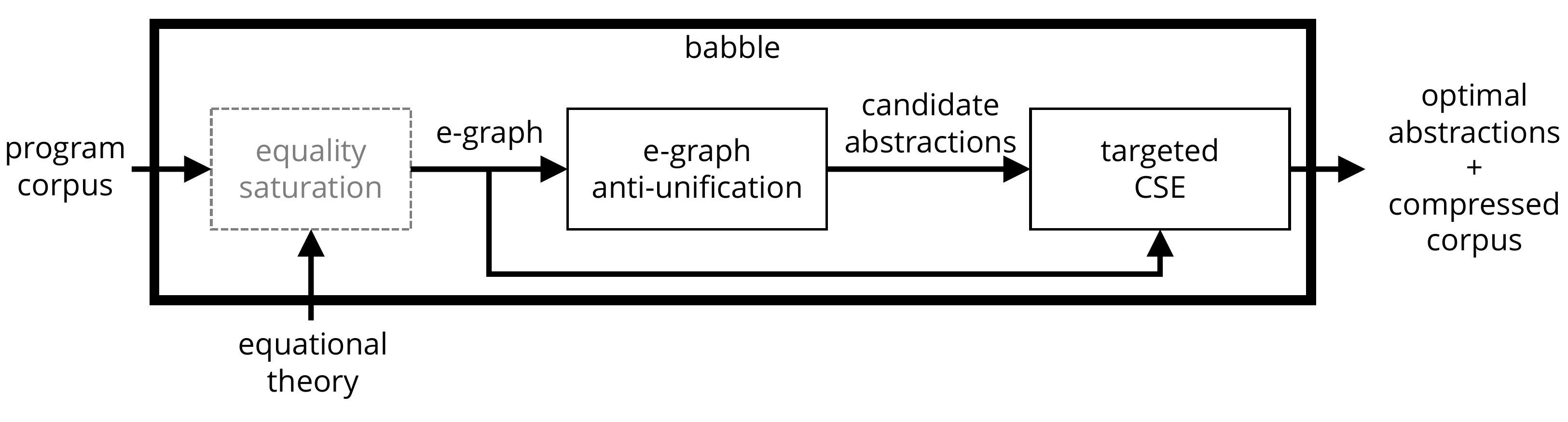}
    \caption{\tool architecture overview}\label{fig:arch}
  \end{figure}

\mypara{\tool}
We have implemented LLMT in \tool,
a tool built on top of the \tname{egg} e-graph library~\cite{willsey2021egg}.
The overview of the \tool's workflow is shown in \autoref{fig:arch},
with gray boxes representing existing techniques
and black boxes representing our contributions.
% Given the input corpus and rewrite rules from \autoref{fig:llmt},
% \tool is able to discover the abstraction $f_0$ and the compressed corpus B
% in a fraction of a second.

We evaluated \tool on benchmarks from two sources:
compression tasks extracted from \dc runs~\cite{compression-bench}
and \cogsci programs used to evaluate concept learning in humans~\cite{cogsci-dataset}.
Our evaluation shows that \tool outperforms \dc on its own benchmarks,
 achieving better compression in orders of magnitude less time.
Adding domain-specific rewrites improves compression even further.
We also show that \tool scales to the larger \cogsci corpora,
which is beyond reach of \dc.
% %
% Although adding domain-specific rewrites does impact scalability,
% it also makes the results more robust to syntactic variations in the corpus.
%
We also present and discuss a selection of abstractions learned by \tool,
 demonstrating that the LLMT approach can learn functions that match human intuition.

\mypara{Contributions}
In summary, this paper makes the following contributions:
\begin{itemize}
    \item \emph{library learning modulo equational theory}:
        a library learning algorithm that can incorporate an arbitrary domain-specific equational theory
        to make learning robust to syntactic variations;
    \item \emph{e-graph anti-unification}:
        an algorithm that efficiently generates the set of candidate abstractions
        using the mechanism of anti-unification extended from terms to e-graphs;
%   \item an anti-unification algorithm that operates on e-graphs rather than
%         concrete terms,
  \item \emph{targeted common subexpression elimination}:
        a new approximate algorithm for extracting the best term from an e-graph
        in the presence of common sub-expressions.
\end{itemize}

\section{Overview}
\label{sec:overview}

\newcommand{\xargs}[4]{\ensuremath{[#1, #2, #3, #4]}}
\newcommand{\xform}[2]{\ensuremath{\T{xform}(#1,\ #2)}}
\newcommand{\rep}[3]{\ensuremath{\T{repeat}(#1,\ #2,\ #3)}}
\newcommand{\comb}[2]{\ensuremath{\T{combine}(#1,\ #2)}}

\newcommand{\scale}[2]{\ensuremath{\T{scale}(#1,\ #2)}}
\newcommand{\repRot}[3]{\ensuremath{\T{repRot}(#1,\ #2,\ #3)}}
\newcommand{\rotate}[2]{\ensuremath{\T{rotate}(#1,\ #2)}}
\newcommand{\translate}[3]{\ensuremath{\T{translate}(#1,\ #2,\ #3)}}

\newcommand{\pSide}[1]{\ensuremath{\T{side}(#1)}}
\newcommand{\ngon}[1]{\ensuremath{\T{ngon}(#1)}}
\newcommand{\scaledHex}[1]{\ensuremath{\T{scaledHexagon}(#1)}}

%\newcommand{\scale}[2]{\xform{#1}{\xargs{#2}{0}{0}{0}}}
%\newcommand{\rotate}[2]{\xform{#1}{\xargs{1}{#2}{0}{0}}}
%\newcommand{\translateX}[2]{\xform{#1}{\xargs{1}{0}{#2}{0}}}
%\newcommand{\translateY}[2]{\xform{#1}{\xargs{1}{0}{0}{#2}}}
%\newcommand{\translate}[3]{\xform{#1}{\xargs{1}{0}{#2}{#3}}}

%\newcommand{\scale}[2]{\ensuremath{\T{scale}\ #1\ #2}}
%\newcommand{\rot}[2]{\ensuremath{\T{rotate}\ #1\ #2}}
%\newcommand{\translX}[2]{\ensuremath{\T{translateX}\ #1\ #2}}
%\newcommand{\translY}[2]{\ensuremath{\T{translateY}\ #1\ #2}}
%\newcommand{\transl}[3]{\ensuremath{\T{translate}\ #1\ (#2, #3)}}

%In this section we illustrate our approach using a running example
%that builds upon \autoref{fig:polygons} from the introduction.

We illustrate LLMT via a running example building on \autoref{fig:polygons}.
The input corpus in \autoref{fig:polygons} is written in a \cogsci DSL by \citet{cogsci-dataset},
which features the following primitives:

\vspace{0.25em}
\begin{tabular}{rl}
  \T{line} &
    a line segment from
    the origin (0, 0) to the point (1, 0)
    \\
  \comb{F_1}{F_2} &
    the union of figures $F_1$ and $F_2$
    \\
  \xform{F}{\tau} &
    applying transformation $\tau$ to figure $F$
    \\
  \rep{F}{0}{\tau} &
    the empty figure
    \\
  \rep{F}{n+1}{\tau} &
    \comb{F}{\xform{\rep{F}{n}{\tau}}{\tau}}, similar to a ``fold''
\end{tabular}
\vspace{0.25em}

%%  \begin{itemize}
%%    \item
%%      \T{line} denotes
%%        a line segment from
%%        the origin (0, 0) to
%%        the point (1, 0)
%%    \item
%%      \comb{F_1}{F_2} denotes
%%        the union of figures
%%        $F_1$ and $F_2$
%%    \item
%%      \xform{F}{\tau} denotes
%%        applying transformation $\tau$
%%        to figure $F$
%%    \item
%%      \rep{F}{0}{\tau} denotes
%%        the empty figure
%%    \item
%%      \rep{F}{n+1}{\tau} denotes
%%        \xform{\comb{F}{\rep{F}{n}{\tau}}}{\tau}
%%  \end{itemize}

A transformation $\tau$ is a 4-tuple
  $\xargs{s}{\theta}{t_x}{t_y}$ denoting
  uniformly scaling by a factor of $s$,
  rotating by $\theta$ radians, and
  translating by $(t_x,\ t_y)$ in the
  $x$ and $y$ directions respectively.
% Using these,
%   \autoref{fig:polygons} features several hexagons,
%   scaled by various values of $s$,
%   implemented as:
For example, the green hexagon in \autoref{fig:polygons} is implemented as:
%
% \vspace{-1em}
\[
  \xform{
    \rep{
      \xform{
        \T{line}
      }{
        \xargs{1}{0}{-0.5}{0.5 / \tan(\pi / 6)}
      }
    }{6}{
      \xargs{1}{2\pi/6}{0}{0}
    }
  }{
    \xargs{2}{0}{0}{0}
  }
\]
That is, a hexagon side $\xform{\T{line}}{\xargs{1}{0}{-0.5}{0.5 / \tan(\pi / 6)}}$
is repeated six times, each time rotated by another $2\pi/6$ radians,
and the resulting unit hexagon is scaled by $2$.

Taking a closer look at
  the two blue hexagons in \autoref{fig:polygons},
  however, we notice something peculiar.
The two occurrences of 
\xform{\T{repeat} \ldots}{\xargs{1}{0}{0}{0}}
  % \scale{\T{repeat}(\ldots)}{1}
  in \autoref{fig:polygons} are no-ops:
  % \footnote{
  %   We use the shortened forms like ``\T{scale}''
  %   in the text merely for concision;
  %   they are fully expanded in both the
  %   corpus from~\cite{cogsci-dataset} and
  %   in \autoref{fig:polygons}.
  % }:
  they merely scale a figure $F$ by a factor of 1
  and neither rotate nor translate it.
These redundant transformations would likely not be there
had the code been written by hand 
or decompiled from a lower-level representation (by a tool like \tname{Szalinski}~\cite{szalinski});%
  \footnote{
    We suspect that these transformations
    ended up in the dataset of~\citet{cogsci-dataset}
    because it was generated programmatically from
    human-designed abstractions, such as ``scaled polygon''.}
and yet, they are crucial
  if we hope to learn the optimal abstraction $f_0$
  with a purely syntactic technique.

\autoref{fig:llmt} shows a simplified and ``more natural''
  version of the corpus from \autoref{fig:polygons},
  which eliminates these no-op transformations.
As illustrated in the figure,
  this causes a purely syntactic technique
  to learn a different function, $f_1$,
  which abstracts over an \emph{unscaled} unit polygon.
Because
  the scaling transformation
  is now outside the abstraction,
  it must be repeated five times.
As a result,
  although the simplified input corpus C from \autoref{fig:llmt}
  is \emph{smaller} than the original corpus A (196 AST nodes instead of 208),
  its compressed version D is \emph{larger} (81 AST nodes instead of 72)!
In other words,
  syntactic library learning is not robust
  with respect to semantics-preserving code variations.
  % This seems counterintuitive;
%   it may be difficult to explain to a
%   user why sprinkling no-ops throughout
%   code would improve learned library quality.

\begin{figure}
  \centering
  \includegraphics[width=\textwidth]{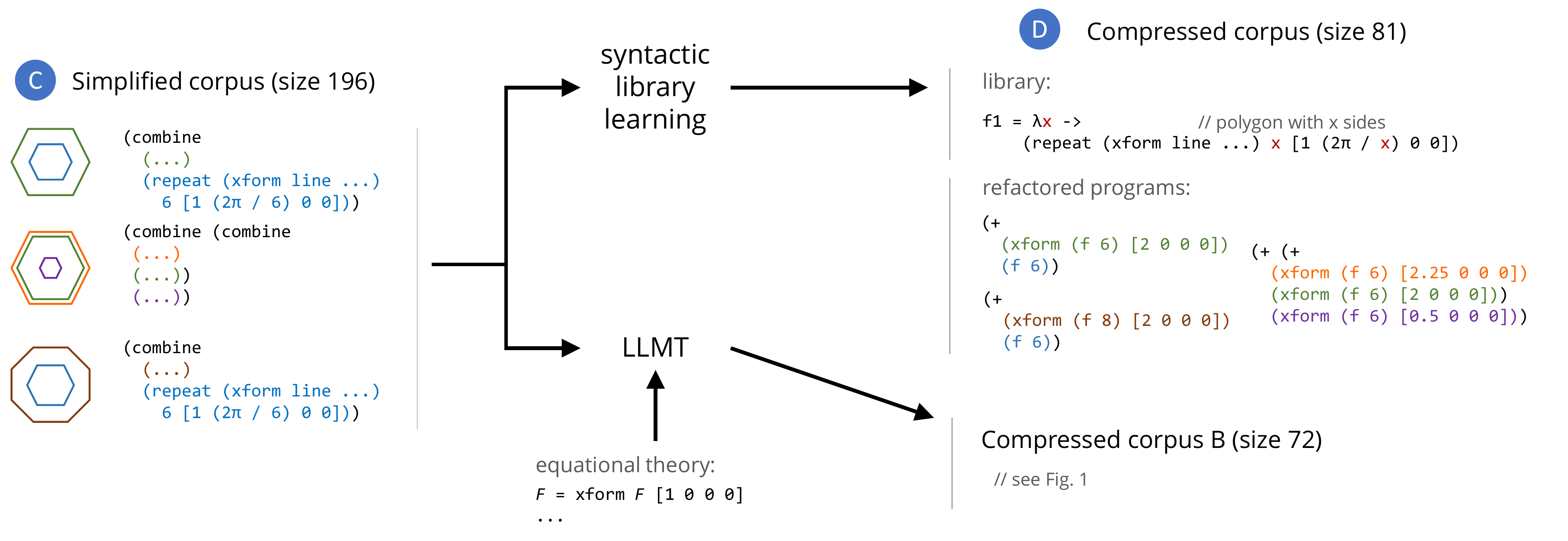}
  \caption{Difference between syntactic library learning and LLMT.
           Here the initial corpus C is the simplified version of corpus A from \autoref{fig:polygons},
           with the redundant transformations on the blue subterms removed
           (unchanged terms are shown as ellipses).
           With this modification, syntactic techniques would learn an inferior abstraction $f_1$,
           leading to corpus D,
           while LLMT still learns the better abstraction $f_0$.
   }\label{fig:llmt}
\end{figure}

In contrast, our tool \tool
  can take the simplified corpus C as input
  and still discover, in a fraction of a second,
  the scaled polygon abstraction $f_0$,
  yielding the compressed corpus B of size 72.
In the rest of this section,
  we illustrate how \tool achieves this using our new algorithm,
  \emph{library learning modulo equational theory} (LLMT).

\mypara{Simplified DSL}
In the rest of this section 
  we use a tailored version of the \cogsci DSL
  with the following additional constructs:

  \vspace{0.25em}
  \begin{tabular}{rl}
    \scale{F}{s} &
      scale $F$ by a factor of $s$
      \\
    \repRot{F}{n}{\theta} &
      a special case of \T{repeat} that only performs rotation between iterations
    \\
    \pSide{n} &
      a side of a regular unit $n$-gon
  \end{tabular}
  \vspace{0.25em}

\noindent
These are expressible in the original DSL,
  and could be even discovered with library learning,
  given an appropriate corpus;
  we treat them as primitives here
  for the sake of simplifying presentation.

\subsection{Representing Equivalent Terms with E-Graphs}

To exploit equivalences during library learning,
%  (like a human would do),
  \tool takes as input a domain-specific equational theory.
%To be able to reason about semantics-preserving transformations like the human programmers do,
%\tool takes as input a domain-specific equational theory.
%
For our running example,
  let us assume that the theory contains a single equation:
\begin{equation}
  F \semeq \scale{F}{1} \label{eq:unit-xform}
\end{equation}
  which stipulates that any figure $F$
  is equivalent to itself 
  % transformed
  % by the no-op scaling $\T{xform}$ as described above.
  scaled by one.
With this equation in hand,
  it is possible to ``rewrite'' corpus C into corpus A,
  and from there learn the optimal compressed corpus B
  by purely syntactic techniques.
The challenge
  is that there are infinitely many
  alternative corpora C may be rewritten to;
  how do we know which to pick to maximize syntactic alignment,
  and thus the chance of discovering an optimal abstraction?

Instead of trying to guess
  the best equivalent corpus
  or enumerating them one by one,
  \tool uses \emph{equality saturation}~\cite{tate2009equality,willsey2021egg}.
Equality saturation takes as input
  a term $t$ and a set of rewrite rules,
  and finds the space of all terms equivalent
  to $t$ under the given rules;
  this is possible due to the high degree of
  sharing provided by the \emph{e-graph} data structure,
  which can compactly represent the resulting space.

\autoref{fig:egraph} (left) shows the e-graph
  built by equality saturation for the blue term in \autoref{fig:llmt}---%
  represented in the simplified DSL as \repRot{\pSide{6}}{6}{2\pi/6}---%
  using the rewrite rule \eqref{eq:unit-xform}.
The blue part of the graph represents the original term,
  and the gray part is added by equality saturation.
The solid rectangles in the e-graph denote
  \emph{e-nodes} (which are similar to regular AST nodes),
  while the dashed rectangles denote
  \emph{e-classes} (which represent equivalence classes of terms).
Importantly,
  the edges in the e-graph go from e-nodes to e-classes,
  which enables compact representation of
  programs with variation in sub-terms:
  for example, making e-class $c_1$
  the first child of the \T{repRot} node,
  enables it to represent both terms
    \repRot{\pSide{6}}{6}{2\pi/6}
%    \rep{{\color{orange}\T{line}}}{6}{\xargs{1}{2\pi/6}{0}{0}}
  and
    \repRot{\scale{\pSide{6}}{1}}{6}{2\pi/6}
%    \rep{{\color{orange}\scale{\T{line}}{1}}}{6}{\xargs{1}{2\pi/6}{0}{0}}
  without duplicating their common parts.
Furthermore, because e-graphs can have cycles,
  they can also represent infinite sets of terms:
  for example, the e-class $c_1$ represents
  all terms of the form:
  \pSide{6},
  \scale{\pSide{6}}{1},
  \scale{\scale{\pSide{6}}{1}}{1},
  etc.
Because this e-graph represents
  the space of \emph{all} equivalent terms
  up to the rewrite~\eqref{eq:unit-xform},
  the term we seek for library learning, namely
    \scale{\repRot{\pSide{6}}{6}{2\pi/6}}{1},
    %\scale{\repRot{\pSide{6}}{6}{2\pi/6}}{1},
    %\scale{\rep{\T{line}}{6}{\xargs{1}{2\pi/6}{0}{0}}}{1}
  is also represented in the e-class $c_0$.

  \begin{figure}
    \begin{minipage}{.4\textwidth}
      \centering
      \includegraphics[width=.95\textwidth]{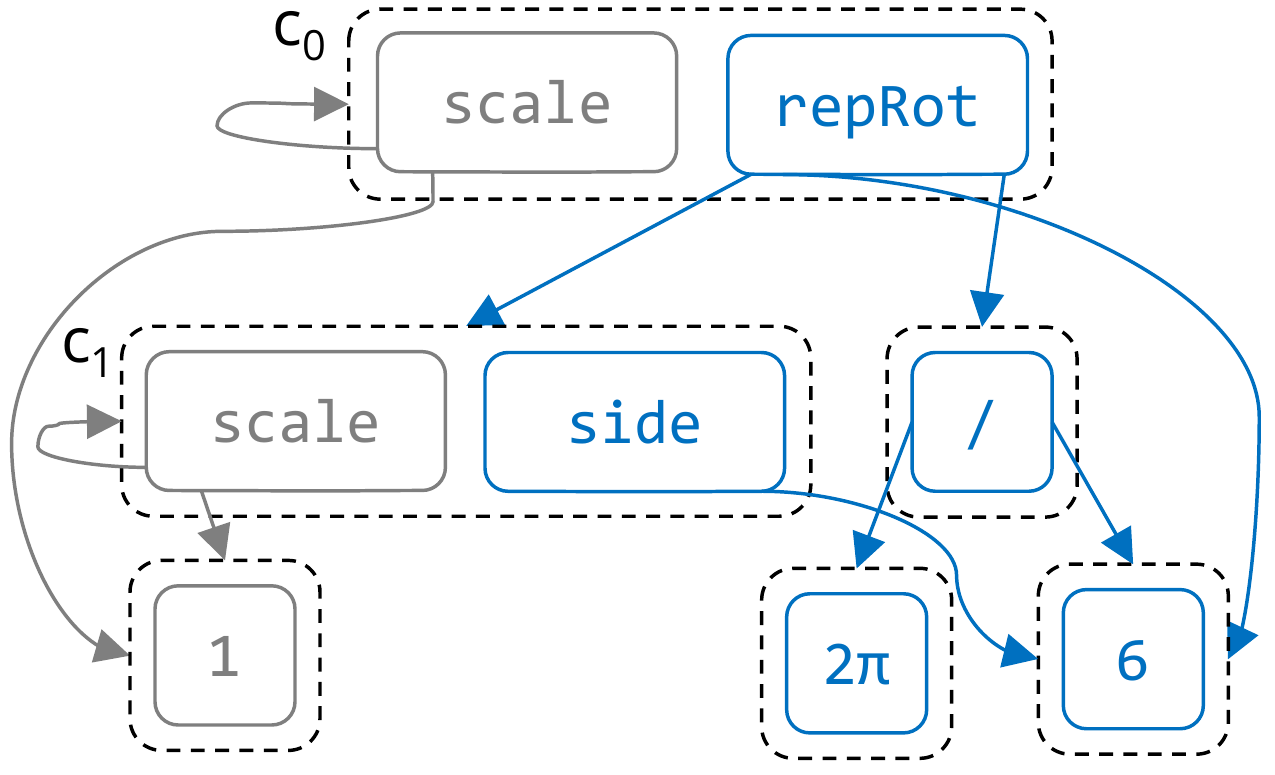}
    \end{minipage}\vline%
    \begin{minipage}{.6\textwidth}
      \centering
      \includegraphics[width=.95\textwidth]{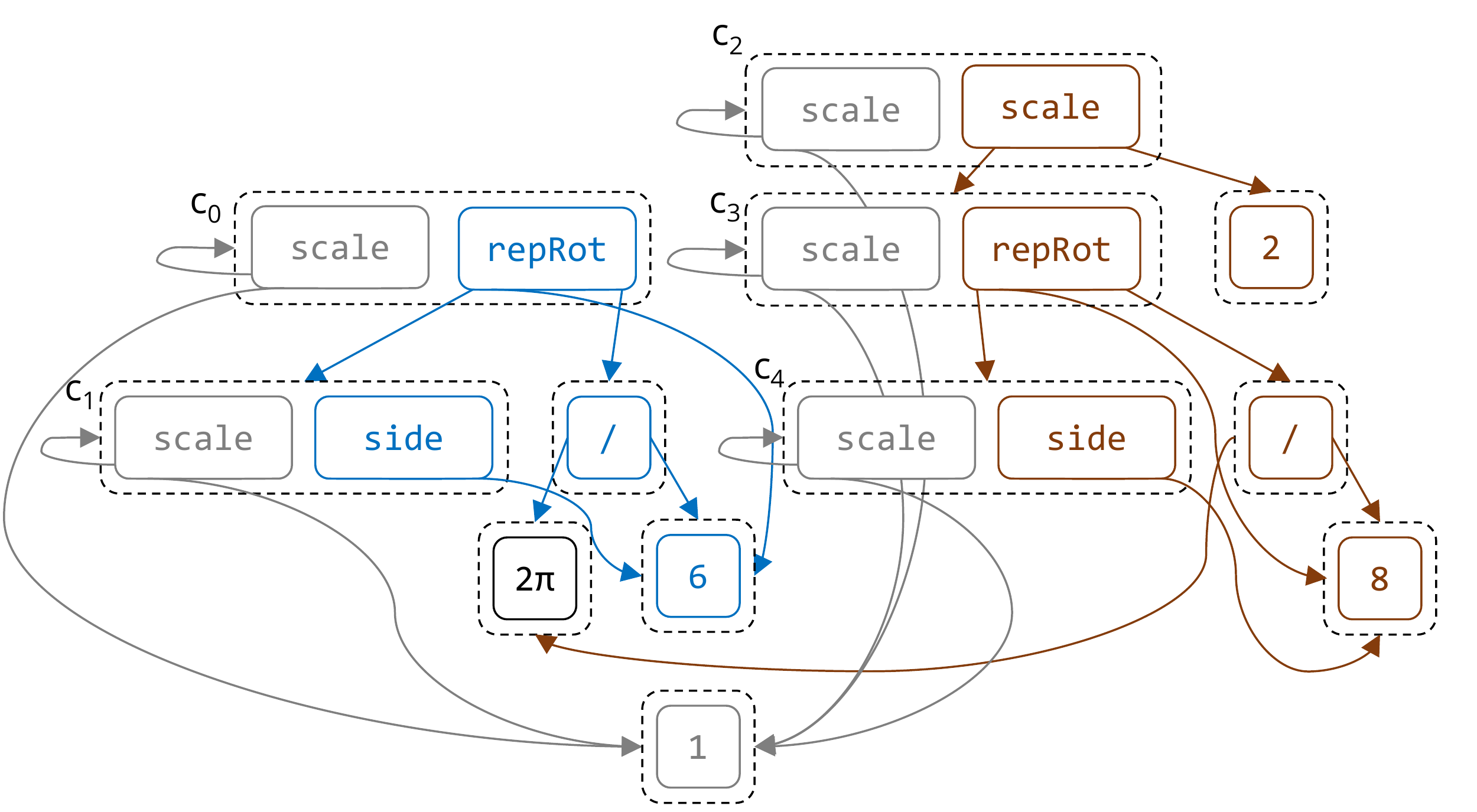}
    \end{minipage}
    \caption{(Left) An e-graph representing the space of all programs
    equivalent the blue term in \autoref{fig:llmt} under rule \eqref{eq:unit-xform}.
    (Right) An e-graph with both the blue and the brown terms from \autoref{fig:llmt}
    after equality saturation.
    All terms are written in the simplified DSL.
    }\label{fig:egraph}
  \end{figure}  

\subsection{Candidate Generation via E-Graph Anti-Unification}

After building an e-graph from the given corpus by
  running equality saturation with the given equational theory,
  the next step in library learning is
  to generate candidate patterns
  that capture syntactic similarities across the corpus.
%\TODO{Idk if ``saturated corpus'' is a good name, because we don't necessaily reach saturation,
%but we need a name for the result of running eqsat on the original corpus.}
%
The challenge is to generate
  sufficiently few candidates
  to make library learning tractable,
  but sufficiently many to achieve good compression.
  % without missing viable candidates
  % that provide optimal compression.
We illustrate candidate generation
  using the e-graph in \autoref{fig:egraph} (right),
  which represents the part of our corpus
  consisting of the (saturated) blue and brown terms.
Recall that the optimal
  pattern---which corresponds to the scaled polygon abstraction $f_0$---is:
\begin{equation}
  P_0 = \scale{\repRot{\pSide{X}}{X}{2\pi/X}}{Y}\label{eq:p0}
  %P_0 = \T{xform (repeat line $\ X\ $ (2pi/$X$)) $\ Y$}\label{eq:p0}
\end{equation}

Prior work on \dc generates patterns by
  picking a random fragment from the corpus,
  and then replacing arbitrarily chosen subterms
  with pattern variables.
For example, to generate the pattern $P_0$,
  \dc needs to pick the entire brown subterm as a fragment,
  and then decide to abstract over its subterms $8$ and $2$.
This approach successfully restricts
  the set of candidates from
  all syntactically valid patterns
  to only those that have
  at least one match in the corpus;
  however, since there are too many ways
  to select the subterms to abstract over,
  this space is still too large to explore exhaustively
  beyond small corpora of short programs.
  % it still produces too many patterns and
  % does not scale beyond small corpora of short programs.
In \tool, this problem
  is exacerbated by equality saturation,
  since the e-graph often contains exponentially or
  infinitely more programs than the original corpus.
%since the saturated corpus often contains exponentially or infinitely more programs than the original corpus.

\mypara{Generality Filters}
%\mypara{Inviable Patterns}
%
To further prune the set of viable candidates in \tool,
  we identify two classes of patterns that can be safely discarded,
  either because they are
  too concrete or too abstract.
First,
  a pattern like
    \repRot{\pSide{8}}{X}{2\pi/8}
    %\T{repeat line $\ X\ $ (2pi/8)}
  can be discarded because it is
  \emph{too concrete} for this corpus:
  the corresponding abstraction
  can be applied only once,
  essentially replacing the sole matching term,
    \repRot{\pSide{8}}{8}{2\pi/8}
    %\T{repeat line 8 (2pi/8)},
  with
    $(\lambda x \bnd \rep{\pSide{8}}{x}{2\pi/8})\ 8$,
   %\T{let f = \\x->repeat line x (2pi/8) in f 8},
  which only adds more AST nodes to the corpus. % we only count vars and apps, but not lambdas, but I don't want to say it here
  %only makes the situation worse
  % by \emph{adding} three AST nodes to the corpus:
    % a $\lambda$,
    % a \T{let}-binding, and
    % an application.
%
% \TODO{There is a subtlety here: we can't say this pattern only matches a single subterm in this corpus,
% because e-class $c_2$ represents infinitely many terms, and each has a subterm that matches this pattern.
% But these terms cannot occur in the corpus simultaneously, they are alterntives, hence this pattern is still out.
% This is what our co-occurrence analysis does.
% But I think we'll talk about this in the technical section.}
%
Second,
  a pattern like
    \repRot{\pSide{X}}{X}{Y}
    %\T{repeat line $\ X\ Y$}
  can be discarded because
  it is \emph{too abstract} for this corpus:
  everywhere it matches,
  a more concrete pattern
    \repRot{\pSide{X}}{X}{2\pi/X}
    %\T{repeat line $\ X\ $ (2pi/$X$)}
  would also match;
  using the more concrete pattern
  leads to better compression,
  since the actual arguments in its applications
  are both fewer and smaller:
    \T{f 6 2pi/6} and \T{f 8 2pi/8} vs.
    \T{f 6} and \T{f 8}.

Thus,
  our \emphbf{first key insight} is 
  to restrict the set of candidates
  to the most concrete patterns that match some pair
  % that
  % any optimal candidate pattern must be the most concrete pattern
  % that matches some tuple\footnote{
  %   Note that just matching \emph{pairs} of subterms is insufficient.
  %   \TODO{ADD EXAMPLE}
  % }
  of subterms in the (saturated) corpus.%
  \footnote{As we discuss in \autoref{sec:au} this can in theory eliminate optimal patterns,
  but our evaluation shows that it works well in practice.}
%\TODO{THIS IS NOT TRUE! This should be a tuple of subterms, not pair,
%but \tool actually only does pairs; see comments about subsumption lattice in Slack; what to do???}
%
For example, pattern $P_0$
  is the most concrete pattern that matches the two terms
  \begin{gather}
    \scale{\repRot{\pSide{6}}{6}{2\pi/6}}{1}\label{eq:hex}\\
    \scale{\repRot{\pSide{8}}{8}{2\pi/8}}{2}\label{eq:oct}
    %\T{xform (repeat line 6 (2pi/6)) 1}\label{eq:hex}\\
    %\T{xform (repeat line 8 (2pi/8)) 2}\label{eq:oct}
  \end{gather}
  represented in \autoref{fig:egraph} by the e-classes $c_0$ and $c_2$, respectively.
%
% You might think this observation is too strong:
% if we consider the most concrete pattern that matches green and blue, we would discover a hexagon;
% but this is okay, because by matching every pair of subterms we will also match brown and blue,
% and discover a generic polygon.

\mypara{Term Anti-Unification}
Computing the most concrete pattern that
  matches two given terms is known as
  \emph{anti-unification} (AU)~\cite{Reynolds1969TransformationalSA,plotkin1970lattice}.
AU works by
  a simple top-down traversal of the two terms,
  replacing any mismatched constructors by pattern variables.
For example,
  to anti-unify the terms \eqref{eq:hex} and \eqref{eq:oct},
  we start from the root of both terms;
  since both AST nodes share the same constructor \T{scale}, %\T{xform},
  it becomes part of the pattern and we recurse into the children.
We eventually encounter a mismatch, where
  the term on the left is \T{6} and
  the term on the right is \T{8};
  so we create a fresh pattern variable $X$ and
  remember the \emph{anti-substitution}
  $\sigma = \{(\T{6}, \T{8}) \mapsto X\}$.
When we encounter a mismatch
  in the denominator of the angle,
  we look up the pair of mismatched terms
  $(\T{6}, \T{8})$ in $\sigma$;
  because we already created a variable for this pair,
  we simply return the existing variable $X$.
The final mismatch is $(\T{1}, \T{2})$
  in the second child of \T{scale}; %\T{xform};
  since this pair is not yet in $\sigma$,
  we create a second pattern variable, $Y$.
At this point,
  the resulting anti-unifier is the pattern $P_0$~\eqref{eq:p0}.

Anti-unifying
  a single pair of terms is simple and efficient.
However,
  in LLMT we want to anti-unify \emph{all pairs of subterms}
  that can occur \emph{in any corpus} equivalent
  (modulo the given equational theory)
  to the original\footnote{
      A careful reader might be wondering if
        we need to compute infinitely many anti-unifiers
        because there might be infinitely many equivalent corpora.
      As we explain shortly,
        there are only finitely many patterns that
        are viable abstraction candidates.
    }.
Explicitly enumerating all equivalent corpora
  represented by the e-graph and then
  performing AU on each pair of subterms
  is infeasible.
Instead, \tool performs
  anti-unification directly on the e-graph.

\mypara{From Terms to E-Classes}
We first explain how to
  anti-unify two e-classes.
This operation takes as input a pair of e-classes and
  returns a set of \emph{dominant anti-unifiers},
  \ie a set of patterns that
    (1) match both e-classes, and
    (2) is guaranteed to include
      the best abstraction 
      among the most concrete patterns 
      that match pairs of terms represented by the two e-classes.

% Consider computing $\au(c_0, c_2)$
%   for the e-classes $c_0$ and $c_2$ in
%   the e-graph from \autoref{fig:egraph} (right).
% AU still proceeds as a top-down traversal,
%   but in this context we must
%   check whether two e-classes have
%   any constructors \emph{in common}.
% In this case they do:
%   the \T{scale} constructor occurs
%   once in $c_0$ and twice in $c_2$.
% Let us first pick the brown \T{scale} e-node from $c_2$
%   (and pair it with the only \T{scale} e-node from $c_0$).
% Recursing into the first child,
%   picking the two \T{repRot} nodes,
%   and recursing again, arrive at $\au(c_1, c_4)$
%   we now need to compute $\au(c_0, c_3)$.
% For these two e-classes,
%   let us first consider the two \T{repRot} nodes;
%   again, we recurse and arrive at $\au(c_1, c_4)$.
% Since these two e-classes have
%   a single pair of e-nodes in common (the \T{repeat} nodes),
%   we again recurse and arrive at $\au(c_1, c_1)$.

Consider computing $\au(c_1, c_4)$
  for the e-classes $c_1$ and $c_4$ in
  the e-graph from \autoref{fig:egraph} (right).
AU still proceeds as a top-down traversal,
  but in this context we must
  check whether two e-classes have
  any constructors \emph{in common}.
In this case they do:
  both a \T{side} constructor and a \T{scale} constructor.
Let us first pick the two \T{side} constructors
  and recurse into their only child,
  computing $\au(\{\T{6}\}, \{\T{8}\})$.
These two e-classes have
  no matching constructors,
  so AU simply returns a pattern variable,
  similarly to term anti-unification;
  this yields the first pattern for $c_1$ and $c_4$: $\T{side}\ X$.
  
Recall, however, that $c_1$ and $c_4$ also 
  have matching \T{scale} constructors.
This is where things get interesting:
  these constructors are involved in a \emph{cycle} 
  (their first child is the parent e-class itself).
If we let AU follow this cycle, the set of anti-unifiers becomes infinite:
  $$
    \au(c_1, c_4) = \{\T{side}\ X\} \cup \{\T{scale}\ p\ 1 \mid p \in \au(c_1, c_4)\}
  $$
Fortunately,
  we can show that $\T{side}\ X$ \emph{dominates}
  all the other patterns from this set,
  because their pattern variables---in this case, just $X$---match
  the same e-classes, but they are also larger
  (in \autoref{sec:egraphs} we show how
  this domination relation lets us prune other patterns,
  not just those caused by cycles).
Hence $\au(c_1, c_4)$ simply returns $\{\T{side}\ X\}$.
% \TODO{In fact it also returns $X$ because they have mismatched constructors as well,
% and we explain this in sec 4. Should we say this here?}

% When AU encounters two e-classes with
%   no matching constructors,
%   it introduces a pattern variable,
%   similarly to term anti-unification
%   (again, keeping tack of the anti-substitution to
%   avoid introducing redundant variables).
Following the same logic for the root e-classes of the two polygons, $c_0$ and $c_2$,
  $\au(c_0, c_2)$ yields that pattern $P_0$~\eqref{eq:p0},
  which is required to learn the optimal abstraction.
% Recall, however,
%   that $c_2$ has a second \T{xform} e-node,
%   which we also need to try.
% This, however, leads to a cycle:
%   $$
%     \au(c_0, c_2) = \{P_0\} \cup \{\T{xform}\ p\ 1 \mid p \in \au(c_0, c_2)\}
%   $$
% The cyclic solution is discarded as non-dominant,
%   for the same reason as explained above.

% \TODO{This generalized to $n$ e-classes in the ``natural way'' ?!}

\mypara{From E-Classes to E-Graphs}
To obtain the set of all candidate abstractions,
we need to perform anti-unification over all pairs of e-classes in the e-graph.
Clearly, these computations have overlapping subproblems
(for example, we have to compute $\au(c_1, c_4)$ 
as part of $\au(c_0, c_2)$ and  $\au(c_0, c_3)$).
To avoid duplicating work, \tool uses an efficient dynamic programming algorithm
that processes all pairs of e-classes in a bottom-up fashion.

\subsection{Candidate Selection via Targeted Common Subexpression Elimination}

We now illustrate the final step of library learning in \tool:
given the set of candidate abstractions generated by e-graph anti-unification,
the goal is to pick a subset that gives the best compression for the corpus as a whole.
%
% This is a non-trivial combinatorial search problem.
%
For example, the candidates generated for the corpus from \autoref{fig:llmt} include:
\[ \begin{array}{rcll}
  f_0 &=& \lambda x\ y \bnd \scale{\repRot{\pSide{x}}{x}{2\pi/x}}{y} &\quad \text{scaled regular $x$-gon} \\[2pt]
  f_1 &=& \lambda x \bnd \repRot{\pSide{x}}{x}{2\pi/x} &\quad \text{regular $x$-gon} \\[2pt]
  f_2 &=& \lambda y \bnd \scale{\repRot{\pSide{6}}{6}{2\pi/6}}{y} &\quad \text{scaled hexagon}
%
%f_0 = \T{\\x y->trans (repeat line x (2pi/x)) y} &\quad \text{scaled polygon}\\
%f_1 = \T{\\x->repeat line x (2pi/x)} &\quad \text{polygon}\\
%f_2 = \T{\\y->trans (repeat line 6 (2pi/6)) y} &\quad \text{scaled hexagon}
\end{array} \]
It is not immediately clear which abstraction is better:
  $f_0$ matches more terms than $f_2$,
  but $f_2$ requires fewer arguments
  (so if we have enough scaled hexagons in the corpus and
  only one octagon, it might be better to leave the
  octagon un-abstracted).
On the other hand,
  $f_1$ might be better,
  since it does not introduce the
  redundant transformation on the blue hexagons.
Finally,
  if we pick $f_0$ \emph{and} $f_2$ together,
  we can also abstract the definition of
  $f_2$ as $f_0\ 6$, thereby getting additional reuse!
As you can see,
  candidate selection is highly non-trivial,
  since it needs to take into account the
  choice between different equivalent programs in the e-graph, %saturated corpus,
  and the fact that some abstractions can be defined using others.

\begin{figure}
  \centering
  \includegraphics[width=.8\textwidth]{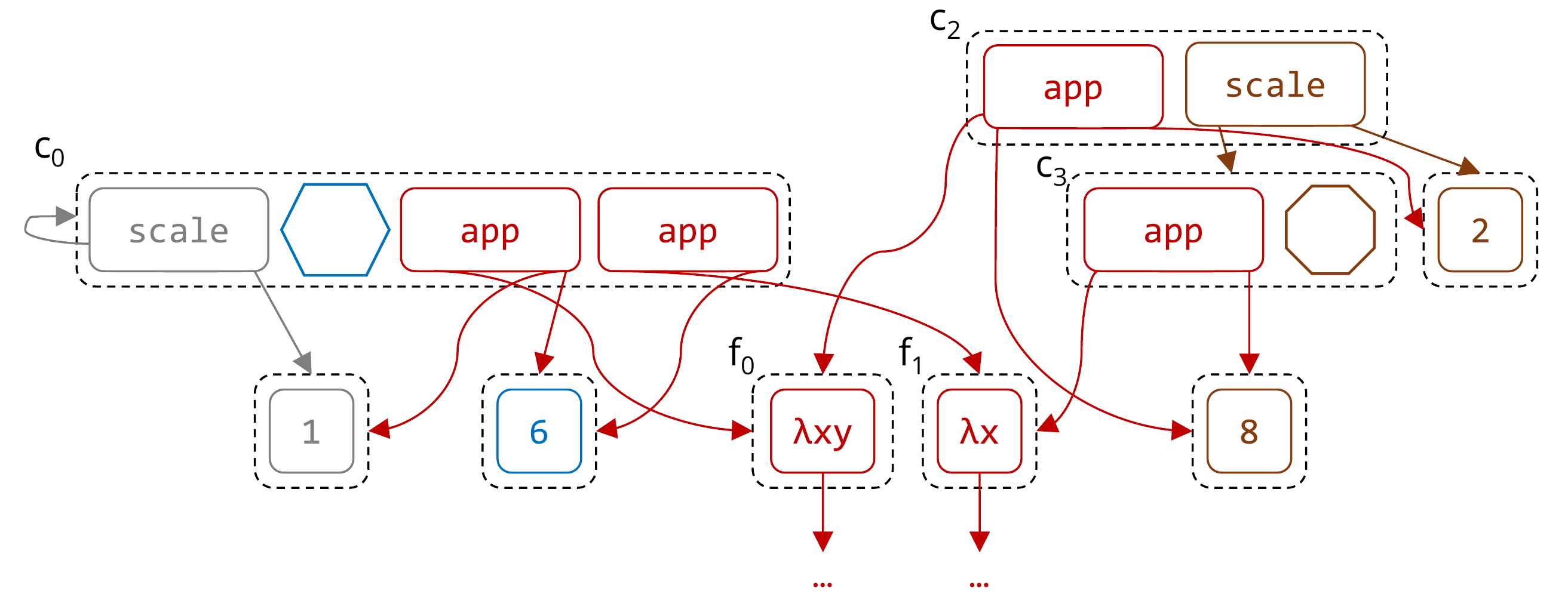}
  \caption{The e-graph from \autoref{fig:egraph} (right) with applications of $f_0$ and $f_1$ depicted in red.
           We show the unchanged parts of the graph representing the unit hexagon and octagon as corresponding shapes;
           we also omit some of the gray e-nodes added in the previous stage.}\label{fig:beam}
\end{figure}

\mypara{Reduction to E-Graph Extraction}
To overcome this difficulty, we once again leverage e-graph and equality saturation.
Our \emphbf{second key insight} is that selecting the optimal subset of abstractions
can be reduced to the problem of extracting the smallest term from an e-graph
in the presence of common sub-expressions.

To illustrate this reduction,
% let us again consider the part of the corpus depicted on \autoref{fig:egraph} (right),
% that is the saturated version of the blue and the brown polygons, and
let us limit our attention to only two candidate abstractions, $f_0$ and $f_1$, defined above.
\tool converts each of the candidate patterns and its corresponding abstraction into a rewrite rule
that introduces a \emph{local} $\lambda$-abstraction followed by application into the corpus;
for our two abstractions these rules are:
\begin{align}
  % \scale{\repRot{\pSide{X}}{X}{2\pi/X}}{Y} &\Rightarrow (\lambda x\ y \bnd \scale{\repRot{\pSide{x}}{x}{2\pi/x}}{y})\ X\ Y\\ 
  % \repRot{\pSide{X}}{X}{2\pi/X} &\Rightarrow (\lambda x \bnd \repRot{\pSide{x}}{x}{2\pi/x})\ X 
  \scale{\repRot{\pSide{X}}{X}{2\pi/X}}{Y} &\Rightarrow f_0\ X\ Y\label{eq:rw1}\\ 
  \repRot{\pSide{X}}{X}{2\pi/X} &\Rightarrow f_1\ X \label{eq:rw2}
\end{align}
The result of applying these rules to the e-graph from \autoref{fig:egraph} (right)
is shown in \autoref{fig:beam}.
For example, you can see that the e-class $c_2$ (which represents the scale-2 octagon)
now stores an alternative representation:
% \T{let f0 = \\x y->trans (repeat line x (2pi/x)) y in f0 2 8}.
% $\T{let}\ f_0 = \lambda x\ y \bnd \scale{\repRot{\pSide{x}}{x}{2\pi/x}}{y}\ \T{in}\ f_0\ 2\ 8$.
% $(\lambda x\ y \bnd \scale{\repRot{\pSide{x}}{x}{2\pi/x}}{y})\  2\ 8$.
$f_0\  8\ 2$.
The e-class $c_0$ (unit hexagon)
has representations using either $f_0$ or $f_1$,
because this class matches both rewrite rules \eqref{eq:rw1} and \eqref{eq:rw2} above. % pattern above.
Note also that because the definitions of the $\lambda$-abstractions are also stored in the e-graph,
equality saturation can use the above rewrite rules inside their bodies,
to make one function use another:
for example, one term stored for the definition of $f_0$ is
$$
\lambda x\ y \bnd \scale{f_1\ x}{y}.
$$

Once we have built the version of the e-graph with local lambdas for all the candidate abstractions,
all that is left is to extract the smallest term out of this e-graph.
The tricky part is that we want to count the size of the duplicated lambdas only once.
For example, in \autoref{fig:beam}, $f_0$ is applied twice (in $c_0$ and $c_2$);
if term extraction were to choose both of these e-nodes,
we want to treat their first child (the definition of $f_0$) as a common sub-expression,
whose size contributes to the final expression only once.
Intuitively, this is because we can ``float'' these lambdas 
into top-level \T{let}-bindings,
thereby defining $f_0$ only once, and replacing each local lambda with a name.

Extraction with common sub-expressions is a known but notoriously hard problem,
which is traditionally reduced to integer linear programming (ILP)~\cite{tensat,spores}.
Because the scalability of the ILP approach is very limited,
we have developed a custom extraction algorithm,
which scales better by using domain-specific knowledge and approximation.

\mypara{Extraction Algorithm}
The main idea for making extraction more efficient is that
for library learning we are only interested in sharing a certain type of sub-terms:
namely, the $\lambda$-abstractions.
Hence for each e-class we only need to keep track
of the the best term for each possible library
(\ie each subset of $\lambda$-abstractions).
More precisely, for each e-class and library,
we keep track of
  (1) the smallest size of the library
  (2) the smallest size of the program refactored using this library 
  (counting the $\lambda$-abstractions as a single node).
We compute and propagate this information bottom-up through the e-graph.
Once this information is computed for the root e-class
(that represents the entire corpus),
we can simply choose the library with the smallest total size.

For example, for the e-class $c_2$ from \autoref{fig:beam},
with the empty library $\emptyset$, the size of library is 0 and the size of the smallest program (\scale{\repRot{\pSide{8}}{8}{2\pi/8}}{2}) is 9;
with the library $\{f_0\}$, the size of the library is 9 and the size of the smallest program ($f_0\ 2\ 8$) is 3;
while with the library $\{f_1\}$, the size of the library is 7 and the size of the smallest program (\scale{f_1\ 8}{2}) is 4.
% For c_0: $\emptyset$ -> 6, $\{f_0\}$ -> 3, $\{f_1\}$ -> 2.
% For c_2: $\emptyset$ -> 8, $\{f_0\}$ -> 3, $\{f_1\}$ -> 4.
%
Clearly for this e-class introducing library functions is not paying off yet ($0 + 9 < 7 + 4 < 9 + 3$),
since each one can be only used once.
This situation changes, however, as we move up towards the root.
Already for the parent e-class of $c_0$ and $c_2$,
the cost of introducing $\{f_0\}$ and $\{f_1\}$ is amortized:
the size of the smallest program is 17 with the empty library and 7 with either $\{f_0\}$ or $\{f_1\}$,
so $\{f_1\}$ is already worthwhile to introduce ($9 + 7 < 0 + 17$).
Including even more programs with scaled polygons eventually makes the library $\{f_0\}$ the most profitable of the four subsets.

Since in a larger corpus, keeping track of all subsets of candidate abstractions is not feasible,
\tool provides a way to trade off scalability and precision
by using a \emph{beam search} approach.

% \mypara{Approximate Extraction with Beam Search}
% %
% The algorithm described so far is precise,
% but does not scale to large sets of candidate abstractions,
% since it has to track all subsets.
% %
% \tool provides a way to trade off scalability and precision
% by using a \emph{beam search} approach.
% %
% To reduce the set of libraries that needs to be stored in each e-class,
% the user can provide bounds both on
%   (1) the maximum \emph{size} of the library, and
%   (2) the maximum number of distinct best libraries considered for each e-class (\ie the \emph{beam size}).
% %
% \TODO{Should we talk about partial-order reduction here?}

\section{Library Learning over Terms}\label{sec:au}

In this section
  we formalize the problem of
  library learning over a corpus of terms and
  motivate our first core contribution:
  generating candidate abstractions via anti-unification.
\autoref{sec:egraphs} generalizes this formalism to library learning over an e-graph.
For simplicity of exposition,
our formalization of library learning is \emph{first order},
that is, the initial corpus does not itself contain any $\lambda$-abstractions,
and all the learned abstractions are first-order functions
(the \tool implementation does not have this limitation).

% The goal of this section is to formalize the core algorithm of LLMT---%
% e-graph anti-unification---%
% and to show that this algorithm is complete as a candidate generator for library learning:
% that is, it does not discard optimal candidate abstractions.
% %
% We first consider library learning over a term
% and then generalize to an e-graph.
% We build up to the full algorithm in steps:
% starting with abstracting two terms at the top level,
% and then generalizing to all subterms in a program,
% and finally to all e-classes in an e-graph.

\subsection{Preliminaries}

\mypara{Terms}
A \emph{signature} $\Sigma$ is a set of constructors,
each associated with an arity.
$\termsof{\Sigma}$ denotes the set of \emph{terms} over $\Sigma$,
defined as the smallest set containing all $s(t_1,\dots,t_k)$
where $s \in \Sigma$, $k = \arity(s)$, and $t_1,\dots,t_k\in\termsof{\Sigma}$.
We abbreviate nullary terms of the form $s()$ as $s$.
The \emph{size} of a term $\sz(t)$ is defined in the usual way
(as the number of constructors in the term).
We use $\subterms(t)$ to denote the set of all subterms of $t$ (including $t$ itself).
We assume that $\Sigma$ contains a dedicated \emph{variadic} tuple constructor,
written $\langle t_1, \ldots, t_n \rangle$,
which we use to represent a corpus of programs as a single term
($\langle \rangle$ does not contribute to the size of a term).

\mypara{Patterns}
If $\mathcal{X}$ is a denumerable set of variables,
$\patsof{\Sigma}{\mathcal{X}}$ is a set of \emph{patterns},
\ie terms that can contain variables from $\mathcal{X}$.
A pattern is \emph{linear} if each of its variables occurs only once:
$\forall X \in \fv(p) . \occ(X, p) = 1$.
A \emph{substitution} $\sigma = [\subst{X_1}{p_1}, \ldots, \subst{X_n}{p_n}]$
is a mapping from variables to patterns.
%
% We are often interested in \emph{grounding substitutions},
% where all right-hand sides are terms.
% that maps each $X_i$ to $p_i$ and is identity elsewhere.
%
We write $\sapp{\sigma}{p}$ to denote the application of $\sigma$ to pattern $p$,
which is defined in the standard way.
We define the size of a substitution $\sz(\sigma)$ as the total size of its right-hand sides.

A pattern $p$ is more general than (or \emph{matches}) $p'$,
written $p' \sub p$,
if there exists $\sigma$ such that $p' = \sapp{\sigma}{p}$;
we will denote such a $\sigma$ by $\match{p'}{p}$.
For example $X + 1 \sub X + Y$ with $\match{X + 1}{X + Y} = [\subst{Y}{1}]$.
The relation $\sub$ is a partial order on patterns,
and induces an equivalence relation $p_1 \sim p_2 \triangleq p_1 \sub p_2 \wedge p_2 \sub p_1$
(equivalence up to variable renaming).
In the following, we always distinguish patterns only up to this equivalence relation. 

The \emph{join} of two patterns $p_1 \join p_2$---also called their \emph{anti-unifier}---%
is the least general pattern that matches both $p_1$ and $p_2$;
the join is unique up to $\sim$.
Note that $\langle \patsof{\Sigma}{\mathcal{X}}, \sub, \join, \top = X \rangle$ is a join semi-lattice
(part of Plotkin's \emph{subsumption lattice}~\cite{plotkin1970lattice}).
Consequently, the join can be generalized to an arbitrary set of patterns.

A \emph{context} $\ctx$ is a pattern with a single occurrence of a distinguished variable $\circ$.
We write $\ctx[p]$---$p$ in context $\ctx$---as a syntactic sugar for $\sapp{[\subst{\circ}{p}]}{\ctx}$.
A \emph{rewrite rule} $R$ is a pair of patterns, written $p_1 \Rightarrow p_2$.
Applying a rewrite rule $R$ to a term or pattern $p$, written $\sapp{R}{p}$ is defined in the standard way:
$\sapp{R}{p} = \sapp{(\match{p}{p_1})}{p_2}$ if $p \sub p_1$ and undefined otherwise.
A pattern $p$ can be \emph{re-written in one step} into $q$ using a rule $R$,
written $p \rwsteps{R} q$, 
if there exists a context $\ctx$ such that $p = \ctx[p']$ and $q = \ctx[\sapp{R}{p'}]$.
The reflexive-transitive closure of this relation is the \emph{rewrite relation} $\rewrites{\ruleset}$,
where $\ruleset$ is a set of rewrite rules.

\mypara{Compressed Terms}
Compressed terms $\absof{\Sigma}{\mathcal{X}}$ are of the form:
$$
\hat{t} ::= X   \mid   s(\hat{t}_1,\ldots,\hat{t}_k)     \mid (\lambda X_1 \ldots X_n \bnd \hat{t})\ \hat{t}_1\ \ldots\ \hat{t}_n  
$$
In other words,
compressed terms may contain applications of a $\lambda$-abstraction with zero or more binders
to the same number of arguments.
Note that this is a first-order language in the sense that all abstractions are fully applied.
%
  % \footnote{Hereafter, $\many{a}$ stands for a sequence of elements of syntactic class $a$ and $\emptyseq$ denotes the empty sequence.}  
%  compressed terms may have lambdas,
%  which \emph{may} reduce the overall size of the term
%  if the same function is used in multiple locations. 
%
Importantly, we define $\sz(\hat{t})$ in such a way 
that multiple occurrences of a $\lambda$-abstraction are \emph{only counted once}.
% as a result, introducing abstractions may reduce the overall size of the term
% if the same function is used in multiple locations. 
%
For simplicity of accounting, the size of an application $(\lambda \many{X} \bnd \hat{t}_1)\ \many{\hat{t}_2}$
is defined as $\sz(\hat{t}_1) + \sum{\many{\sz(\hat{t}_2)}}$,
that is, abstraction nodes themselves do not add to the size.%
\footnote{Hereafter, $\many{a}$ stands for a sequence of elements of syntactic class $a$ and $\emptyseq$ denotes the empty sequence.}  

% (but this does not make any qualitative difference).
%
% We introduce the standard syntactic sugar $\lambda \many{X} \bnd \hat{t}$%
% for nested abstractions,
% $\hat{t}\ \many{\hat{t}}$ for nested applications,
% and $\T{let}\ X = \hat{t}_1\ \T{in}\ \hat{t}_2$
% for $(\lambda X \bnd \hat{t}_2)\ \hat{t}_1$.

\emph{Beta-reduction} on compressed terms, denoted $\hat{t}_1 \betared \hat{t}_2$, is defined in the usual way:
\[
  (\lambda \many{X} \bnd \hat{t}_1)\ \many{\hat{t}_2}  \betared \sapp{[\many{X \mapsto \hat{t}_2}]}{\hat{t}_1} 
\]
where substitution on compressed terms is the standard capture-avoiding substitution for $\lambda$-calculus.
Note that because our language is first-order and has no built-in recursion,
it is strongly normalizing
(the proof of this statement, as well as other proofs omitted from this section, can be found in the supplementary material).
Hence, the order of $\beta$-reductions is irrelevant,
and without loss of generality we can define \emph{evaluation} on compressed terms, $\hat{t}_1 \evals \hat{t}_2$,
to follow applicative order, \ie innermost $\beta$-redexes are reduced first.
As a result, any application reduced during evaluation
has the form $(\lambda \many{X} \bnd p_1)\ \many{p_2}$,
that is, neither the body $p_1$ nor the actual arguments $\many{p_2}$ contain any redexes
(and hence any $\lambda$-abstractions).
This simplifies several aspects of our formalization;
for example, there is no need for $\alpha$-renaming, 
since with no binders in $p_1$, no variable capture can occur.

\mypara{Problem Statement}
%
% We will say that $\hat{t} \in \absof{\Sigma}{\mathcal{X}}$ 
% is a \emph{compression of} the term $t \in \termsof{\Sigma}$ 
% if $\hat{t} \evals t$ 
% ($\hat{t}$ evaluates to $t$).
%
We can now formalize
  the \emph{library learning problem} as follows:
  given a term $t \in \termsof{\Sigma}$,
  the goal is to find the smallest compressed term $\hat{t} \in \absof{\Sigma}{\mathcal{X}}$
  that evaluates to $t$ (\ie $\hat{t} \evals t$).
The reason such $\hat{t}$ may be smaller than $t$,
  is that it may contain multiple occurrences of the same $\lambda$-abstraction
  (applied to different arguments),
  whose size is only counted once.
An example is shown
  in \autoref{fig:lib-learning-example}.

Although in full generality the solution might include nested lambdas with free variables
(defined in the outer lambdas),
in the rest of the paper we restrict our attention to \emph{global library learning},
where all lambdas are closed terms.
This is motivated by the purpose of library learning
to discover reusable abstraction for a given problem domain.
The solution in \autoref{fig:lib-learning-example} already has this form.

\begin{figure}
  \begin{minipage}{.5\textwidth}
    \begin{align*}
      \langle &f(g(a) + g(a)) + (g(1) + h(2))\\
      ,\ &f(g(b) + g(b)) + (g(3) + h(4))\\
      ,\ &g(5) + h(6) \rangle
    \end{align*}
  \end{minipage}%
  \begin{minipage}{.5\textwidth}
    \begin{align*}
      &\langle \hat{f}_2\ g(a)\ 1\ 2, \hat{f}_2\ g(b)\ 2\ 3, \hat{f}_1\ 5\ 6\rangle\quad \text{where}\\
      &\hat{f}_1 = \lambda Y\ Z \bnd g(Y) + h(Z)\\
      &\hat{f}_2 = \lambda X\ Y\ Z \bnd f(X + X) + \hat{f}_1\ Y\ Z
    \end{align*}
  \end{minipage}
  \caption{Library learning.
           Left: initial term (size 29).
           Right: an optimal solution with two abstractions, one of which uses the other (size 26).
           A solution with $\hat{f}_2 = \lambda X\ Y\ Z \bnd f(g(X) + g(X)) + \hat{f}_1\ Y\ Z$ also has size 26.}\label{fig:lib-learning-example}
\end{figure}

\subsection{Pattern-Based Library Learning}

At a high-level, our approach to library learning
is to use \emph{patterns} that occur in the original corpus
as candidate bodies for $\lambda$-abstractions in the compressed corpus.
Looking at the example in \autoref{fig:lib-learning-example},
it is not immediately obvious that using just the patterns from the original corpus is sufficient,
since the body of $f_2$ contains an application of $f_1$.
Perhaps surprisingly, this is not an issue:
the key idea is that we can compress $t$ into $\hat{t}$ by inverting the evaluation $\hat{t} \evals t$,
and because the evaluation order is applicative,
the rewritten sub-term at every step will not contain any $\beta$-redexes.

\mypara{Compression}
More formally, given a pattern $p$,
let us define its \emph{compression rule} (or $\kappa$-rule for short)
as the rewrite rule
\[
  \kappa(p) \; \triangleq \; p \Rightarrow (\lambda \many{X} \bnd p)\ \many{X} \quad \text{where}\quad \many{X} = \fv(p)
\]
In other words, $\kappa(p)$ will replace any term 
matching $p$ with an application of a function whose body is \emph{exactly} $p$.
For example, if $p = g(Y) + h(Z)$,
its $\kappa$-rule is $g(Y) + h(Z) \Rightarrow (\lambda Y\ Z \bnd g(Y) + h(Z))\ Y\ Z$.
Note that on the right-hand side of this rule 
only the \emph{free} occurrences of $Y$ and $Z$ will be substituted during rewriting;
the bound $Y$ and $Z$ will be left unchanged, following the usual semantics of substitution for $\lambda$-calculus.
For example, this rule can rewrite the third term in \autoref{fig:lib-learning-example} (left) as follows:
\[
  g(5) + h(6) \;\rwsteps{\kappa(g(Y) + h(Z))}\; (\lambda Y\ Z \bnd g(Y) + h(Z))\ 5\ 6
\]
%
% We will refer to a rewrite step such as the one above as a $\kappa$-step.
%
A sequence of $\kappa$-rewrites $t \rwsteps{\kappa(p_1)} \dots \rwsteps{\kappa(p_n)} \hat{t}$,
where all $p_i \in \patset$, 
is called a \emph{compression} of $t$ into $\hat{t}$ using patterns $\patset$
and written $t \rewrites{\kappa(\patset)} \hat{t}$.
We can now show that the library learning problem is equivalent to finding the smallest compression of $t$
using only patterns that occur in $t$.

\begin{theorem}[Soundness and Completeness of Pattern-Based Library Learning]
  For any term $t \in \termsof{\Sigma}$ and compressed term $\hat{t} \in \absof{\Sigma}{\mathcal{X}}$:
  \begin{description}[font=\itshape]
    \item[(Soundness)] If $t$ compresses into $\hat{t}$, then $\hat{t}$ evaluates to $t$: $\forall \patset . t \rewrites{\kappa(\patset)} \hat{t} \Longrightarrow \hat{t} \evals t$.
    \item[(Completeness)] If $\hat{t}$ is a solution to the (global) library learning problem, 
    then $t$ compresses into $\hat{t}$ using only patterns that have a match in $t$:
    $\hat{t} \in \argmin_{\hat{t}' \evals t} \sz(\hat{t}') \Longrightarrow t \rewrites{\kappa(\patset)} \hat{t}$, 
    where $\patset = \{p \in \patsof{\Sigma}{\mathcal{X}} \mid t' \in \subterms(t) , t' \sub p\}$.
  \end{description}
\label{thm:pat-lib-learning}
\end{theorem} 
The proof can be found in the supplementary material. % \autoref{app:proofs}.

\mypara{Example}
Consider once again the library learning problem in \autoref{fig:lib-learning-example}.
Here the set of patterns used to compress the original corpus into the solution on the right is:
\[
  p_1 = g(Y) + h(Z)\quad\quad\quad
  p_2 = f(X + X) + g(Y) + h(Z)
\]
Rewriting the first term of the corpus proceeds in two steps (the redexes of $\kappa$-steps are highlighted):
\begin{multline*}
  \textcolor{darkblue}{f(g(a) + g(a)) + (g(1) + h(2))}    \;\rwsteps{\kappa(p_2)}\\     
  (\lambda X\ Y\ Z \bnd f(X + X) + \textcolor{darkblue}{g(Y) + h(Z)})\ g(a)\ g(a)\ (g(1) + h(2))\;\rwsteps{\kappa(p_1)}\\ 
  (\lambda X\ Y\ Z\ \bnd f(X + X) + (\lambda Y\ Z \bnd g(Y) + h(Z))\ Y\ Z)\ g(a)\ g(a)\ (g(1) + h(2))
\end{multline*}
In other words, we first rewrite the entire term using $p_2$,
and then rewrite inside the body of the introduced abstraction using $p_1$
(note that this order of compression is the inverse of the applicative evaluation order).
The second term of the corpus compresses analogously;
the third term compresses in a single step using $p_1$.
Although this is not obvious from the rewrite sequence above,
the resulting compressed corpus is indeed smaller than the original
thanks to sharing of both $\lambda$-abstractions,
as illustrated on the right of \autoref{fig:lib-learning-example}.

\mypara{Library Learning as Term Rewriting}
Theorem~\ref{thm:pat-lib-learning} reduces library learning to a \emph{term rewriting} problem.
Namely, given a term $t$ and a finite set of rewrite rules $\ruleset = \{\kappa(p) \mid  t' \in \subterms(t) , t' \sub p\}$,
% that consists of $\kappa$-rules for all patterns that have a match in $t$,
our goal is to find a minimal-size term $\hat{t}$ such that $t \rewrites{\ruleset} \hat{t}$,
which is a standard formulation in term rewriting.
Unfortunately, this particular problem is notoriously difficult because
(a) the rule set $\ruleset$ is very large for any non-trivial term $t$, and
(b) our $\sz$ function is non-local (it takes sharing into account)
In the rest of this section we discuss how we can prune the rule set $\ruleset$ 
to reduce it to a tractable size.
\autoref{sec:beam} discusses how we tackle the remaining term rewriting problem
using the \emph{equality saturation} technique~\cite{tate2009equality,willsey2021egg}.

\subsection{Pruning Candidate Patterns}
\label{sec:pruning}

In this section, we discuss which patterns can be discarded from consideration
when constructing the set of $\kappa$-rules $\ruleset$ for the term rewriting problem.

% In this section, 
%  we consider candidate generation
%  starting from the naive but complete
%  approach of generating the infinite set of all patterns $\patsof{\Sigma}{\mathcal{X}}$.
% We show that
%  several classes of patterns are provably suboptimal,
%  eventually reducing the set of viable candidates to a 
%  finite set of patterns that can be enumerated efficiently.

\mypara{Cost of a Pattern}
Consider a compression $t \rewrites{\kappa(\patset)} \hat{t}$
where each pattern $p \in \patset$ is used some number $n$ times, 
with substitutions $\sigma^p_1, \ldots, \sigma^p_n$.
We can break down the total amount of compression into contributions of individual patterns:
\[
  \sz(\hat{t}) - \sz(t) = \sum_{p \in \patset} \mathsf{cost}(p, \{\sigma^p_1, \dots \sigma^p_n\})
\]
The cost of a pattern $p$, in turn, consists of three components.
The cost of \emph{introducing} the abstraction is the size of its body, \ie $\sz(p)$.
The cost of using an abstraction---%
$\textsf{use}(p, \sigma)$ \eqref{eq:use}---%
 includes the application itself and the size of the arguments.
The cost saved by using an abstraction---%
$\textsf{save}(p, \sigma)$ \eqref{eq:save}---%
 is just the cost of the term matched by $p$
 (\ie the redex of the corresponding $\kappa$-step).
%  rewritten in terms of the size of $p$ itself 
%  and the size of subterms in the substitution.
%
\begin{align}
  % \intro(p) &= \sz(p) \label{eq:intro}\\
  \textsf{use}(p, \sigma) &= 1 + \sum_{X \in \fv(p)} \sz(\sigma(X)) \label{eq:use} \\
  \textsf{save}(p, \sigma) &= \sz(\sapp{\sigma}{p}) = 
  \sz(p) + \sum_{X \in \fv(p)} \occ(X, p) \cdot (\sz(\sigma(X)) - 1)\label{eq:save}
\end{align}
The total cost of $p$ is 
 the cost of introducing the abstraction paid a single time,
 plus the cost of each use, minus what you save for each application:
\begin{equation}
  \textsf{cost}(p, \{\sigma_1,\ldots,\sigma_n\}) 
  = \intro(p) + \sum_{\sigma_i} (\textsf{use}(p, \sigma_i) - \textsf{save}(p, \sigma_i))
\end{equation}
When $p$ is linear (all $\occ(X, p) = 1$), 
 the cost depends only on $n$ but not on the substitutions $\sigma_i$:
\begin{align}
  \textsf{cost}(p, \{\sigma_1,\ldots,\sigma_n\}) 
  &= \intro(p) + \sum_{\sigma_i} \left(
    1 - \sz(p) + |\fv(p)|
  \right) \\
  &= \intro(p) + n \cdot \left(
    1 - \sz(p) + |\fv(p)|
  \right)
\end{align}
We can show that a pattern $p$ with a \emph{non-negative} cost can be safely discarded,
that is: there exists another compression using only $\patset \setminus \{p\}$,
whose result is at least as small.
% A pattern is only worth including into the compression if its \textsf{cost} is negative.

\mypara{Trivial Patterns}
Based on this analysis, any linear pattern $p$
 with $\sk(p) \leq 1$ can be discarded,
 where $\sk(p) = \sz(p) - |\fv(p)|$ is the size of $p$'s ``skeleton'', \ie it's body without the variables.
Intuitively, the skeleton of $p$ is simply too small to pay for introducing an application.
In this case, $\mathsf{cost}(p,\_) > 0$ \emph{independently} of how many times it is used.
We refer to such patterns as \emph{trivial}.
Examples of trivial patterns are $X$ and $X + Y$.

% \mypara{Patterns with No Matches}
% %
% Like in prior work~\cite{dreamcoder},
% we can trivially discard patterns that do not match any subterm of $t$
% (where $n=0$ above).
% %
% In this case we only pay for adding a definition,
% but never get to gain from applying it,
% so such a pattern clearly cannot be part of an optimal solution.

\mypara{Patterns with a Single Match}
We can show that patterns with only a single match in the corpus can also be discarded.
If $p$ has a single match in $t$,
then it can appear \emph{at most once} in any compression of $t$.
If $p$ is linear, $\mathsf{cost}(p, \_) = 1 + |\fv(p)|$,
which is always positive, so $p$ can be discarded.
But what about non-linear patterns, where even a single $\kappa$-step can decrease size thanks to variable reuse?
It turns out that any non-linear pattern with a single match
can always be replaced by a nullary pattern (with no variables) that is more optimal.

Without loss of generality, assume that $p$ has a single variable $X$ that occurs $m > 1$ times,
and let its sole $\kappa$-step be $\sapp{\sigma}{p} \rwsteps{\kappa(p)}  (\lambda X \bnd p)\  \sapp{\sigma}{X}$.
The size of the right-hand side is $\sz(p) + 1 + \sz(\sapp{\sigma}{X})$,
or, rewritten in terms of $p$'s skeleton: $1 + (\sk(p) + m) + \sz(\sapp{\sigma}{X})$.
Instead, we can rewrite the same redex $\sapp{\sigma}{p} \rewrites{\{R\}} p'$
using $m$ applications of the rule $R \;\triangleq\; \sapp{\sigma}{X} \Rightarrow (\lambda \epsilon \bnd \sapp{\sigma}{X})\ \epsilon$.
This is a $\kappa$-rule for a \emph{nullary} pattern $\sapp{\sigma}{X}$ with no variables
(hence the corresponding $\lambda$-abstraction has zero binders).
The size of $p'$ obtained in this way is $\sz(\sapp{\sigma}{X}) + m + \sk(p)$
(where the former is the size of the shared $\lambda$-abstraction, $m$ is the number of applications, 
and $\sk(p)$ is the size of the term around the applications).
As you can see, this term is one smaller than the one we get by applying $p$.
Intuitively, this result says that instead of using a non-linear pattern that occurs only once,
it is better to perform common sub-expression elimination.

\mypara{Parameterization Lattice}
Eliminating from consideration all patterns with fewer than two matches in the corpus
suggests an algorithm for generating a complete set $\patset$ of candidate patterns:
  (1) start from the set $\plau{t} = \{t_1 \join t_2 \mid t_i \in \subterms(t)\}$
      of all pairwise joins of subterms of the input program,
  (2) explore all elements of the subsumption semi-lattice above those patterns,
      by gradually replacing sub-patterns with variables,
      until we hit trivial patterns at the top of the lattice.
We will refer to this semi-lattice above $\plau{t}$ as the \emph{parametrization lattice} of $t$, 
denoted $\pl{t}$.
A fragment of $\pl{t}$ for $t = \langle f(a+b), f(a+c), f(b+c)\rangle$ is shown in \autoref{fig:lattice} (left).

% The lattice $\pl{t}$ gives us an algorithmic way of exploring the space of viable patterns:
% (1) start by computing $\plau{t}$ by anti-unifying all pairs of subterms of $t$, then
% (2) traverse the lattice upwards by gradually replacing sub-patterns with variables,
% until we hit trivial patterns at the top of the lattice.

% The discussion above suggests the following definition of the set of candidate patterns $\pl{t}$ for compressing a term $t$:
% $\pl{t}$ must include all pairwise joins of subterms of $t$ and all their minimal generalizations
% (except if any of these patterns is trivial).

% To compute $\pl{t}$, we define \emph{generalizing anti-unification} $\gau(t_1, t_2)$, 
% which takes two terms and computes a set of patterns,
% consisting of their join and all its minimal generalizations.
% %
% This function is formalized in \autoref{fig:gau}.
% %
% The main difference from traditional term anti-unification is highlighted in blue:
% when the root constructor of two terms is the same,
% $\gau$ still includes a a variable (in addition to recursing into subterms);
% this enables $\gau$ to include minimal generalizations of join.
% %
% Note that unlike some traditional presentations of anti-unification,
% we do not keep track of the anti-substitution 
% to avoid abstracting the same pairs of terms into different variables;
% instead we achieve the same effect simply by naming the variables $X_{t_1,t_2}$
% after the pair of terms they abstract.

\begin{figure}
  \begin{minipage}{.5\textwidth}
    \centering
    \includegraphics[width=.8\textwidth]{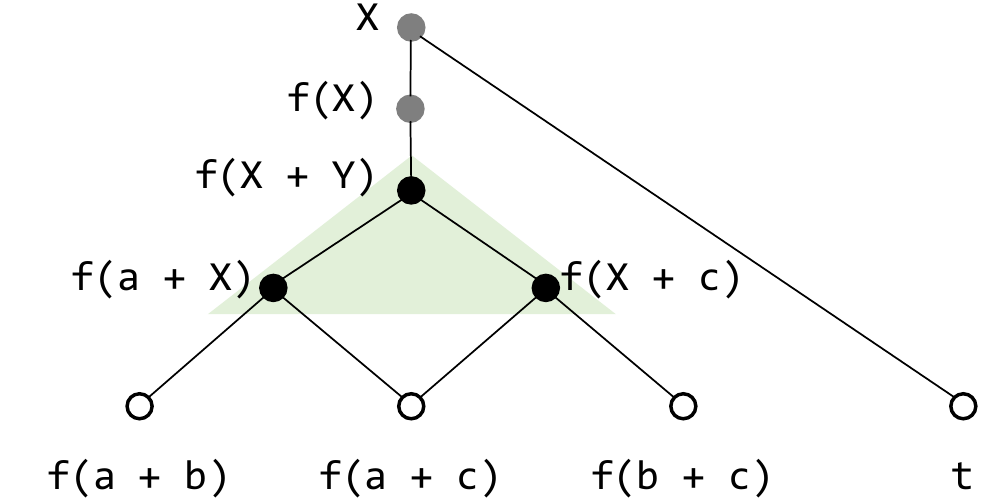}      
  \end{minipage}%
  \begin{minipage}{.5\textwidth}
    \centering
    \includegraphics[width=.8\textwidth]{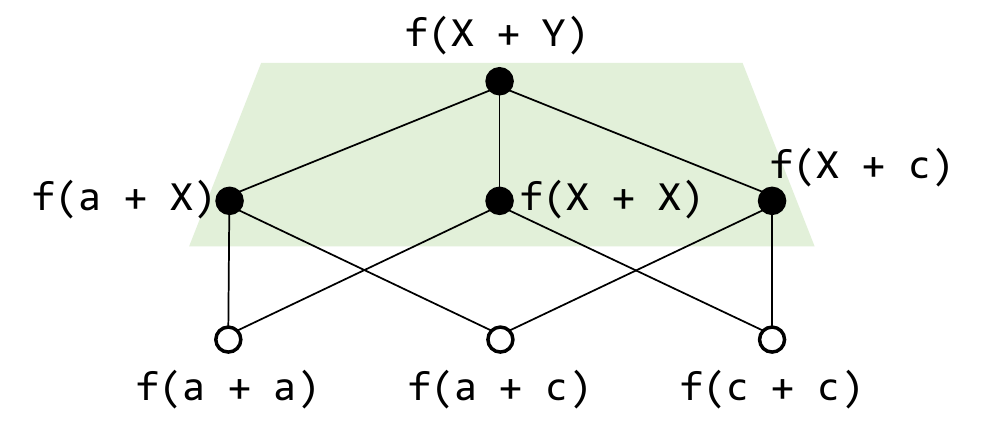}      
  \end{minipage}
  \caption{Left: fragment of the parametrization lattice for the term $t = \langle f(a+b), f(a+c), f(b+c)\rangle$;
           only filled black circles correspond to candidate patterns:
           hollow circles match fewer than two terms and gray circles are trivial.
           Right: an example where is it insufficient to consider pairwise joins
           to obtain an optimal pattern.}
  \label{fig:lattice}
\end{figure}

\mypara{Approximation}
In practice, computing the set $\plau{t}$ is feasible:
although there are quadratically many pairs of subterms,
most of them do not have a common constructor at the root, and hence their join is trivially $X$.
An example is the join of $t$ with any of its subterms in \autoref{fig:lattice} (left).
Unfortunately, generalizing the patterns from $\plau{t}$ 
according to the parameterization lattice (\autoref{fig:lattice}) is expensive.
For this reason, \tool adopts an approximation
and simply uses $\plau{t}$ as the set of candidates.

This approximation makes our pattern generation theoretically incomplete.
Consider a pattern $p \in \pl{t}\setminus \plau{t}$;
%% Nadia: actually, another case if when a more concrete pattern would block a later compression of actual args: is this a real thing?
there are two reasons why we might need $p$ in the optimal compression of $t$:
\begin{enumerate}
  \item there is \emph{no} $p' \in \plau{t}$ with the same set of matches as $p$, or
  \item there is such a $p'$ but it has \emph{few enough matches} 
  that its larger size does not pay off. 
  % despite a lower use cost per match.
\end{enumerate}
%
% Let us consider both options in more detail.

The first kind of incompleteness occurs when $p$ matches a set of subterms $t_1, \ldots, t_n$ ($n > 2$),
whose join is distinct from all their pairwise joins
(otherwise some $t_i \join t_j \in \plau{t}$ would also match all $t_1, \ldots, t_n$).
An example is shown in \autoref{fig:lattice} (right),
where the three subterms in question are $f(a+a)$, $f(a+c)$, and $f(c+c)$.
In this case, an optimal compression might use the pattern $f(X+Y)$ to rewrite all three subterms,
but our approximation would only include the patterns $f(a+X)$, $f(X+X)$, and $f(X+c)$,
each of which can rewrite only two of the three subterms.

The second kind of incompleteness occurs when there exists $p' \in \plau{t}$
that has the same set of matches as $p$, despite being strictly more specific,
and yet using $p'$ instead of $p$ still doe not pay off.
The understand when this happens, consider the difference in costs between $p$ and $p'$, 
assuming that they are both used to rewrite the same $n$ subterms
(\ie their \textsf{save} cost is the same):
\begin{align*}
  \mathsf{cost}(p', \many{\sigma'}) - \mathsf{cost}(p, \many{\sigma}) &= \sz(p') - \sz(p) + \sum_i (\mathsf{use}(p', \sigma'_i) - \mathsf{use}(p, \sigma_i))\\
  &= \sz(p') - \sz(p) + \sum_i (\sz(\sigma'_i) - \sz(\sigma_i))
\end{align*}
Because $p'$ is strictly more specific than $p$, we know that $\sz(p') \geq \sz(p)$,
but all its substitutions $\sigma'_i$ must be strictly smaller than $\sigma_i$.
Hence, with enough uses, $p'$ is bound to become more compressive than $p$;
when there are just a few uses, however, $p$ can be more optimal.
For example, consider the corpus $\langle \ctx[f(1,2,3)],  \ctx[f(4,5,6)]\rangle$,
where $\ctx$ is some sufficiently large context.
Here, a more general pattern $p = \ctx[X]$ is more optimal than the more specific $p' = \ctx[f(X,Y,Z)]$,
because $\sz(p') - \sz(p) = 3$, each use of $p'$ is only one node cheaper,
and there are only two uses. 

Despite the lack of theoretical completeness guarantees, 
we argue that restricting candidate patterns to $\plau{t}$ is a reasonable trade-off.
Note, that the counter-examples above are quite contrived,
and they no longer apply once the corpus contains sufficiently many and sufficiently diverse instances of a pattern
(for example, adding $f(b,b)$ to the first corpus would make $f(X,Y)$ appear in $\plau{t}$,
and adding just one more occurrence of $p'$ into the second corpus would make it as optimal as $p$).
Our empirical evaluation confirms that this approximation works well in practice.

\section{Library Learning modulo Equational Theory}\label{sec:egraphs}

\begin{figure}
  \begin{minipage}{.5\textwidth}
    \textbf{Syntax}
    $$
    \begin{array}{rll}
      \text{e-class ids} & a,b &\in \mathcal{I} \\
      \text{e-nodes}     & n \gramdef s(a_1, \ldots, a_k)  &\in N\\
      \text{e-classes}   & c \gramdef \{n_1, \ldots, n_m\} &\in C     
    \end{array}
    $$      
  \end{minipage}%
  \begin{minipage}{.5\textwidth}
    \textbf{Denotation} $\denot{\cdot}\colon N\to 2^{\termsof{\Sigma}}$, $\denot{\cdot}\colon C\to 2^{\termsof{\Sigma}}$
    $$
    \begin{array}{rl}
      \\
      \denot{s(a_1, \ldots, a_k)} &= \{s(t_1, \ldots, t_k) \mid t_i \in \denot{M(a_i)}\}\\
      \denot{\{n_1, \ldots, n_m\}} &= \bigcup_{i \in [1,m]} \denot{n_i}
    \end{array}
    $$      
  \end{minipage}
  \caption{Syntax, metavariables, and denotation for the components of an e-graph.
           Here $s \in \Sigma$ and $k = \arity(s)$.}\label{fig:syntax}
\end{figure}

% \subsection{Preliminaries}

\mypara{E-Graphs}
Let $\mathcal{I}$ be a denumerable set of \emph{e-class ids}.
An \emph{e-graph} $\mathcal{G}$ is a triple $\langle C, M, r\rangle$,
where $C$ is a set of \emph{e-classes},
$M\colon \mathcal{I} \to C$ is an \emph{e-class map}.
and $r\in \mathcal{I}$ is the root class id.
% \footnote{
%   Other formalizations of e-graphs do not feature a distinguished root class.
%   We include it, because denoting the e-graph to the set of
%   terms represented by its root e-class (as opposed to other e-classes)
%   is useful for our purposes.
% }
%
An \emph{e-class} $c \in C$ is a set of \emph{e-nodes} $n \in N$,
and an e-node is a constructor applied to e-class ids.
The syntax of e-classes and e-nodes is summarized in \autoref{fig:syntax} (left).
% $$
% c ::= \{n_1, \ldots, n_m\}\quad\quad\quad
% n ::= s(a_1, \ldots, a_k) \quad\text{where}\ s\in \Sigma, k = \arity(s), a_i \in \mathcal{I}\\
% $$
An e-graph has to satisfy the \emph{congruence invariant},
which states that the e-graph has no two identical e-nodes
(or alternatively, that all e-classes are disjoint).%
\footnote{In a real e-graph implementation, the definitions of e-graphs and the congruence invariant
are more involved, 
because efficient merging of e-classes requires introducing a non-trivial equivalence relation over e-class ids;
these details are irrelevant for our purposes.
Also, other formalizations of e-graphs do not feature a distinguished root e-class.} 

The \emph{denotation} of an \emph{e-graph}---the set of terms it represents---%
is the denotation of its root e-class $\denot{M(r)}$,
where the denotation of e-classes and e-nodes is defined mutually-recursively
in \autoref{fig:syntax} (right).
Note that the denotation can be infinite if the e-graph has cycles.
% $$
% \denot{\many{n}} = \bigcup_n \denot{n}\quad\quad\quad
% \denot{s(a_1, \ldots, a_k)} = \{s(t_1, \ldots, t_k) \mid t_i \in \denot{M(a_i)}\}
% $$
An e-graph induces an equivalence relation $\eqg$, 
where $t_1 \eqg t_2$ iff there exists an e-class $c \in C$ such that $t_1 \in \denot{c} \wedge t_2\in \denot{c}$.

E-graphs provide means to \emph{extract} the cheapest term from an e-class according to some cost function:
$\extract_{\cost}(a) = \argmin_{t\in \denot{M(a)}} \cost(t)$.
If $\cost$ is \emph{local},
 meaning that the cost of a term can be computed 
 from the costs of its immediate children,
% \todo{david: we don't necessarily mean monotonic---we're referring to a term
% costing at leasting as much as the SUM of its children}
extraction can be done efficiently by a greedy algorithm,
which recursively extracts the best term from each e-class.

\mypara{E-Matching}
E-matching is a generalization of pattern matching to e-graphs,
where matching an e-class $c$ against a pattern $p$ yields a set of \emph{e-class substitutions} 
$\theta\colon \mathcal{X}\to \mathcal{I}$
such that $\sapp{\theta}{p}$ is a ``subgraph'' of $c$.
To formalize this notion, we introduce \emph{partial terms} $\pi \in \patsof{\Sigma}{\mathcal{I}}$,
which are terms whose leaves can be e-class ids
(or, alternatively, patterns with e-class ids for variables).
The containment relation $\contains{\pi}{a}$ for some e-class id $a$ is defined as follows:
$$
\contains{a}{a} \quad\quad\quad
\contains{s(\pi_1, \ldots, \pi_k)}{a}\ \text{iff}\ s(a_1, \ldots, a_k) \in M(a) \wedge \contains{\pi_i}{a_i}
$$
% if either $\pi = a$ or $\pi$'s constructor occurs as an e-node $n$ in $M(a)$ 
% and all $\pi$'s children are recursively contained in $n$'s children.
%
With this definition, a pattern $p$ matches an e-class $a$, $a \sub p$, if there exists an e-class substitution $\theta$,
such that $\contains{\sapp{\theta}{p}}{a}$.
We denote the set of such substitutions as $\matches{a}{p}$.

\mypara{Rewriting and Equality Saturation}
Equality saturation (EqSat)~\cite{tate2009equality,willsey2021egg} 
takes as input a term $t$ and a set of equations that induce an equivalence relation $\semeq$,
and produces an e-graph $\mathcal{G}$ such that $\denot{\mathcal{G}} = \{t' \mid t \semeq t'\}$
and $t \semeq t'$ iff $t \eqg t'$.
The core idea of EqSat is to convert equations into rewrite rules
and apply them to the e-graph in a non-destructive way:
so that the original term and the rewritten terms are both represented in the same e-class.
Applying a rewrite rule $p_1 \Rightarrow p_2$ to an e-class $a$ works as follows:
for each $\theta \in \matches{a}{p_1}$, 
we obtain the rewritten partial term $\pi' = \sapp{\theta}{p_2}$
and then add this partial term to the same e-class $a$,
restoring the congruence invariant
(\ie merging e-classes that now have identical e-nodes).

\subsection{Top-Level Algorithm}

\begin{figure}
\begin{lstlisting}[
    commentstyle=\small\ttfamily\textcolor{gray},
    basicstyle=\small\ttfamily,
    xleftmargin=2.5em,
    numbers=left]
# Original term $t$, set of equational rewrite rules $\ruleset_\semeq$, maximum library size $N$
def LLMT($t$, $R_\semeq$, $N$):
    # EqSat phase:
    $\mathcal{G}$ = egraph($t$)                            # initialize with a single term $t$
    $\mathcal{G}$ = eqSat($\mathcal{G}, \ruleset_\semeq$)  # $\mathcal{G}$ represents all terms that are $\semeq t$

    # candidate generation phase:
    $\patset$ = AU($\mathcal{G}$)                          # generate candidate patterns by anti-unification
    $\ruleset_\kappa$ = $\{\kappa(p) \mid p \in \patset\}$ # construct a compression rule from every pattern
    $\mathcal{G}'$ = eqSat($\mathcal{G}, \ruleset_\kappa$) # $\mathcal{G}'$ represents all ways to compress $\mathcal{G}$

    # candidate selection phase:
    $\ruleset'$ = select_library($\mathcal{G}', N$)        # select the best $N$ rules from $\ruleset_\kappa$ using beam search
    $\mathcal{G}''$  = eqSat($\mathcal{G}, R'$)            # $\mathcal{G}''$ represents all ways to compress $\mathcal{G}$ using the optimal library          
    return extract($\mathcal{G}''$)                        # extract the smallest compressed term from $\mathcal{G}''$
\end{lstlisting}
  \caption{The top-level LLMT algorithm.}
\label{fig:algo-llmt}
\end{figure}

%\begin{figure}
%\begin{lstlisting}[
%    basicstyle=\normalsize\ttfamily,
%    xleftmargin=2.5em,
%    numbers=left]
%def anti-unify:
%    input: E: an e-graph containing the set of programs to anti-unify
%    output: R: a set of rewrites which introduces function definitions and applications
%               at the point of use
%
%    # TODO: i forgot how this works lmao
%\end{lstlisting}
%\caption{The antiunification algorithm of \babble.}
%\label{fig:algo-au}
%\end{figure}

We can formalize the problem of \emph{library learning modulo equational theory} (LLMT) as follows:
given a term $t$ and a set of equations that induce an equivalence relation $\semeq$,
the goal is to find a compressed term $\hat{t} \in \absof{\Sigma}{\mathcal{X}}$,
such that $\hat{t} \evals t' \semeq t$ (for some $t'$),
and $\hat{t}$ has a minimal size.

% This problem can be equivalently restated as library learning \emph{over an e-graph}:
% given an e-graph $\mathcal{G} = (C, M, r)$,
% the goal is to find a minimal-size compressed term $\hat{t} \in \absof{\Sigma}{\mathcal{X}}$,
% such that $\hat{t} \evals t$ for some $t \in \denot{\mathcal{G}}$.
% %
% If we can solve this problem, we can also solve the original LLMT problem,
% building the e-graph by applying EqSat to the input term and equations.

Our top-level algorithm \T{LLMT} is depicted in \autoref{fig:algo-llmt}.
This algorithm takes as input the original corpus $t$ 
and the equational theory, represented as a set of require rules $\ruleset_\semeq$
(another input to the algorithm is the maximum size of the library;
this parameter is introduced for the sake of efficiency, as we explain in \autoref{sec:beam}).
As the first step (lines 4--5),
\T{LLMT} applies EqSat to obtain an e-graph $\mathcal{G}$ 
that represents all terms $t'$ such that $t' \semeq t$.
This reduces the LLMT problem to library learning over an e-graph:
\ie the goal is to find a minimal-size compressed term $\hat{t}$,
such that $\hat{t} \evals t'$ for some $t' \in \denot{\mathcal{G}}$.
Similarly to \autoref{sec:au},
we take a pattern-based approach to this problem,
that is, we select a set $\patset$ of \emph{candidate patterns}
and then perform compression rewrites using these patterns.

The rest of the algorithm is split into two phases:
\emph{candidate generation} and \emph{candidate selection}.
Candidate generation (lines 8--10)
first computes the set of candidate patterns $\patset$ using the anti-unification mechanism extended to e-graphs.
Then it creates a compression rule ($\kappa$-rule, see \autoref{sec:au}) from each candidate pattern,
and once again applies EqSat to obtain a new e-graph $\mathcal{G}'$.
This new e-graph represents all possible ways to compress the terms from $\mathcal{G}$ 
using patterns in $\patset$.
Finally, the candidate selection phase (lines 13--15)
selects the optimal subset of compression rules $\ruleset'$,
constructs an e-graph $\mathcal{G}''$ that represents all possible compressions using only the selected compression rules,
and finally extracts the smallest compressed term from this e-graph.

The rest of this section focuses on the candidate generation via e-graph anti-unification
(line 8).
The candidate selection functions \T{select\_library} and \T{extract} are discussed in \autoref{sec:beam}.

\subsection{Candidate Generation via E-Graph Anti-Unification}

The goal of candidate generation is to find a set of patterns
that are useful for compression.
Following the discussion in \autoref{sec:au},
we restrict our attention to patterns $\gau(t_1, t_2)$,
where $t_1$ and $t_2$ are subterms of some $t \in \denot{\mathcal{G}}$.
The na\"ive approach is to enumerate all such $t$,
and for each one, perform anti-unification on all pairs of subterms;
this is suboptimal at best,
and impossible at worst (when $\denot{\mathcal{G}}$ is infinite).
Hence in this section we show how to compute a finite set of candidate patterns 
directly on the e-graph,
without discarding any optimal patterns.

\mypara{E-Class Anti-Unification}
Let us first consider anti-unification of two e-classes, $\gau(a,b)$,
which takes as input e-class ids $a$ and $b$
and returns a set of patterns.
We define $\gau(a,b) = \ecau{\emptyseq}{a}{b}$,
where $\ecau{\Gamma}{a}{b}$ is an auxiliary function
that additionally takes into account a \emph{context} $\Gamma$.
A context is a list of pairs of e-classes that have been visited while computing the AU,
and is required to prevent infinite recursion in case of cycles in the e-graph.

\begin{figure}
  \textbf{E-node Anti-Unification} $\ecau{\Gamma}{n_1}{n_2}$
  $$
  \begin{array}{rl}
    \ecau{\Gamma, (a,b)}{s(a_1,\ldots,a_k)}{s(b_1,\ldots,b_k)} &= \{s(p_1,\ldots,p_k) \mid p_i \in \ecau{\Gamma}{a_i}{b_i}\}\\
    \ecau{\Gamma, (a,b)}{s_1(\ldots)}{s_2(\ldots)} &= \{X_{a,b}\} \quad\text{if}\ s_1 \neq s_2\\
    \end{array}
  $$
  \textbf{E-class Anti-Unification} $\ecau{\Gamma}{a}{b}$
  $$
  \begin{array}{rl}
    \ecau{\Gamma}{a}{b} &= \emptyset \quad\text{if}\ (a,b) \in \Gamma\\
    % \textcolor{darkblue}{\ecau{\Gamma}{a}{a}} &\textcolor{darkblue}{= \extract_{\sz}(a)} \\
    \ecau{\Gamma}{a}{b} &= \textcolor{darkblue}{\dominant} \left(
       \bigcup_{n_a \in M(a), n_b \in M(b)} \ecau{\Gamma, (a, b)}{n_a}{n_b} 
    \right)
  \end{array}
  $$
\caption{E-class anti-unification defined as two mutually-recursive functions of e-nodes and e-class ids.}\label{fig:ecau}
\end{figure}

\autoref{fig:ecau} defines this operation using two mutually recursive functions
that anti-unify e-classes and e-nodes.
Note that e-node AU is always invoked in a non-empty context.
The first equation anti-unifies two e-nodes with the same constructor:
in this case, we recursively anti-unify their child e-classes
and return the cross-product of the results.
The second equation applies to e-nodes with different constructors:
as in term AU, this results in a pattern variable.
A nice side-effect of dealing with e-graphs
% is that we can use e-class ids (instead of full terms) 
% in the names the pattern variables $X_{a,b}$.
is that we need not keep track of the anti-substitution
to guarantee that each pair of subterms maps to the same variable:
because any duplicate terms are represented by the same e-class,
we can simply use the e-class ids $a$ and $b$ in the name of the pattern variable $X_{a,b}$.
%
% Although this causes the algorithm to generate equivalent patterns
% (when a pattern has more than one match in either of the e-classes),
% they can be be de-duplicated after the fact.

The second block of equations defines anti-unification of e-classes
(let us first ignore the \textcolor{darkblue}{\sf dominant} which will be explained shortly).
The first equation applies when $a$ and $b$ have already been visited:
in this case, we break the cycle and return the empty set.
Otherwise, the last equation anti-unifies all pairs of e-nodes from the two e-classes
in an updated context and merges the results.
Note that this will add a pattern variable unless all e-nodes in both e-classes
have the same constructor;
we have omitted this detail in \autoref{sec:overview} for simplicity,
but this is implemented in \tool and often yields more optimal patterns.
For example, consider anti-unifying $c_1$ and $c_2$ from \autoref{fig:egraph} (right):
although they have the constructor \T{scale} in common,
$X$ is actually a better pattern for abstracting these two e-classes than $\T{scale}\  X\ Y$,
because the total size of the actual arguments to pattern $X$ ($\pSide{6}$ and $\scale{\repRot{\pSide{8}}{8}{2\pi/8}}{2}$)
is the same as those to the pattern $\T{scale}\  X\ Y$ ($\pSide{6}$, $1$, $\repRot{\pSide{8}}{8}{2\pi/8}$, and $2$),
and the pattern $X$ itself is smaller.
This happens because the class $c_1$ represents several different terms, 
and the term ``compatible with'' $X$ in this case is smaller than the term ``compatible wit'' $\T{scale}\  X\ Y$.

%  because the argument to $X$ (\ie \T{(side 6)}) is shorter than 
%  the total size of both arguments to $Y$ (\ie \T{(side 6)} and \T{1}).

\mypara{Dominant Patterns}
% \mw{this part is a bit confusing, could come back to it}
% \mw{
%   What are we claiming here? That we aren't throwing away anything optimal?
%   Is it possible that a dominated pattern actually matches more terms, and thus is better?
% }
%
Recall that $\gau(a,b)$ produces patterns
with variable names $X_{a_i,b_i}$, which record the e-class ids they abstract;
let us refer to such a pattern $p$ as \emph{uniquely matched}
and define $\theta_l(p) = \{\many{X_{a_i,b_i} \mapsto a_i}\}$ and $\theta_r(p) = \{\many{X_{a_i,b_i} \mapsto b_i}\}$;
these substitutions are necessarily among $\matches{a}{p}$ and $\matches{b}{p}$, respectively.
Given two uniquely matched patterns $p_1,p_2 \in \gau(a,b)$,
we say that $p_1$ \emph{dominates} $p_2$ in the context of $(a,b)$ if
  (1) $\fv(p_1) \subseteq \fv(p_2)$, and
  (2) $\sz(p_1) \leq \sz(p_2)$.
We can show that if $p_1$ dominates $p_2$,
then we can safely discard $p_2$ from the set of candidate patterns.
First, since $p_1$ is no larger than $p_2$, 
the definition of its $\lambda$-abstraction is also no larger.
Second, given a term $t_a\in \denot{a}$,
compressing this term using $p_1$ vs $p_2$,
requires choosing actual arguments from $\rng{\theta_l(p_1)}$ vs $\rng{\theta_l(p_2)}$;
because the former is a subset of the latter,
the first application can always be made no larger
(symmetric argument applies for $t_b\in \denot{b}$).

Hence it is sufficient that $\gau(a,b)$ only returns the set of \emph{dominant patterns}
(\ie a pattern dominated by any pattern in the set can be discarded).
This is what the function $\textcolor{darkblue}{\dominant}$ does in the last equation of \autoref{fig:ecau}.
This pruning technique is especially helpful in the presence of equational theories.
Suppose our theory contains the equation $X + Y \semeq Y + X$,
and that the original term $t$ contains subterms $1 + 2$ and $3 + 1$.
After saturation, the e-graph will represent $1 + 2$ and $2 + 1$ in some e-class $a$
and $3 + 1$ and $1 + 3$ in another e-class $b$;
$\gau(a,b)$ will then produce both patterns $X + 1$ and $1 + X$,
but it is clearly redundant to have both, 
since they match the same e-classes with the same substitutions $\theta$.
Pruning of dominated patterns will eliminate one of them.

\mypara{Avoiding Cycles}
Interestingly enough,
the same notion of dominant patterns justifies why we need not follow cycles in the e-graph when computing $\gau(a,b)$,
or, alternatively, why a finite set of candidate patterns is sufficient
to compress any term in $\denot{\mathcal{G}}$, even if this set is infinite.
Removing the first equation that short-circuits cycles can only lead to solutions $p'$ of the form
$$
p' = \sapp{[X \mapsto p]}{q}
$$
where $p \in \gau(a,b)$ is another solution, and $q$ is some context, added by the cycle.
It is clear that any such $p'$ is dominated by $p$, and hence can be discarded.

% \mypara{Self Anti-Unitification}
% %
% \TODO{I think we should remove this optimization from here and the figure: 
% seems more low-level, and the proof is non-obvious.}
% %
% The final optimization in \autoref{fig:ecau} for anti-unifying an e-class with itself.
% %
% In this case we can show that it is sufficient to simply extract the smallest term from $a$.
% %
% The argument is as follows:
% if $a$ occurs multiple times as a sub-term in the corpus,
% we can either pick the same term from it every time or pick different terms.
% %
% If we do pick the same term,
% then it is best to pick the smallest term and perform common sub-expression elimination on that term,
% which is what the pattern $\extract_{\sz}(a)$ will let us do.
% %
% Picking different terms, also cannot be more optimal,
% because in this case less than the entire term will be reused by the pattern.

% \TODO{Would be nice to add a walk-through example \eg with the e-graph from \autoref{fig:egraph}.}

\mypara{E-Graph Anti-Unification}
The algorithm $\gau(a,b)$ computes a set of patterns that can be used to compress terms
represented by the e-classes $a$ and $b$.
Our ultimate goal, however, is to compute candidate patterns for abstracting
\emph{all subterms} of some $t \in \denot{\mathcal{G}}$.
The most straightforward way to achieve this is to apply $\gau(a,b)$ to all pairs of e-classes in $\mathcal{G}$.
We can do better, however:
some pairs of e-classes need not be considered,
because they cannot occur together in a single term $t$.
For an example, consider the following e-graph, with $\mathcal{I} = \mathbb{N}$ and $r = 0$:
$$
0 \mapsto \{f(1),g(2)\}\quad 1\mapsto\{g(3)\} \quad 2\mapsto \{f(3)\} \quad 3\mapsto \{a\}
$$
This e-graph can result, for example, by rewriting a term $f(g(a))$
using an equation $f(g(X)) \semeq g(f(X))$.
In this e-graph, the e-classes $1$ and $2$ (representing $g(a)$ and $f(a)$, respectively)
clearly cannot \emph{co-occur} in the same term:
since the e-graph only represents two terms, $f(g(a))$ and $g(f(a))$.

To formalize this intuition, we define the co-occurrence relation between two e-class ids as follows.
An e-class $a$ is a \emph{sibling} of $b$ if there is an e-node that has both $a$ and $b$ as children.
An e-class $a$ is an \emph{ancestor} of $b$ if $a = b$ or $b$ is a child of some e-node $n \in M(c)$ and $a$ is an ancestor of $c$;
$a$ is a \emph{proper ancestor} of $b$ if $a$ is an ancestor of $b$ and $a \neq b$.
Two e-classes $a$ and $b$ are \emph{co-occurring} if
  (1) one of them is a proper ancestor of another, or
  (2) they have ancestors that are siblings.

Finally, to compute the set $\gau(\mathcal{G})$ of all candidate patterns for an e-graph $\mathcal{G}$,
\T{LLMT} first computes the co-occurrence relation between all e-classes in $\mathcal{G}$,
and then computes $\gau(a,b)$ of all pairs of e-classes that are co-occurring.
As mentioned in \autoref{sec:overview},
we use a dynamic programming algorithm that memoizes the results of $\gau(a,b)$
to avoid recomputation.

% Note that this will add a pattern variable unless all e-nodes in both e-classes
% have the same constructor;
% we have omitted this detail in \autoref{sec:overview} for simplicity,
% but this is required for completeness and is implemented in \tool.
% %
% For example, consider anti-unifying $c_1$ and $c_2$ from \autoref{fig:egraph} (right):
% although they have the constructor \T{xform} in common,
% $X$ is actually a better pattern for abstracting these two e-classes than $\T{xform}\  X\ Y$,
% because the term \T{line} is shorter than \T{xform line 1}.

\section{Candidate Selection via E-Graph Extraction}\label{sec:beam}

After generating candidate abstractions,
  the \T{LLMT} algorithm invokes \T{select_library} to pick the subset of candidate patterns
  that can best be used to compress the input corpus.
This section describes our approach to selecting the optimal library dubbed
  \emph{targeted common subexpression elimination}.

%In this section we describe \emph{beam extraction},
%  our e-graph based approach to selecting a
%  library from the set of candidate abstractions.

%\subsection{Library Selection as E-Graph Extraction}
\mypara{Library Selection as E-Graph Extraction}
%
% At this stage (lines 6--12 in~\autoref{fig:algo-llmt}),
%   LLMT has initialized the e-graph $\mathcal{G}$
%     with the original input corpus,
%   run equality saturation on $\mathcal{G}$
%     using the given equational theory,
%   and has identified candidate abstractions
%     as a set of patterns $P$ via
%     e-graph anti-unification.
% For each candidate pattern $p \in P$
%   with parameters $\many{X}$,
%   we introduce the rewrite rule
%   $\kappa(p) \triangleq p \Rightarrow (\lambda \many{X} \bnd p)\ \many{X}$
%   as defined in \autoref{sec:au}.
%   % introduce a new rewrite rule
%   % $p \Rightarrow \T{let}\ f_p =
%   %   \lambda \many{X} \bnd p\ \T{in}\ f_p\ \many{X}$.
% LLMT uses these rewrite rules for
%   equality saturation on $\mathcal{G}$ to explore
%   all the ways the corpus may be (re-)expressed
%   using different combinations of
%   the candidate library functions (L13).
% We now desribe the next step in the algorithm which selects
%   the optimal abstractions and uses them to
%   rewrite the original corpus (L15 - L20).
Recall that candidate selection starts with an e-graph $\mathcal{G}'$,
  which represents all the ways of compressing the initial corpus and its equivalent terms
  using the candidate patterns $\patset$.
We will refer to a subset $\libset \subset \patset$ as a \emph{library}.
The optimal size of an e-class $c$ compressed with $\libset$
  can be computed as the sum of the sizes of
  (1) the smallest term $t\in \denot{c}$
      using only the library functions in $\libset$ and
      where $\lambda$-abstractions 
      do not count toward the size of $t$, and
  (2) the smallest version of each $p \in \libset$.
      % (recall that $\lambda$-bodies themselves can be compressed using other patterns).
Note that
  the cost of defining an abstraction
  in $\libset$ is only counted once, and
  that abstractions in $\libset$ can be used
  to compress other abstractions in $\libset$.
Our goal is to find $\libset$ such that the root e-class compressed with $\libset$
  has the smallest size.

Given a \emph{particular} library,
  we can find the size of the smallest term
  via a relatively straightforward top-down traversal of the e-graph.
Hence, a na\"ive approach to library selection
  would be to enumerate all subsets of $\patset$
  and pick the one that produces the smallest term at the root. 
Unfortunately, this approach becomes intractable as the size of $\patset$ grows.
  % $\T{minSize}(c)$
%   represented by each e-class $c$
%   with a fixpoint approach:
%   if e-node $s() \in c$
%     (i.e., $s$ is a leaf with no children),
%     then $\T{minSize}(c) = 1$;
%   otherwise,
%     $\T{minSize}(c) = \T{min}(
%       \{1 + m_1 + \ldots + m_k |\ s(c_1, \ldots, c_k) \in c,\ m_i = \T{minSize}(c_i)\})$.
% %  if we let $T$ be
% %    $\{s(t_1, \ldots, t_k)\ |\ s(c_1, \ldots, c_k) \in c,\ t_i \in \T{minSize}(c_i)\}$
% %    then $\T{minSize}(c)$ is
% %    the subset of $T$ with minimal size.
% Because size never decreases,
%   this process terminates.
% Finding the minimal library is thus naively decidable
%   with a top-down approach:
%   simply enumerate all subsets of $\patset$,
%   compute their library cost, and
%   take the smallest.
% However, this top-down approach becomes intractable
%   once $\patset$ is moderately sized as it will have
%   exponentially many subsets.

\begin{figure}
  \textbf{E-node cost set} $\cs_N(n)$
  $$
  \begin{array}{rl}
    \cs_N(s()) &= \{ (\emptyset, 1) \} \\
    \cs_N(s(\many{a_i})) &= \{ (\libset, u + 1) \mid (\libset, u) \in \mathsf{cross}(\many{a_i}) \} \\
    % \cs_N(\T{let}\ s = \lambda \many{X} \bnd a\ \T{in}\ b) &= \mathsf{addlib}(s, \cs(a), \cs(b)) \\
    \cs_N((\lambda \many{X} \bnd a)\ b) &= \mathsf{addlib}(a, \cs(a), \cs(b)) \\
    \end{array}
  $$

  \textbf{E-class cost set} $\cs(a)$
  $$
  \begin{array}{rl}
    \cs(\{ \many{n_j} \}) &= \sfprune(\sfunify(\bigcup \many{\cs_N(n_j)})) \\
  \end{array}
  $$

  \textbf{Auxiliary definitions} $\mathsf{cross}$, $\mathsf{addlib}$, $\sfunify$, $\sfprune$
  $$
  \begin{array}{rl}
    \mathsf{cross}(z_1, z_2) &= \sfprune(\sfunify(\{ (\libset_1 \cup \libset_2, u_1 + u_2) \mid (\libset_1, u_1) \in z_1, (\libset_2, u_2) \in z_2 \})) \\
    \mathsf{addlib}(a, z_1, z_2) &= \sfprune(\sfunify(\{ (\libset_1 \cup \libset_2 \cup \{ a \}, u_2) \mid (\libset_1, u_1) \in z_1, (\libset_2, u_2) \in z_2 \})) \\
    \sfunify(z) &= \{ (\libset_1, u_1) \in z \mid \forall (\libset_2, u_2) \in z \setminus (\libset_1, u_1).\ \libset_1 \subset \libset_2 \lor u_1 < u_2 \} \\
    \sfprune(N, K, z) &=  \sftopk(\{(\libset, u) \in z \mid |\libset| \leq N \}, K)\\
  \end{array}
  $$
\caption{Cost set propagation, defined as two mutually
  recursive functions. Blue text corresponds to the partial order
  reduction optimization and red text corresponds to the beam approximation.
  $\sftopk(S, K)$ is a helper function that returns top $K$
  elements from the sorted set $S$.}
  \label{fig:costset}
\end{figure}

\mypara{Exploiting Partial Shared Structure}
Instead, \tool selects the optimal library using a bottom-up dynamic programming algorithm.
To this end, it associates each e-node and e-class with a \emph{cost set},
which is a set of pairs $(\libset, u)$
  where $\libset$ is a library and $u$ is the \emph{use cost} of this library,
  \ie the size of the smallest term represented by the e-node / e-class if it is allowed to use $\libset$
  (excluding the size of $\libset$ itself).
Cost sets are propagated up the e-graph using the rules shown in \autoref{fig:costset}.
The base case is a nullary e-node $s()$,
  which cannot use any library functions and whose size is always $1$.
For an e-node that has children,
\babble computes the $\mathsf{cross}$ product over the cost sets of all its child \eclasses.
Finally, for an application e-node $(\lambda \many{X} \bnd a)\ b$,
  the cost set \emph{must include} the library function $a$
  (in addition to some combination of libraries from the cost sets of $a$ and $b$);
  note that the use cost of the abstraction node only includes the use cost of $b$,
  since abstraction bodies are excluded from the use cost.
% Here, in addition to the cross-product between the child cost sets,
%   we need to consider the option of adding the current $\lambda$-abstraction to the children's libraries.

To compute the cost set of an e-class,
  \babble takes the union of the cost sets of all its e-nodes.
However, doing this naively would result in the size of the cost set growing exponentially.
To mitigate this, we define a \emph{partial order reduction} ($\sfunify$),
which only prunes provably sub-optimal cost sets.
Given two pairs $(\libset_1, u_1)$ and $(\libset_2, u_2)$ in the cost set of an e-class, 
if $\libset_1 \subset \libset_2$ and $u_1 \leq u_2$, 
then $\libset_2$ is subsumed by $\libset_1$ and can be removed from the cost set,
intuitively because $\libset_1$ can compress this e-class even better and with fewer library functions.
In practice this optimization prunes libraries with redundant abstraction,
where two different abstractions can be used to compress the same subterms.

\mypara{Beam Approximation}
Even with the partial order reduction,
  calculating the cost set for every e-node and e-class
  can blow up exponentially.
%Even with the partial order reduction, calculating the cost
%set for every e-node and e-class in an e-graph representing a large
%program can result in exponential blowup.
%
To mitigate this,
  \tool provides the option to limit
  both the \emph{size of each library} inside a cost set
  and the \emph{size of the cost set} stored for each e-class.
This results in a \emph{beam-search} style algorithm,
where the cost set of each e-class is $\sfprune$ed, as shown in \autoref{fig:costset}.
This pruning operation first filters out libraries that have more than $N$ patterns,
then ranks the rest by the total cost (\ie the sum of use cost and the size of the library),
and finally returns the top $K$ libraries from that set.
  
  % limits the number of candidate libraries
  % (\ie a \emph{beam size})
  % and the size of each library considered at each e-class,
  % resulting in a novel beam search approach
  % to efficient, sharing-aware, e-graph extraction.
%
% This \emph{pruning} operation is listed as $\sfprune$
%   in \autoref{fig:costset}.
% It first prunes out libraries that have more than $\mathsf{n}$ abstractions,
%   then ranks the rest by the combined
%   cost of the abstractions and the term they appear in, and finally
%   returns the top $\mathsf{m}$ from that set.

%In practice, \tool provides multiple parameters which limit the number
%of candidate libraries at each e-class (\ie the \emph{beam size}),
%as well as the size of each library.
% As the next section shows, this technique is effective
%   in practice.

%def cs-prune($a, M, n$):
%    sort_by(a.libs, full-cost)
%    for a-sel in $a$:
%       if len(a.libs) >= n:
%          a.remove(a-sel)
%    retain($a$, M)
%    return sort($a$)
%\input{extract_algorithm}

\section{Evaluation}
\label{sec:eval}

We evaluated \tool and the LLMT algorithm behind it
 with two quantitative research questions and a third qualitative one:

\begin{enumerate}[label=RQ \arabic*.,ref=RQ \arabic*]
  \item Can \tool compress programs better than a state-of-the-art library learning tool?\label{rq:1}
  \item Are the main techniques in LLMT (anti-unification and equational theories)
        important to the algorithm's performance?\label{rq:2}
  \item Do the functions \tool learns make intuitive sense?\label{rq:3}
\end{enumerate}

\mypara{Benchmark Selection}
We use two suites of benchmarks to evaluate \tool, both shown in \autoref{tab:benchmarks}.
The first suite originates from the \dc work~\cite{dreamcoder},
 and is available as a public repository~\cite{compression-bench}.
\dc is the current state-of-the-art library learning tool,
and using these benchmarks allows us to perform a head-to-head comparison.
The \dc benchmarks are split into five domains (each with a different DSL);
we selected two of the domains---List and Physics---which we understood best,
to add an equational theory.

The second benchmark suite, called \cogsci, comes from \citet{cogsci-dataset}.
This work collects a large suite of programs in a graphics DSL for the
 purpose of studying connections between the generated objects and their
 natural language descriptions.
There are 1,000 programs in the ``Drawings'' portion of this dataset,
 divided into four subdomains (listed in \autoref{tab:benchmarks})
 of 250 programs each.\looseness=-1
%For use in \tool,
% we combine each suite of 250 programs into a large single benchmark
% from which to learn libraries.

We ran \tool on all benchmarks on an AMD EPYC 7702P processor at 2.0 GHz.
Each benchmark was run on a single core.
The \dc results were taken from
 the benchmark repository~\cite{compression-bench};
 \dc was run on 8 cores of an AMD EPYC 7302 processor at 3.0 GHz.

%Times for \dc were taken from \citet{dreamcoder}
%  as we were unable to run \dc,
%  much less reproduce their results,
%  even after multiple consultations with the authors.
%\mw{should we say we can't run \dc? is this too aggressive?}

% %
% Describe the two benchmark suites: \dc and \cogsci.
% %
% Explain the subtlety that for the \dc benchmarks we want to compress the minimal solution
% to each synthesis problem, and how we handle that
% (we put different solutions into the same e-class).
% %
% Explain the role of the \cogsci domain: these are benchmarks with more and larger programs,
% which \dc cannot handle.

\begin{table}
  \centering
  \begin{tabular}{lrr}
    \multicolumn{3}{c}{\dc~\cite{dreamcoder}} \\
    \bf Domain & \bf \# Benchmarks & \bf \# Eqs \\
    List    & 59 & 14 \\
    Physics & 18 & 8 \\
    Text    & 66 & - \\
    Logo    & 12 & - \\
    Towers  & 18 & -
  \end{tabular}\qquad%
  \begin{tabular}{lrr}
    \multicolumn{3}{c}{\cogsci~\cite{cogsci-dataset}} \\
    \bf Domain & \bf \# Benchmarks & \bf \# Eqs \\
    Nuts \& Bolts & 1 & 7 \\
    Vehicles      & 1 & 9 \\
    Gadgets       & 1 & 17 \\
    Furniture     & 1 & 9 \\
    & %dummy line
  \end{tabular}
  \caption{
    We selected our benchmark domains from two previous works:
     \dc~\cite{dreamcoder}
     and \cogsci~\cite{cogsci-dataset}.
    Each domain from \dc has multiple benchmarks;
     \cogsci has one large benchmark per domain.
    For some domains,
     we additionally supplied \tool with an equational theory;
     we report the number of equations in the final column.
  }\label{tab:benchmarks}
\end{table}

\subsection{Comparison with \dc}

\begin{figure}
  \centering
  \begin{subfigure}{0.48\linewidth}
    \includegraphics[width=\linewidth]{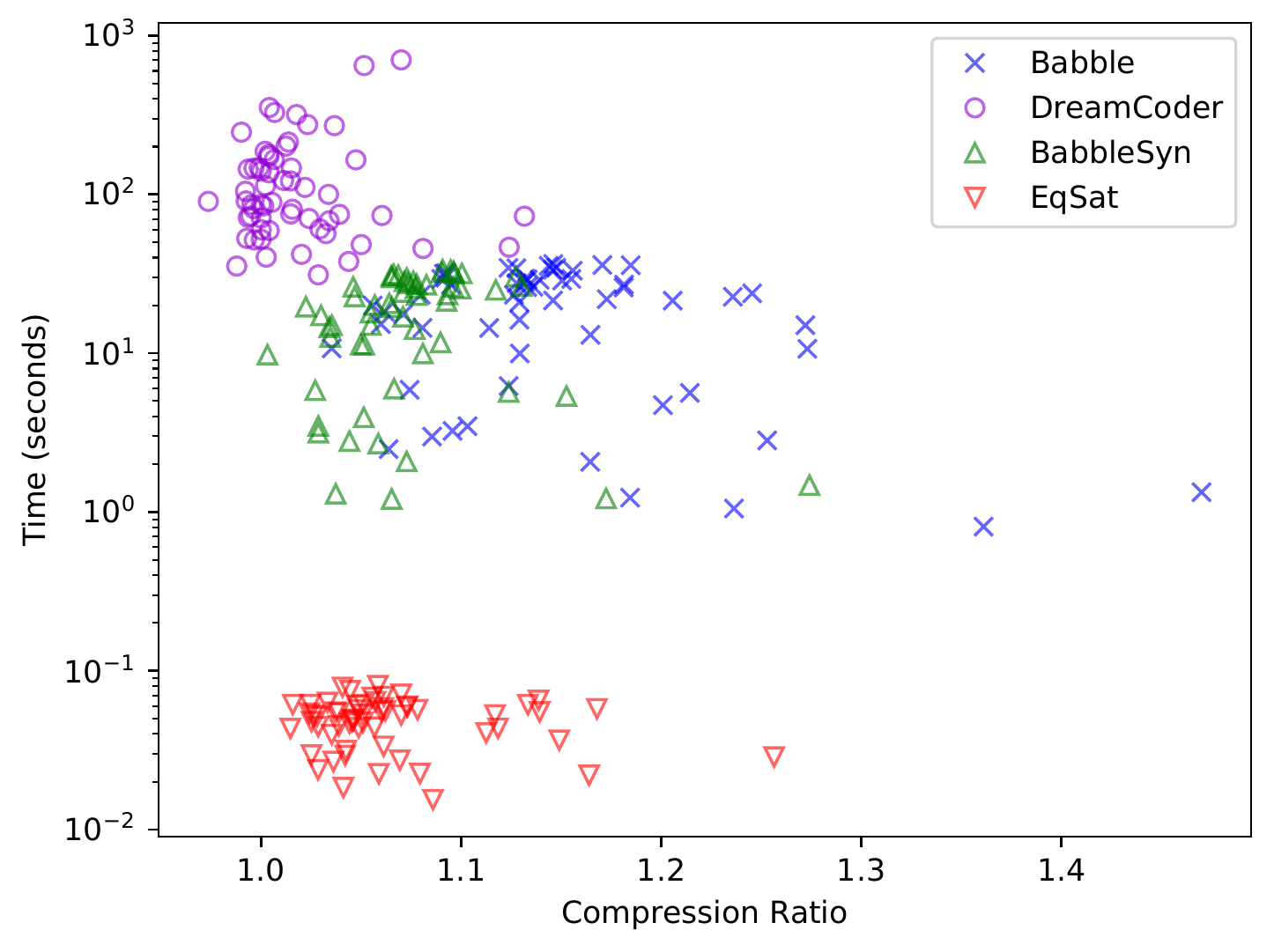}
    \caption{List domain}
    \label{fig:dc-domains-list}
  \end{subfigure}
  \hfill
  \begin{subfigure}{0.48\linewidth}
    \includegraphics[width=\linewidth]{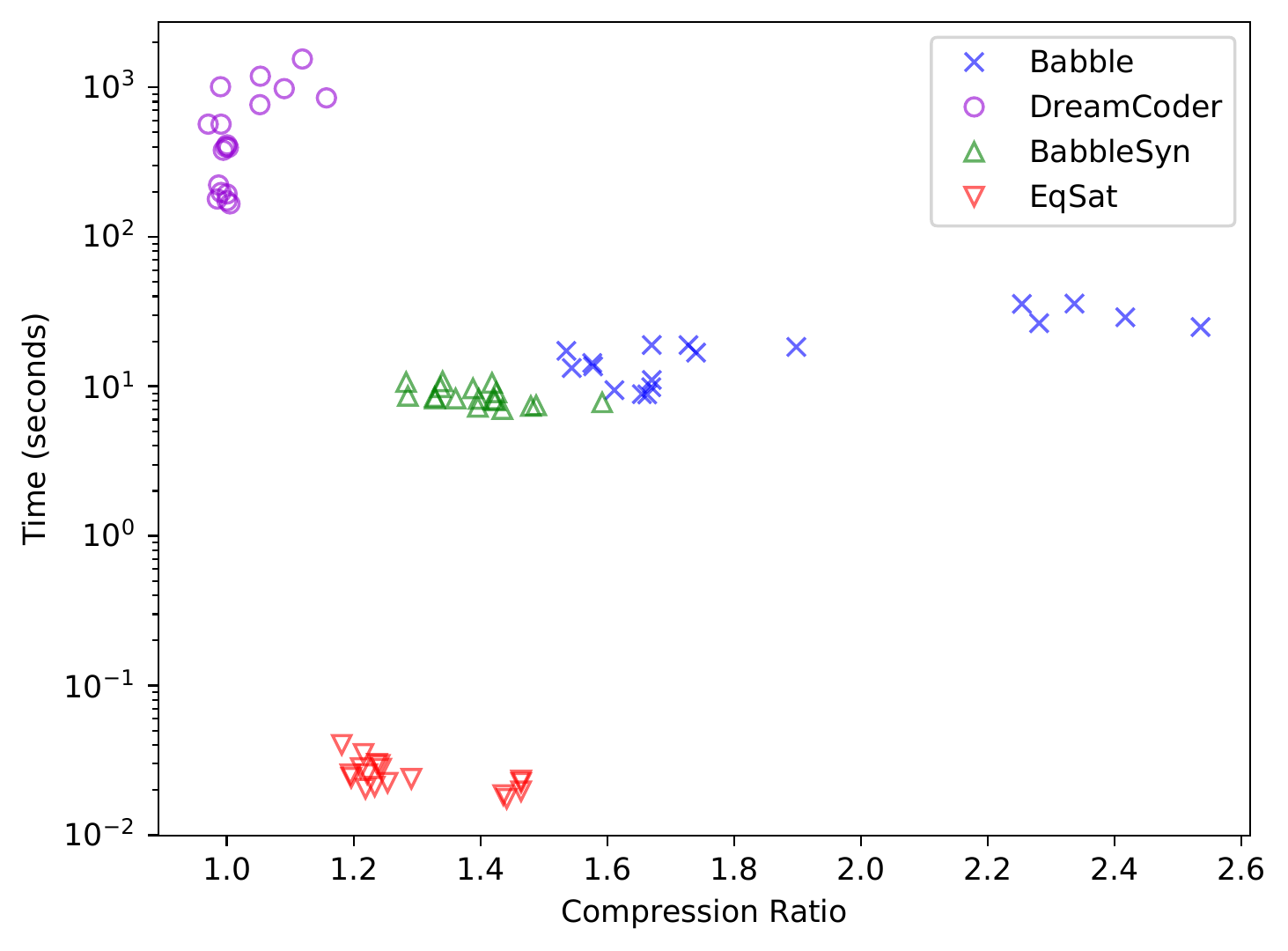}
    \caption{Physics domain}
    \label{fig:dc-domains-physics}
  \end{subfigure}

  \begin{subfigure}{0.48\linewidth}
    \includegraphics[width=\linewidth]{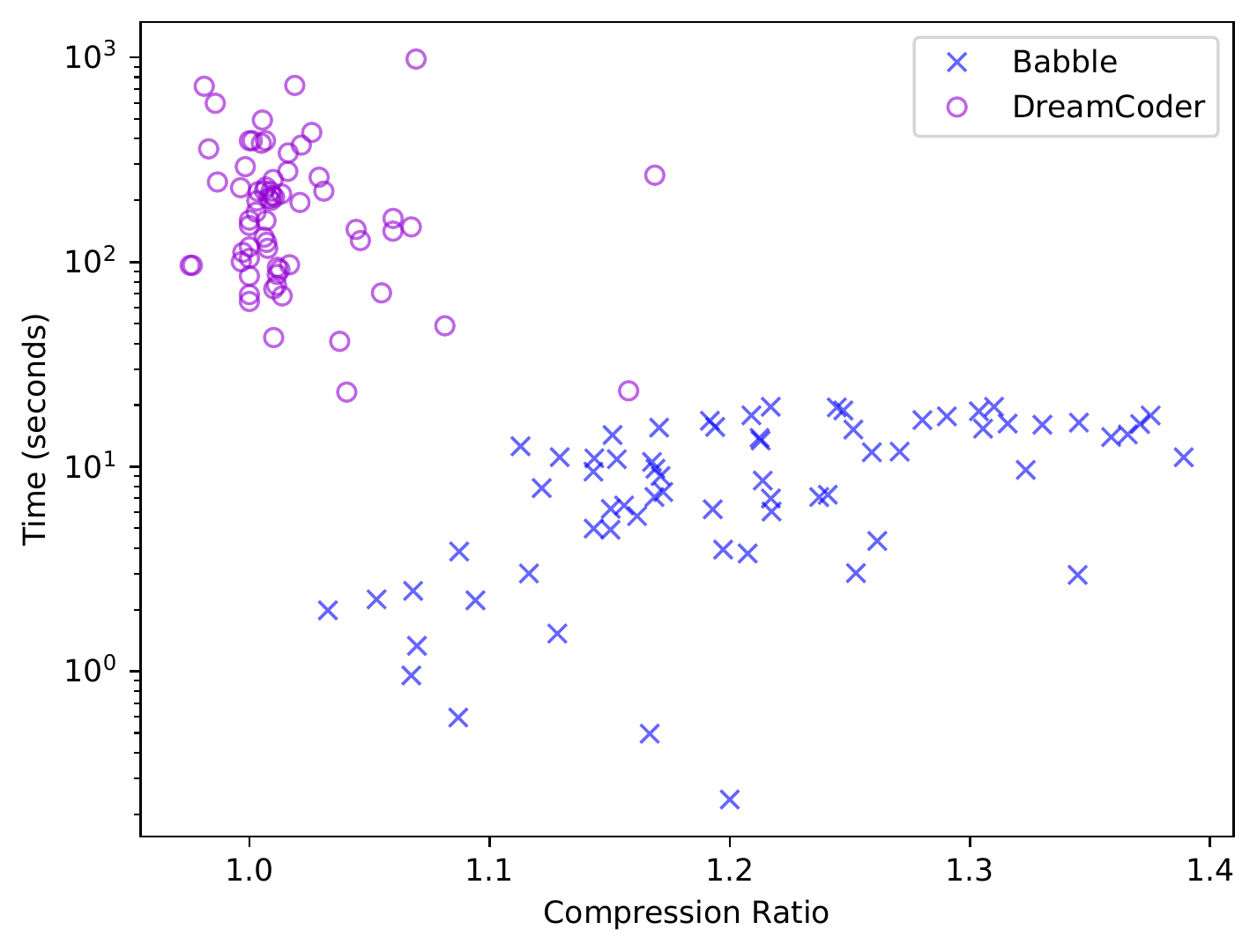}
    \caption{Text domain}
  \end{subfigure}
  \hfill
  \begin{subfigure}{0.48\linewidth}
    \includegraphics[width=\linewidth]{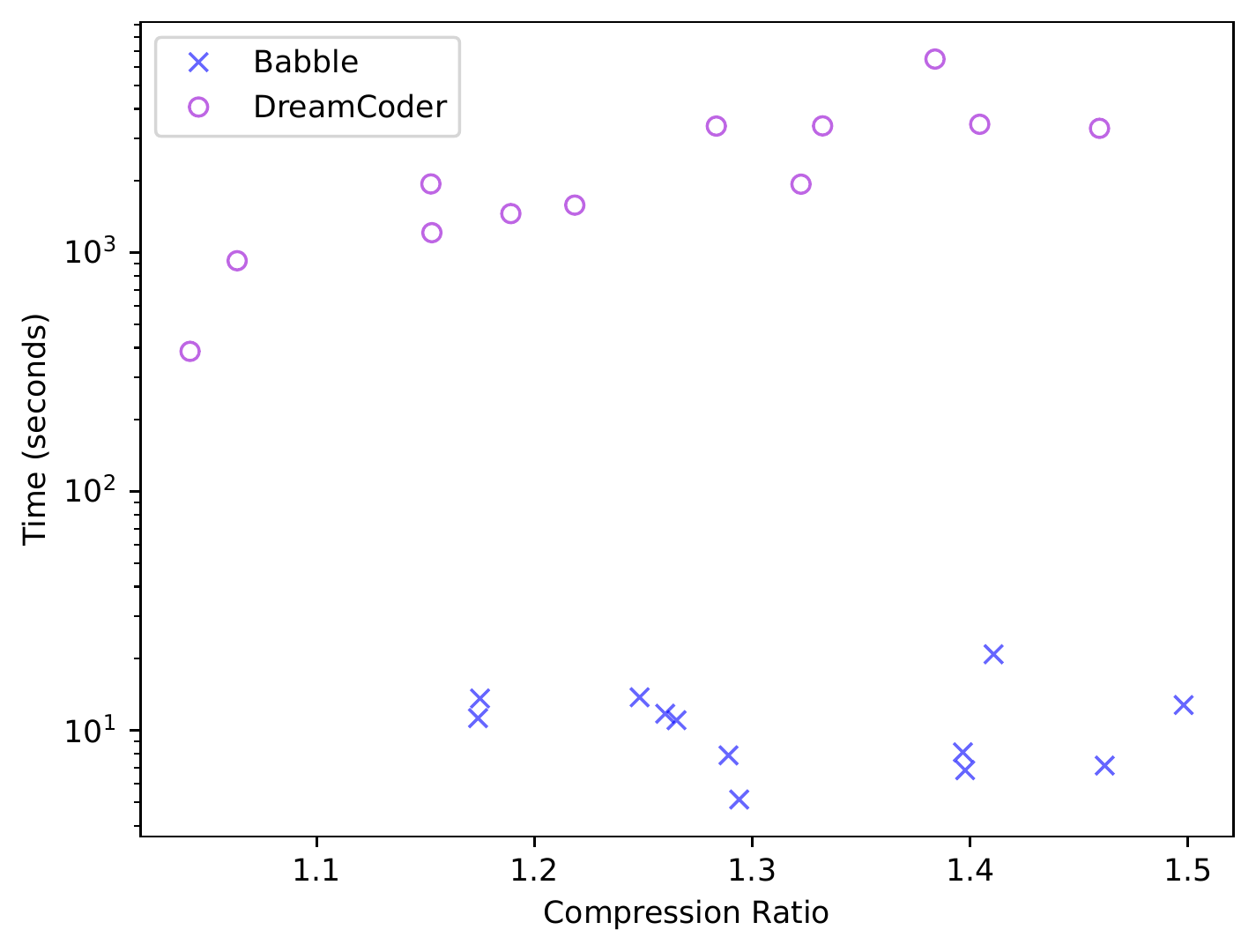}
    \caption{Logo domain}
  \end{subfigure}

  \begin{subfigure}{0.48\linewidth}
    \includegraphics[width=\linewidth]{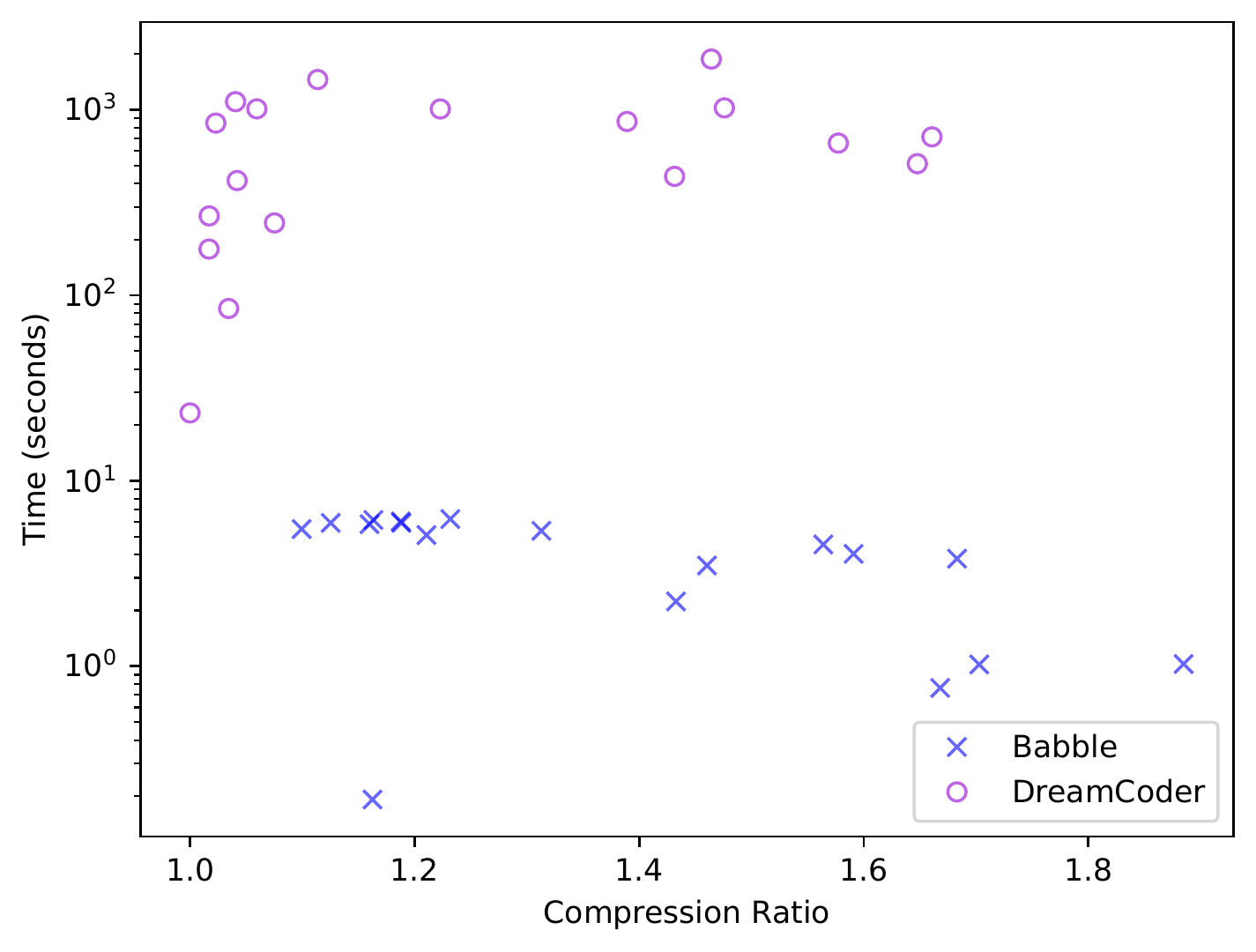}
    \caption{Towers domain}
  \end{subfigure}
  \caption{
    \tool consistently achieves better compression ratios than
     \dc on benchmarks from the \dc domains,
     and it does so 1--2 orders of magnitude faster.
    Each marker shows the compression ratio (x-axis)
     and run time (y-axis) of a benchmark.
    Each benchmark is one \dc input,
     \ie, a set of groups of programs as described above.
    Lower and to the right is better.
    In the domains where we supplied
     \tool with an equational theory
     (List and Physics),
     additional markers show the performance
     of \tool using purely syntactic learning (without equations, ``BabbleSyn'')
     or only equality saturation without library learning (``EqSat'').
  }\label{fig:dc-domains}
\end{figure}

To answer \ref{rq:1},
 we compare to the state-of-the-art \dc tool~\cite{dreamcoder} on its own benchmarks.
The \dc benchmarks are suited to its workflow;
while the input to a library learning task
 is conceptually a set of programs (or just one big program),
 each input to \dc is a set of \emph{groups} of programs.
Each group is a set of programs
 that are all output from the same program synthesis task
 (from an earlier part of the \dc pipeline).
When compressing a program via library learning,
 \dc is minimizing the cost of the program
 made by concatenating the most compressed
 program from each group together, in other words:
\[
  \sum_{\textsf{group $g$}}
  \min_{\textsf{program $p \in g$}}
  \textsf{cost}(p)
\]
\dc takes this approach to give its library learning
 component many variants of the same program,
 in order to introduce more shared structure
 between solution programs across different synthesis problems.
 %in hopes that some variants will compress better than others in different contexts.
 %\mw{someone who understands dc should check me on this}

To implement \dc's benchmarks in \tool,
 we use the \egraph to capture the notion of program variants in a group.
Since every program in a group is the output of the same synthesis task,
 \tool considers them equivalent and places them in the same \eclass.

% The input to each library learning task
%  is a set $P$ of sets of programs:
% \[
%   P = \{
%     \{p_{1,1}, p_{1,2}, \ldots \},
%     \{p_{2,1}, p_{2,2}, \ldots \},
%     \{p_{3,1}, p_{3,2}, \ldots \},
%     \ldots
%   \}
% \]
% Let $p_i = \{p_{i,j} \mid j\}$ refer to an inner set of programs.
% Each program in a set $p_i$
%  is the output from the same synthesis task.

\mypara{Results}
We ran \tool on five domains from the \dc benchmark suite.
The results are shown in \autoref{fig:dc-domains}.
In summary,
 \tool consistently achieves better compression ratios than
 \dc on benchmarks from the \dc domains,
 and it does so 1--2 orders of magnitude faster.

% without any DSRs.
% %
% Show that we get better compression
% \TODO{even if we run for as many rounds as \dc?}
% %
% Say something about time/memory to show that we are not much slower than DC.

\mypara{The Role of Equational Theory}
\label{sec:ablation}
To answer \ref{rq:2},
 we again turn to the \dc benchmarks,
 focusing on the domains where we supplied \tool
 with an equational theory: List and Physics.
In these domains, we ran \tool in two additional configurations:
\begin{itemize}
  \item ``BabbleSyn'' ignores the equational theory, just doing syntactic library learning.
  \item ``EqSat'' just optimizes the program using Equality Saturation with the
        rewrites from the equational theory.
        This configuration \emph{does not} do any library learning.
\end{itemize}
\autoref{fig:dc-domains-list} and \ref{fig:dc-domains-physics}
 show the results for these additional configurations,
 as well as \dc and the normal \tool configuration.
All \tool configurations rely on
  targeted subexpression elimination
  to select the final learned library.
The ``EqSat'' configuration is unsurprisingly very fast
 but performs relatively little compression,
 as it does not learn any library abstractions.
The ``BabbleSyn'' configuration
 does indeed compress the inputs,
 in fact it is still better than \dc in both domains.
However,
 the addition of the equational theory (the ``\tool'' markers in the plots)
 significantly improves compression
 and adds relatively little run time,
 well within an order of magnitude.

\subsection{Large-Scale \cogsci Benchmarks}

\begin{table}
  \begin{tabular}{lr|rrr|rrr}
    \multicolumn{2}{c}{} & \multicolumn{3}{c}{Without Eqs} & \multicolumn{3}{c}{With Eqs} \\
    \bf Benchmark & \bf Input Size & \bf Out Size & \bf CR & \bf Time (s) & \bf Out Size & \bf CR & \bf Time (s) \\
    Nuts \& Bolts   &      19009 &       2059 &       9.23 &      18.74 &       1744 &      10.90 &      40.75 \\
    Vehicles        &      35427 &       6477 &       5.47 &      79.50 &       5505 &       6.44 &      78.03 \\
    Gadgets         &      35713 &       6798 &       5.25 &      75.07 &       5037 &       7.09 &      82.29 \\
    Furniture       &      42936 &      10539 &       4.07 &     133.25 &       9417 &       4.56 &     110.00 \\
    \hline
    Nuts \& Bolts (clean) & 18259 &       2215 &       8.24 &      18.12 &       1744 &      10.47 &      40.91
  \end{tabular}
  \vspace{1em}
  \captionof{table}{
    Results for running \tool on the four large examples
     from the \cogsci dataset~\cite{cogsci-dataset},
     both without and with an equational theory.
    Each row includes the input and output AST sizes,
     the compression ratio (CR),
     and \tool's run time in seconds.
    \autoref{fig:cogsci-scatter} plots this data.
    The final row shows performance on a modified dataset.
  }\label{tab:cogsci}
\end{table}
\begin{figure}
  \includegraphics[height=5cm]{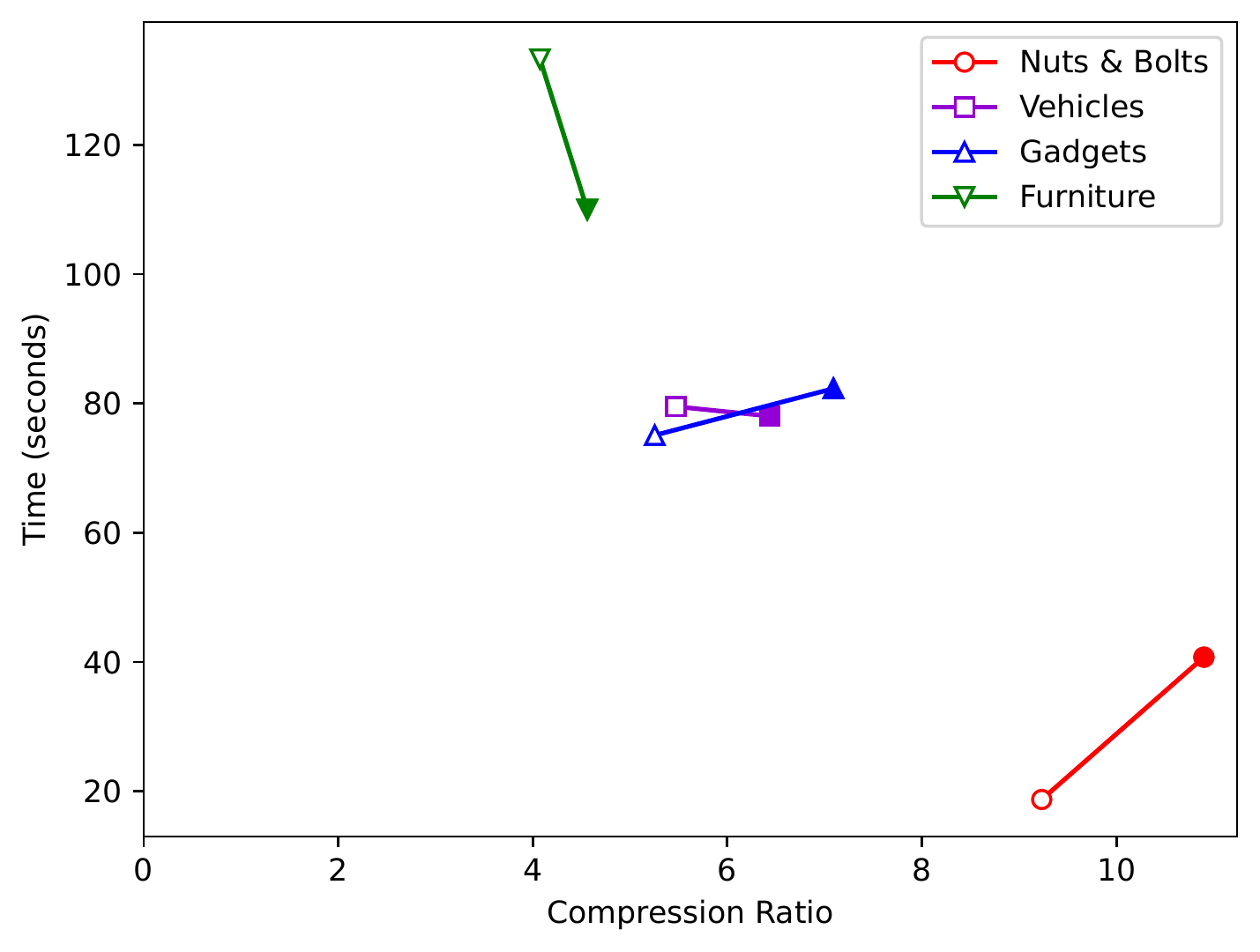}
  \caption{
    Data from \autoref{tab:cogsci} in scatter plot.
    Each line segment shows compression ratio and run time for a domain
     without (hollow marker) and with (solid marker) an equational theory.
    Using an equational theory improves compression in all cases,
     and even improves run time in two cases.
  }\label{fig:cogsci-scatter}
\end{figure}

The previous section demonstrated that \tool's performance
 far surpasses the state of the art.
In this section,
 we present and discuss the results
 of running \tool on benchmarks from the \cogsci domain.
These benchmarks,
 taken from \citet{cogsci-dataset},
 are significantly larger (roughly 10x - 100x)
 than those from the \dc dataset
 and out of reach for \dc.

\mypara{Quantitative Results}
\autoref{tab:cogsci} and \autoref{fig:cogsci-scatter} show the results
 of running \tool on the benchmarks from the \cogsci domain.
The plot in \autoref{fig:cogsci-scatter} makes two observations clear.
First and relevant for \ref{rq:2},
  the addition of an equational theory improved all four benchmarks
  (the solid marker is always to the right of the hollow marker).
Second, and perhaps surprisingly,
 equational theories can sometimes make \tool \emph{faster}!
This is consistent with previous observations about equality saturation~\cite{willsey2021egg}:
 while equality saturation typically makes an e-graph larger,
 it can sometimes combine two relevant e-classes into one and
 reduce the amount of work that some operation over an e-graph must do.

We also observed that the \nuts dataset contains
  several redundant transformations,
  like the ``scale by 1'' featured in the running example of \autoref{sec:overview}. %in \autoref{eq:unit-xform} and \autoref{fig:polygons}.
These redundancies can be useful for finding
  abstractions in the absence of an equational theory.
However, they should not be required in
  \tool since LLMT can
  introduce the redundancies wherever
  required.
We therefore removed all
  existing redundant transformations
  from \nuts and ran \tool
  on the transformed dataset.
The results are in the final row of \autoref{tab:cogsci}.
On the modified dataset,
 \tool achieves identical compression when using the equational theory,
 but without the equations it performs worse than on the unmodified dataset.

\mypara{Qualitative Evaluation}
\label{sec:quali}
\begin{figure}
  \centering
  \includegraphics[width=\textwidth]{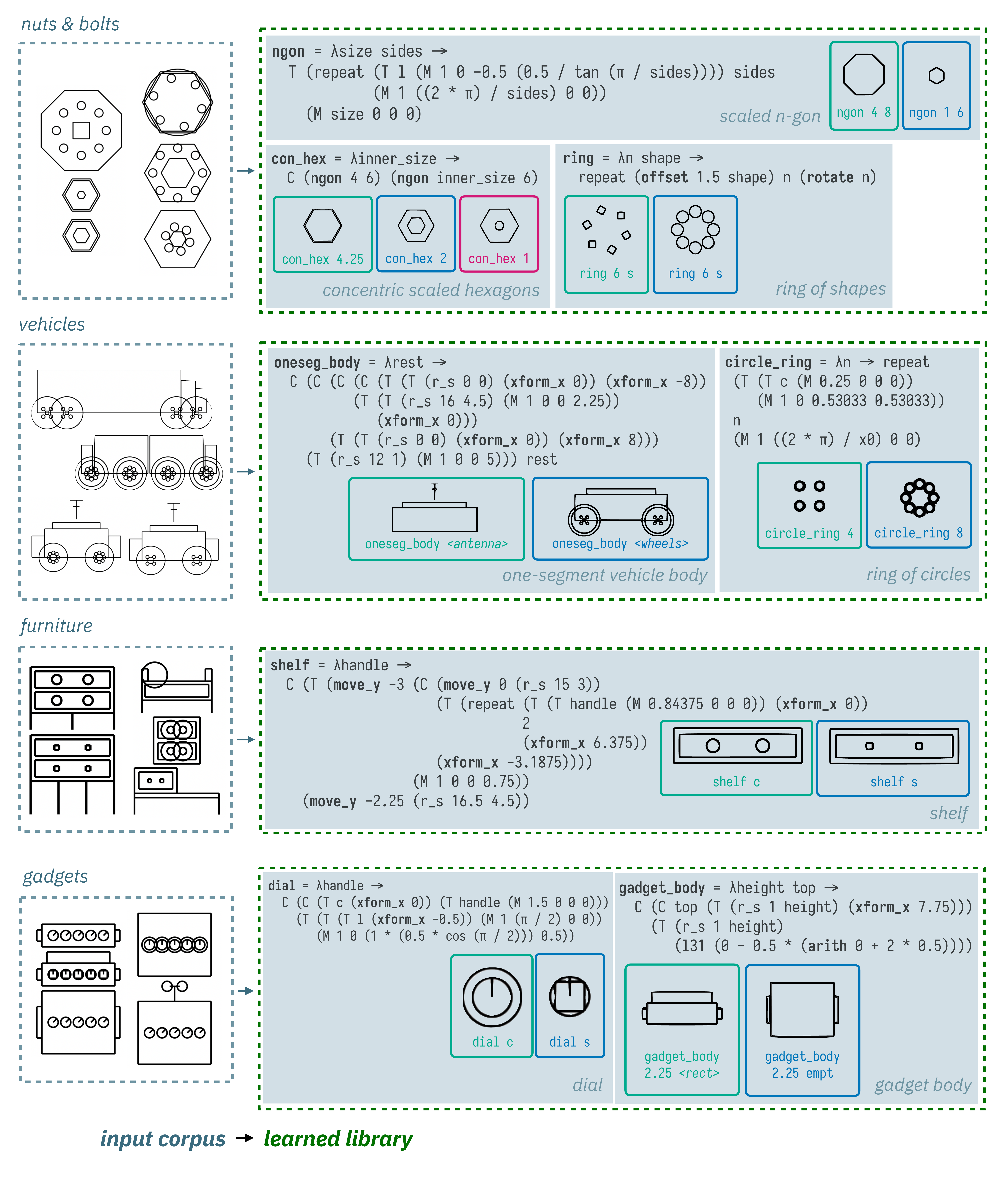}
  \caption{A selection of evaluated programs from each of the domains
           in the \cogsci dataset, along with a selection of functions
           that \tool learns within the first ten rounds for each domain.
           Bolded functions represent learned abstractions.
           Note this figure uses the concrete syntax from the \cogsci
           dataset; it is similar to the simplified form shown in the overview.
           We named the learned functions and their parameters for clarity.
   }\label{fig:qual_cogsci}
\end{figure}
\autoref{fig:qual_cogsci}
 highlights a sample of abstractions that
 \tool discovered from the \cogsci benchmarks.
We ran \tool on each of the benchmarks
  and applied the learned abstractions on a few input programs to
  visualize their usage.
Questions about usability and readability of
  learned libraries are difficult to answer
  without rigorous user studies which we leave for
  future work.
Nevertheless, \autoref{fig:qual_cogsci} shows that
  \tool identifies common structures that
  are similar across
  different benchmarks, which
  makes its output easier to reuse and interpret.

First, we revisit the \nuts example
  from \autoref{sec:intro}:
  \autoref{fig:qual_cogsci} shows that \tool %successfully
  learns the scaled polygon (\texttt{ngon}) abstraction which
  is applicable to several programs in the dataset.
We also see that \tool consistently
  finds a similar abstraction
  representing a ``ring of shapes'' for both
  \nuts and Vehicles.
Finally, as the Gadgets example shows,
  \tool finds abstractions for both
  the entire model as well as its components.
In this case,
  it learned the function \texttt{gadget\_body}
  that abstracts the entire outer shape,
  and it also learned \texttt{dial} that abstracts
  the handles of the outer shape.\looseness=-1

\section{Related Work}
\label{sec:related}

\tool is inspired by work on library learning,
specifically the \dc line of work,
as well as equality saturation-based program synthesis and decompilation.

\mypara{\dc}
\dc~\cite{dreamcoder} is a program synthesizer
  that learns a library of abstractions
  from solutions to a set of related synthesis tasks.
The library is intended to be used for solving other
  similar synthesis tasks.
% The DSLs for the different domains in \dc
%   are carefully designed to help scale the search.
\dc
  uses version spaces~\cite{vsa1, vsa2}
  to compactly store a large
  set of programs
  and leverages ideas from \egraphs
  (such as e-matching)
  but only for exploring the space of refactorings
  of the original program using the candidate libraries,
  not for making library learning robust to syntactic variation.
Our evaluation shows that \tool can find more optimal abstractions faster than \dc.

\dc has sparked several direction of follow-up work
that attempt to improve the efficiency of its library learning procedure
and the quality of the learned abstractions.
One of them is by \citet{laps},
which uses natural language annotations and a neural network
to guide library learning.
Another one is \stitch~\cite{stitch},
which was developed \emph{concurrently} with this work
and is the most closely related to \tool;
we discuss \stitch in some detail below.

\mypara{\stitch}
The core difference between the two approaches
is that \stitch focuses on improving \emph{efficiency} of purely syntactic library learning,
whereas \tool attempts to improve its \emph{expressiveness} by adding equational theories.
While \tool separates library learning into two phases---%
candidate generation via anti-unification and candidate selection via e-graph extraction---%
the \stitch algorithm \emph{interleaves} the generation and selection phases
in a branch-and-bound top-down search.
Starting from the ``top'' pattern $X$,
\stitch gradually refines it until further refinement does not pay off.
To quickly prune suboptimal candidate patterns,
\stitch computes an upper bound on their compression
by summing up the compression at each match of this pattern in the corpus
(this bound is imprecise because it does not take into account that matches might overlap).
For candidates that are not pruned this way,
\stitch computes their true compression by searching for the optimal subset of matches to rewrite,
a so-called ``rewrite strategy''.
\tool's extraction algorithm can be seen as a generalization of \stitch's rewrite strategy:
while the former searches over both subsets of patterns and how to apply them to the corpus at the same time,
the latter considers a single pattern at a time and only searches for the best way to apply it.
Since the search space in the former case is much larger,
\tool uses a beam search approximation,
while in \stitch the rewrite strategy is precise.
To sum up, the main pros and cons of the two approaches are:
\begin{itemize}
\item \tool can learn libraries modulo equational theories, while \stitch cannot;
\item \stitch provides optimality guarantees for learning a single best abstraction at a time, while \tool can learn multiple abstractions at once, but sacrifices theoretical optimality.
\end{itemize}

\mypara{Other Library Learning Techniques}
\tname{Knorf}~\cite{knorf} is a library learning tool for logic programs,
which, like \tool, proceeds in two phases.
Their candidate generation phase is similar to the upper bound computation in \stitch,
while their selection phase uses an off-the-shelf constraint solver.
It would be interesting to explore whether their constraint-based technique can be generalized beyond logic programs.

Other work develops limited forms of library learning,
where only certain kinds of sub-terms can be abstracted.
For example, ShapeMod~\cite{shapemod} learns macros
  for 3D shapes represented
  in a DSL called ShapeAssembly,
  and only supports abstracting over numeric parameters,
  like dimensions of shapes.
Our own prior work~\cite{cogsci-smiley} extracts common structure from graphical programs,
  but only supports abstracting over primitive shapes
  and applying the abstraction at the top level of the program.
Such restrictions make the library learning problem more computationally tractable,
  but limit the expressiveness of the learned abstractions.

There are several neural program synthesis tools%
  ~\cite{polozov, srini, learner, gredilla}
  that learn \emph{programming idioms}
  using statistical techniques.
  % learn programming idioms that make synthesis more efficient.
Some of these tools have used
  ``explore-compress'' algorithms~\cite{learner} to
  iteratively enumerate a set of
  programs from a grammar and find a solution
  that exposes abstractions that make the set of
  programs maximally compressible.
This is similar to
  common subexpression elimination which \tool uses
  for guiding extraction.
%\cite{gredilla} have used a similar approach for
%  helping robots learn to complete various tasks by
%  selecting from a library of abstractions.
%A recent work from us combines this
%  explore-compress algorithm with \egraphs
%  for co-optimizing designs and fabrication instructions
%  for carpentry~\cite{cc2}
%  although it does not focus on synthesizing
%  new abstractions.

% \tool's library learning is based on \egraph
%   anti-unification combined with DSRs, which differs
%   different from prior approaches; our
%   results show that it can find better compressions
%   faster than the state-of-the-art~\cite{dreamcoder}.

\mypara{Loop rerolling}
Loop rerolling is related to library learning in that it also aims to discover hidden structure in a program,
except that this structure is in the form of loops.
A variety of domains have used loop rerolling
  to infer abstractions from flat input programs.
In hardware, loop rerolling is
  used to optimize programs
  for code size~\cite{hardware1, hardware2, hardware3}.
In many of these tools,
  the compiler first unrolls a loop,
  applies optimizations, then rerolls it --- the compiler
  therefore has structural information about the loop
  that can be used for rerolling~\cite{hardware1}.
The graphics domain has used
  loop-rerolling to discover
  latent structure from low-level representations.
CSGNet~\cite{csgnet} and Shape2Prog~\cite{shape2prog} used
  neural program generators to discover for loops from
  pixel- and voxel-based input representations.
\cite{handdrawn} used program synthesis
  and machine learning
  to infer loops from hand-drawn images.
Szalinski~\cite{szalinski} used
  equality saturation to
  automatically learn loops in the form
  of maps and folds
  from flat 3D CAD programs that are
  synthesized by mesh decompilation tools~\cite{reincarnate}.
WebRobot~\cite{webrobot} has used speculative
  rewriting for inferring loops from traces
  of web interactions.
% WebRobot is similar to Szalinski in that they both
%   use the notion of speculative rewrites.
Similar to \tool (and unlike Szalinski),
  WebRobot finds abstractions
  over \textit{multiple} input traces.
% Compared to these tools however,
%   \tool's approach is more general --- it is not
%   restricted to finding only loops or limited
%   types of macros and does not rely on
%   any information
%   about how the original
%   ``flat'' programs were generated.

\mypara{Applications of Anti-Unification}
Anti-unification is a well-established technique for discovering common structure in programs.
It is the core idea behind bottom-up Inductive Logic Programming~\cite{ilp},
and has also been used for software clone detection~\cite{Bulychev10},
programming by example~\cite{raza2014programming},
and learning repetitive code edits~\cite{lase,refaser}.
It is possible that these applications could also benefit from \tool's notion of anti-unification over e-graphs
to make them more robust to semantics-preserving transformations.

\mypara{Synthesis and Optimization using \Egraphs}
While traditionally \egraphs have been used
  in SMT solvers for
  facilatiting communication between
  different theories,
  several tools have demonstrated
  their use for optimization
  and synthesis.
\citet{tate2009equality} first used \egraphs for
  equality saturation: a rewrite-driven technique
  for optimizing Java programs with loops.
Since then, several tools have used equality saturation
  for finding programs equivalent to,
  but better than, some
  input program~\cite{willsey2021egg, herbie, tensat, szalinski, diospyros, spores,
  cc1}.
\tool uses an anti-unification
  algorithm on \egraphs
  (together with domain specific
  rewrites), which prior work has not shown.
Additionally, prior work has either used
  greedy or ILP-based extraction strategies,
  whereas \tool uses a new targeted common
  subexpression elimination approach which
  we believe can be used in many other applications of
  equality saturation, especially given its
  amenability to approximation via beam search.

\section{Conclusion and Future Work}

We presented library learning modulo theory (LLMT),
 a technique for learning abstractions from a corpus of programs
 modulo a user-provided equational theory.
We implemented LLMT in \tool.
Our evaluation showed that \tool
 achieves better compression orders of magnitude faster
 than the state of the art.
On a larger benchmark suite of \cogsci programs,
 \tool learns sensible functions that
 compress a dataset that was---until now---too large for library learning techniques.

LLMT and \tool present many avenues for future work.
First,
 our evaluation showed that equational theories are important
 for achieving high compression,
 but these must be provided by domain experts.
Recent work in automated theory synthesis like
 Ruler~\cite{ruler} or \tname{TheSy}~\cite{thesy}
 could aid the user in this task.
Second,
 LLMT uses e-graph anti-unification to generate
 promising abstraction candidates, 
 but this approach is incomplete and misses some patterns
 that could achieve better compression.
% Our prototype indicated that 
%  a complete approach based on traversing the
%  pattern generalization lattice 
%  was prohibitively slow;
%  future work could identify a subset of patterns
%  that does not sacrifice optimality 
%  but can still be enumerated efficiently. 
An exciting direction for future work is to combine LLMT 
with more efficient top-down search from \stitch~\cite{stitch}.
This is challenging because \stitch crucially relies on the ability 
to quickly compute an upper bound on the compression of a given pattern 
by summing up the local compression at each of its matches in the corpus.
This upper bound does not straightforwardly extend to e-graphs
because in an e-graph different matches of a pattern may come from different syntactic variants of the corpus,
and one needs to trade-off the compression from abstractions
against the size difference between different syntactic variants.

\begin{acks}
  We are grateful to the anonymous reviewers for their insightful comments.
  We would like to thank Matthew Bowers for many helpful discussions
  and especially for publishing the \dc compression benchmark:
  we know it was a lot of work to assemble!
  This work has been supported by the National Science Foundation under Grants No. 1911149 and 1943623.
\end{acks}

\bibliography{citations.bib}


\begin{thebibliography}{39}

%%% the Web address.

\ifx \showCODEN    \undefined \def \showCODEN     #1{\unskip}     \fi
\ifx \showDOI      \undefined \def \showDOI       #1{#1}\fi
\ifx \showISBNx    \undefined \def \showISBNx     #1{\unskip}     \fi
\ifx \showISBNxiii \undefined \def \showISBNxiii  #1{\unskip}     \fi
\ifx \showISSN     \undefined \def \showISSN      #1{\unskip}     \fi
\ifx \showLCCN     \undefined \def \showLCCN      #1{\unskip}     \fi
\ifx \shownote     \undefined \def \shownote      #1{#1}          \fi
\ifx \showarticletitle \undefined \def \showarticletitle #1{#1}   \fi
\ifx \showURL      \undefined \def \showURL       {\relax}        \fi
% invisible to TeX
\providecommand\bibfield[2]{#2}
\providecommand\bibinfo[2]{#2}
\providecommand\natexlab[1]{#1}
\providecommand\showeprint[2][]{arXiv:#2}

\bibitem[Bowers(2022)]%
        {compression-bench}
\bibfield{author}{\bibinfo{person}{Matt Bowers}.}
  \bibinfo{year}{2022}\natexlab{}.
\newblock \bibinfo{title}{\dc Compression Benchmark}.
\newblock
\newblock
\urldef\tempurl%
\url{https://github.com/mlb2251/compression\_benchmark}
\showURL{%
\tempurl}


\bibitem[Bowers et~al\mbox{.}(2023)]%
        {stitch}
\bibfield{author}{\bibinfo{person}{Matthew Bowers}, \bibinfo{person}{Theo~X.
  Olausson}, \bibinfo{person}{Catherine Wong}, \bibinfo{person}{Gabriel Grand},
  \bibinfo{person}{Joshua~B. Tenenbaum}, \bibinfo{person}{Kevin Ellis}, {and}
  \bibinfo{person}{Armando Solar-Lezama}.} \bibinfo{year}{2023}\natexlab{}.
\newblock \showarticletitle{Top-Down Synthesis For Library Learning}.
\newblock \bibinfo{journal}{\emph{Proceedings of the ACM on Programming
  Languages}} \bibinfo{volume}{7}, \bibinfo{number}{POPL}
  (\bibinfo{year}{2023}).
\newblock
\urldef\tempurl%
\url{https://doi.org/10.1145/3571234}
\showDOI{\tempurl}


\bibitem[Bulychev et~al\mbox{.}(2010)]%
        {Bulychev10}
\bibfield{author}{\bibinfo{person}{Peter~E. Bulychev}, \bibinfo{person}{Egor~V.
  Kostylev}, {and} \bibinfo{person}{Vladimir~A. Zakharov}.}
  \bibinfo{year}{2010}\natexlab{}.
\newblock \showarticletitle{Anti-unification Algorithms and Their Applications
  in Program Analysis}. In \bibinfo{booktitle}{\emph{Perspectives of Systems
  Informatics}}, \bibfield{editor}{\bibinfo{person}{Amir Pnueli},
  \bibinfo{person}{Irina Virbitskaite}, {and} \bibinfo{person}{Andrei
  Voronkov}} (Eds.). \bibinfo{publisher}{Springer Berlin Heidelberg},
  \bibinfo{address}{Berlin, Heidelberg}, \bibinfo{pages}{413--423}.
\newblock
\showISBNx{978-3-642-11486-1}


\bibitem[Cropper and Dumancic(2022)]%
        {ilp}
\bibfield{author}{\bibinfo{person}{Andrew Cropper} {and}
  \bibinfo{person}{Sebastijan Dumancic}.} \bibinfo{year}{2022}\natexlab{}.
\newblock \showarticletitle{Inductive Logic Programming At 30: {A} New
  Introduction}.
\newblock \bibinfo{journal}{\emph{J. Artif. Intell. Res.}}
  \bibinfo{volume}{74} (\bibinfo{year}{2022}), \bibinfo{pages}{765--850}.
\newblock
\urldef\tempurl%
\url{https://doi.org/10.1613/jair.1.13507}
\showDOI{\tempurl}


\bibitem[Dechter et~al\mbox{.}(2013)]%
        {learner}
\bibfield{author}{\bibinfo{person}{Eyal Dechter}, \bibinfo{person}{Jon
  Malmaud}, \bibinfo{person}{Ryan~P. Adams}, {and} \bibinfo{person}{Joshua~B.
  Tenenbaum}.} \bibinfo{year}{2013}\natexlab{}.
\newblock \showarticletitle{Bootstrap Learning via Modular Concept Discovery}.
  In \bibinfo{booktitle}{\emph{Proceedings of the Twenty-Third International
  Joint Conference on Artificial Intelligence}} (Beijing, China)
  \emph{(\bibinfo{series}{IJCAI '13})}. \bibinfo{publisher}{AAAI Press},
  \bibinfo{pages}{1302--1309}.
\newblock
\showISBNx{9781577356332}


\bibitem[Dong et~al\mbox{.}(2022)]%
        {webrobot}
\bibfield{author}{\bibinfo{person}{Rui Dong}, \bibinfo{person}{Zhicheng Huang},
  \bibinfo{person}{Ian~Iong Lam}, \bibinfo{person}{Yan Chen}, {and}
  \bibinfo{person}{Xinyu Wang}.} \bibinfo{year}{2022}\natexlab{}.
\newblock \showarticletitle{WebRobot: Web Robotic Process Automation Using
  Interactive Programming-by-Demonstration}. In
  \bibinfo{booktitle}{\emph{Proceedings of the 43rd ACM SIGPLAN International
  Conference on Programming Language Design and Implementation}} (San Diego,
  CA, USA) \emph{(\bibinfo{series}{PLDI 2022})}.
  \bibinfo{publisher}{Association for Computing Machinery},
  \bibinfo{address}{New York, NY, USA}, \bibinfo{pages}{152--167}.
\newblock
\showISBNx{9781450392655}
\urldef\tempurl%
\url{https://doi.org/10.1145/3519939.3523711}
\showDOI{\tempurl}


\bibitem[Dumancic et~al\mbox{.}(2021)]%
        {knorf}
\bibfield{author}{\bibinfo{person}{Sebastijan Dumancic}, \bibinfo{person}{Tias
  Guns}, {and} \bibinfo{person}{Andrew Cropper}.}
  \bibinfo{year}{2021}\natexlab{}.
\newblock \showarticletitle{Knowledge Refactoring for Inductive Program
  Synthesis}. In \bibinfo{booktitle}{\emph{Thirty-Fifth {AAAI} Conference on
  Artificial Intelligence, {AAAI} 2021, Thirty-Third Conference on Innovative
  Applications of Artificial Intelligence, {IAAI} 2021, The Eleventh Symposium
  on Educational Advances in Artificial Intelligence, {EAAI} 2021, Virtual
  Event, February 2-9, 2021}}. \bibinfo{publisher}{{AAAI} Press},
  \bibinfo{pages}{7271--7278}.
\newblock
\urldef\tempurl%
\url{https://ojs.aaai.org/index.php/AAAI/article/view/16893}
\showURL{%
\tempurl}


\bibitem[Ellis et~al\mbox{.}(2017)]%
        {handdrawn}
\bibfield{author}{\bibinfo{person}{Kevin Ellis}, \bibinfo{person}{Daniel
  Ritchie}, \bibinfo{person}{Armando Solar-Lezama}, {and}
  \bibinfo{person}{Joshua~B. Tenenbaum}.} \bibinfo{year}{2017}\natexlab{}.
\newblock \bibinfo{title}{Learning to Infer Graphics Programs from Hand-Drawn
  Images}.
\newblock
\newblock
\urldef\tempurl%
\url{https://doi.org/10.48550/ARXIV.1707.09627}
\showDOI{\tempurl}


\bibitem[Ellis et~al\mbox{.}(2021)]%
        {dreamcoder}
\bibfield{author}{\bibinfo{person}{Kevin Ellis}, \bibinfo{person}{Catherine
  Wong}, \bibinfo{person}{Maxwell~I. Nye}, \bibinfo{person}{Mathias
  Sabl{\'{e}}{-}Meyer}, \bibinfo{person}{Lucas Morales},
  \bibinfo{person}{Luke~B. Hewitt}, \bibinfo{person}{Luc Cary},
  \bibinfo{person}{Armando Solar{-}Lezama}, {and} \bibinfo{person}{Joshua~B.
  Tenenbaum}.} \bibinfo{year}{2021}\natexlab{}.
\newblock \showarticletitle{DreamCoder: bootstrapping inductive program
  synthesis with wake-sleep library learning}. In
  \bibinfo{booktitle}{\emph{{PLDI} '21: 42nd {ACM} {SIGPLAN} International
  Conference on Programming Language Design and Implementation, Virtual Event,
  Canada, June 20--25, 2021}}, \bibfield{editor}{\bibinfo{person}{Stephen~N.
  Freund} {and} \bibinfo{person}{Eran Yahav}} (Eds.).
  \bibinfo{publisher}{{ACM}}, \bibinfo{pages}{835--850}.
\newblock
\urldef\tempurl%
\url{https://doi.org/10.1145/3453483.3454080}
\showDOI{\tempurl}


\bibitem[Iyer et~al\mbox{.}(2019)]%
        {srini}
\bibfield{author}{\bibinfo{person}{Srinivasan Iyer}, \bibinfo{person}{Alvin
  Cheung}, {and} \bibinfo{person}{Luke Zettlemoyer}.}
  \bibinfo{year}{2019}\natexlab{}.
\newblock \bibinfo{title}{Learning Programmatic Idioms for Scalable Semantic
  Parsing}.
\newblock
\newblock
\urldef\tempurl%
\url{https://doi.org/10.48550/ARXIV.1904.09086}
\showDOI{\tempurl}


\bibitem[Jones et~al\mbox{.}(2021)]%
        {shapemod}
\bibfield{author}{\bibinfo{person}{R.~Kenny Jones}, \bibinfo{person}{David
  Charatan}, \bibinfo{person}{Paul Guerrero}, \bibinfo{person}{Niloy~J. Mitra},
  {and} \bibinfo{person}{Daniel Ritchie}.} \bibinfo{year}{2021}\natexlab{}.
\newblock \showarticletitle{ShapeMOD: Macro Operation Discovery for 3D Shape
  Programs}.
\newblock \bibinfo{journal}{\emph{ACM Trans. Graph.}} \bibinfo{volume}{40},
  \bibinfo{number}{4}, Article \bibinfo{articleno}{153} (\bibinfo{date}{jul}
  \bibinfo{year}{2021}), \bibinfo{numpages}{16}~pages.
\newblock
\showISSN{0730-0301}
\urldef\tempurl%
\url{https://doi.org/10.1145/3450626.3459821}
\showDOI{\tempurl}


\bibitem[Lau et~al\mbox{.}(2001)]%
        {vsa2}
\bibfield{author}{\bibinfo{person}{Tessa Lau}, \bibinfo{person}{Steven~A.
  Wolfman}, \bibinfo{person}{Pedro Domingos}, {and} \bibinfo{person}{Daniel~S.
  Weld}.} \bibinfo{year}{2001}\natexlab{}.
\newblock \bibinfo{title}{Programming By Demonstration Using Version Space
  Algebra}.
\newblock
\newblock


\bibitem[L\'azaro-Gredilla et~al\mbox{.}(2018)]%
        {gredilla}
\bibfield{author}{\bibinfo{person}{Miguel L\'azaro-Gredilla},
  \bibinfo{person}{Dianhuan Lin}, \bibinfo{person}{J.~Swaroop Guntupalli},
  {and} \bibinfo{person}{Dileep George}.} \bibinfo{year}{2018}\natexlab{}.
\newblock \bibinfo{title}{Beyond imitation: Zero-shot task transfer on robots
  by learning concepts as cognitive programs}.
\newblock
\newblock
\urldef\tempurl%
\url{https://doi.org/10.48550/ARXIV.1812.02788}
\showDOI{\tempurl}


\bibitem[Meng et~al\mbox{.}(2013)]%
        {lase}
\bibfield{author}{\bibinfo{person}{Na Meng}, \bibinfo{person}{Miryung Kim},
  {and} \bibinfo{person}{Kathryn~S. McKinley}.}
  \bibinfo{year}{2013}\natexlab{}.
\newblock \showarticletitle{Lase: Locating and applying systematic edits by
  learning from examples}. In \bibinfo{booktitle}{\emph{2013 35th International
  Conference on Software Engineering (ICSE)}}. \bibinfo{pages}{502--511}.
\newblock
\urldef\tempurl%
\url{https://doi.org/10.1109/ICSE.2013.6606596}
\showDOI{\tempurl}


\bibitem[Mitchell(1977)]%
        {vsa1}
\bibfield{author}{\bibinfo{person}{Tom~Michael Mitchell}.}
  \bibinfo{year}{1977}\natexlab{}.
\newblock \showarticletitle{Version Spaces: A Candidate Elimination Approach to
  Rule Learning}. In \bibinfo{booktitle}{\emph{IJCAI}}.
\newblock


\bibitem[Nandi et~al\mbox{.}(2018)]%
        {reincarnate}
\bibfield{author}{\bibinfo{person}{Chandrakana Nandi},
  \bibinfo{person}{James~R. Wilcox}, \bibinfo{person}{Pavel Panchekha},
  \bibinfo{person}{Taylor Blau}, \bibinfo{person}{Dan Grossman}, {and}
  \bibinfo{person}{Zachary Tatlock}.} \bibinfo{year}{2018}\natexlab{}.
\newblock \showarticletitle{Functional Programming for Compiling and
  Decompiling Computer-Aided Design}.
\newblock \bibinfo{journal}{\emph{Proc. ACM Program. Lang.}}
  \bibinfo{volume}{2}, \bibinfo{number}{ICFP}, Article \bibinfo{articleno}{99}
  (\bibinfo{date}{jul} \bibinfo{year}{2018}), \bibinfo{numpages}{31}~pages.
\newblock
\urldef\tempurl%
\url{https://doi.org/10.1145/3236794}
\showDOI{\tempurl}


\bibitem[Nandi et~al\mbox{.}(2020)]%
        {szalinski}
\bibfield{author}{\bibinfo{person}{Chandrakana Nandi}, \bibinfo{person}{Max
  Willsey}, \bibinfo{person}{Adam Anderson}, \bibinfo{person}{James~R. Wilcox},
  \bibinfo{person}{Eva Darulova}, \bibinfo{person}{Dan Grossman}, {and}
  \bibinfo{person}{Zachary Tatlock}.} \bibinfo{year}{2020}\natexlab{}.
\newblock \showarticletitle{Synthesizing structured {CAD} models with equality
  saturation and inverse transformations}. In
  \bibinfo{booktitle}{\emph{Proceedings of the 41st {ACM} {SIGPLAN}
  International Conference on Programming Language Design and Implementation,
  {PLDI} 2020, London, UK, June 15--20, 2020}},
  \bibfield{editor}{\bibinfo{person}{Alastair~F. Donaldson} {and}
  \bibinfo{person}{Emina Torlak}} (Eds.). \bibinfo{publisher}{{ACM}},
  \bibinfo{pages}{31--44}.
\newblock
\urldef\tempurl%
\url{https://doi.org/10.1145/3385412.3386012}
\showDOI{\tempurl}


\bibitem[Nandi et~al\mbox{.}(2021)]%
        {ruler}
\bibfield{author}{\bibinfo{person}{Chandrakana Nandi}, \bibinfo{person}{Max
  Willsey}, \bibinfo{person}{Amy Zhu}, \bibinfo{person}{Yisu~Remy Wang},
  \bibinfo{person}{Brett Saiki}, \bibinfo{person}{Adam Anderson},
  \bibinfo{person}{Adriana Schulz}, \bibinfo{person}{Dan Grossman}, {and}
  \bibinfo{person}{Zachary Tatlock}.} \bibinfo{year}{2021}\natexlab{}.
\newblock \showarticletitle{Rewrite Rule Inference Using Equality Saturation}.
\newblock \bibinfo{journal}{\emph{Proc. ACM Program. Lang.}}
  \bibinfo{volume}{5}, \bibinfo{number}{OOPSLA}, Article
  \bibinfo{articleno}{119} (\bibinfo{date}{oct} \bibinfo{year}{2021}),
  \bibinfo{numpages}{28}~pages.
\newblock
\urldef\tempurl%
\url{https://doi.org/10.1145/3485496}
\showDOI{\tempurl}


\bibitem[Panchekha et~al\mbox{.}(2015)]%
        {herbie}
\bibfield{author}{\bibinfo{person}{Pavel Panchekha}, \bibinfo{person}{Alex
  Sanchez-Stern}, \bibinfo{person}{James~R. Wilcox}, {and}
  \bibinfo{person}{Zachary Tatlock}.} \bibinfo{year}{2015}\natexlab{}.
\newblock \showarticletitle{Automatically Improving Accuracy for Floating Point
  Expressions}.
\newblock \bibinfo{journal}{\emph{SIGPLAN Not.}} \bibinfo{volume}{50},
  \bibinfo{number}{6} (\bibinfo{date}{jun} \bibinfo{year}{2015}),
  \bibinfo{pages}{1--11}.
\newblock
\showISSN{0362-1340}
\urldef\tempurl%
\url{https://doi.org/10.1145/2813885.2737959}
\showDOI{\tempurl}


\bibitem[Plotkin(1970)]%
        {plotkin1970lattice}
\bibfield{author}{\bibinfo{person}{Gordon Plotkin}.}
  \bibinfo{year}{1970}\natexlab{}.
\newblock \bibinfo{booktitle}{\emph{Lattice Theoretic Properties of
  Subsumption}}.
\newblock \bibinfo{publisher}{Edinburgh University, Department of Machine
  Intelligence and Perception}.
\newblock
\urldef\tempurl%
\url{https://books.google.com/books?id=2p09cgAACAAJ}
\showURL{%
\tempurl}


\bibitem[Raza et~al\mbox{.}(2014)]%
        {raza2014programming}
\bibfield{author}{\bibinfo{person}{Mohammad Raza}, \bibinfo{person}{Natasa
  Milic-Frayling}, {and} \bibinfo{person}{Sumit Gulwani}.}
  \bibinfo{year}{2014}\natexlab{}.
\newblock \showarticletitle{Programming by Example using Least General
  Generalizations}. \bibinfo{publisher}{AAAI - Association for the Advancement
  of Artificial Intelligence}.
\newblock
\urldef\tempurl%
\url{https://www.microsoft.com/en-us/research/publication/programming-by-example-using-least-general-generalizations/}
\showURL{%
\tempurl}


\bibitem[Reynolds(1969)]%
        {Reynolds1969TransformationalSA}
\bibfield{author}{\bibinfo{person}{John~C. Reynolds}.}
  \bibinfo{year}{1969}\natexlab{}.
\newblock \showarticletitle{Transformational systems and the algebraic
  structure of atomic formulas}.
\newblock


\bibitem[Rocha et~al\mbox{.}(2022)]%
        {hardware1}
\bibfield{author}{\bibinfo{person}{{Rodrigo C. O.} Rocha},
  \bibinfo{person}{Pavlos Petoumenos}, \bibinfo{person}{Bj{\"o}rn Franke},
  \bibinfo{person}{Pramod Bhatotia}, {and} \bibinfo{person}{Michael
  O{\textquoteright}Boyle}.} \bibinfo{year}{2022}\natexlab{}.
\newblock \showarticletitle{Loop Rolling for Code Size Reduction}. In
  \bibinfo{booktitle}{\emph{2022 IEEE/ACM International Symposium on Code
  Generation and Optimization (CGO)}}, \bibfield{editor}{\bibinfo{person}{{Jae
  W.} Lee}, \bibinfo{person}{Sebastian Hack}, {and} \bibinfo{person}{Tatiana
  Shpeisman}} (Eds.). \bibinfo{publisher}{IEEE}, \bibinfo{pages}{217--229}.
\newblock
\showISBNx{978-1-6654-0585-0}
\urldef\tempurl%
\url{https://doi.org/10.1109/CGO53902.2022.9741256}
\showDOI{\tempurl}


\bibitem[Rolim et~al\mbox{.}(2017)]%
        {refaser}
\bibfield{author}{\bibinfo{person}{Reudismam Rolim}, \bibinfo{person}{Gustavo
  Soares}, \bibinfo{person}{Loris D'Antoni}, \bibinfo{person}{Oleksandr
  Polozov}, \bibinfo{person}{Sumit Gulwani}, \bibinfo{person}{Rohit Gheyi},
  \bibinfo{person}{Ryo Suzuki}, {and} \bibinfo{person}{Bj\"{o}rn Hartmann}.}
  \bibinfo{year}{2017}\natexlab{}.
\newblock \showarticletitle{Learning Syntactic Program Transformations from
  Examples}. In \bibinfo{booktitle}{\emph{Proceedings of the 39th International
  Conference on Software Engineering}} (Buenos Aires, Argentina)
  \emph{(\bibinfo{series}{ICSE '17})}. \bibinfo{publisher}{IEEE Press},
  \bibinfo{pages}{404--415}.
\newblock
\showISBNx{9781538638682}
\urldef\tempurl%
\url{https://doi.org/10.1109/ICSE.2017.44}
\showDOI{\tempurl}


\bibitem[Sharma et~al\mbox{.}(2017)]%
        {csgnet}
\bibfield{author}{\bibinfo{person}{Gopal Sharma}, \bibinfo{person}{Rishabh
  Goyal}, \bibinfo{person}{Difan Liu}, \bibinfo{person}{Evangelos Kalogerakis},
  {and} \bibinfo{person}{Subhransu Maji}.} \bibinfo{year}{2017}\natexlab{}.
\newblock \bibinfo{title}{CSGNet: Neural Shape Parser for Constructive Solid
  Geometry}.
\newblock
\newblock
\urldef\tempurl%
\url{https://doi.org/10.48550/ARXIV.1712.08290}
\showDOI{\tempurl}


\bibitem[Shin et~al\mbox{.}(2019)]%
        {polozov}
\bibfield{author}{\bibinfo{person}{Eui~Chul Shin}, \bibinfo{person}{Miltiadis
  Allamanis}, \bibinfo{person}{Marc Brockschmidt}, {and} \bibinfo{person}{Alex
  Polozov}.} \bibinfo{year}{2019}\natexlab{}.
\newblock \showarticletitle{Program Synthesis and Semantic Parsing with Learned
  Code Idioms}. In \bibinfo{booktitle}{\emph{Advances in Neural Information
  Processing Systems}}, \bibfield{editor}{\bibinfo{person}{H.~Wallach},
  \bibinfo{person}{H.~Larochelle}, \bibinfo{person}{A.~Beygelzimer},
  \bibinfo{person}{F.~d\textquotesingle Alch\'{e}-Buc},
  \bibinfo{person}{E.~Fox}, {and} \bibinfo{person}{R.~Garnett}} (Eds.),
  Vol.~\bibinfo{volume}{32}. \bibinfo{publisher}{Curran Associates, Inc.}
\newblock
\urldef\tempurl%
\url{https://proceedings.neurips.cc/paper/2019/file/cff34ad343b069ea6920464ad17d4bcf-Paper.pdf}
\showURL{%
\tempurl}


\bibitem[Singher and Itzhaky(2021)]%
        {thesy}
\bibfield{author}{\bibinfo{person}{Eytan Singher} {and}
  \bibinfo{person}{Shachar Itzhaky}.} \bibinfo{year}{2021}\natexlab{}.
\newblock \showarticletitle{Theory Exploration Powered by Deductive Synthesis}.
  In \bibinfo{booktitle}{\emph{Computer Aided Verification}},
  \bibfield{editor}{\bibinfo{person}{Alexandra Silva} {and}
  \bibinfo{person}{K.~Rustan~M. Leino}} (Eds.). \bibinfo{publisher}{Springer
  International Publishing}, \bibinfo{address}{Cham},
  \bibinfo{pages}{125--148}.
\newblock
\showISBNx{978-3-030-81688-9}


\bibitem[Stiff and Vahid(2005)]%
        {hardware2}
\bibfield{author}{\bibinfo{person}{G. Stiff} {and} \bibinfo{person}{F. Vahid}.}
  \bibinfo{year}{2005}\natexlab{}.
\newblock \showarticletitle{New Decompilation Techniques for Binary-Level
  Co-Processor Generation}. In \bibinfo{booktitle}{\emph{Proceedings of the
  2005 IEEE/ACM International Conference on Computer-Aided Design}} (San Jose,
  CA) \emph{(\bibinfo{series}{ICCAD '05})}. \bibinfo{publisher}{IEEE Computer
  Society}, \bibinfo{address}{USA}, \bibinfo{pages}{547--554}.
\newblock
\showISBNx{078039254X}


\bibitem[Su et~al\mbox{.}(1984)]%
        {hardware3}
\bibfield{author}{\bibinfo{person}{Bogong Su}, \bibinfo{person}{Shiyuan Ding},
  {and} \bibinfo{person}{Lan Jin}.} \bibinfo{year}{1984}\natexlab{}.
\newblock \showarticletitle{An Improvement of Trace Scheduling for Global
  Microcode Compaction}.
\newblock \bibinfo{journal}{\emph{SIGMICRO Newsl.}} \bibinfo{volume}{15},
  \bibinfo{number}{4} (\bibinfo{date}{dec} \bibinfo{year}{1984}),
  \bibinfo{pages}{78--85}.
\newblock
\showISSN{1050-916X}
\urldef\tempurl%
\url{https://doi.org/10.1145/384281.808217}
\showDOI{\tempurl}


\bibitem[Tate et~al\mbox{.}(2009)]%
        {tate2009equality}
\bibfield{author}{\bibinfo{person}{Ross Tate}, \bibinfo{person}{Michael Stepp},
  \bibinfo{person}{Zachary Tatlock}, {and} \bibinfo{person}{Sorin Lerner}.}
  \bibinfo{year}{2009}\natexlab{}.
\newblock \showarticletitle{{Equality Saturation: A New Approach to
  Optimization}}. In \bibinfo{booktitle}{\emph{Proceedings of the 36th annual
  ACM SIGPLAN-SIGACT Symposium on Principles of Programming Languages}}.
  \bibinfo{pages}{264--276}.
\newblock


\bibitem[Tian et~al\mbox{.}(2019)]%
        {shape2prog}
\bibfield{author}{\bibinfo{person}{Yonglong Tian}, \bibinfo{person}{Andrew
  Luo}, \bibinfo{person}{Xingyuan Sun}, \bibinfo{person}{Kevin Ellis},
  \bibinfo{person}{William~T. Freeman}, \bibinfo{person}{Joshua~B. Tenenbaum},
  {and} \bibinfo{person}{Jiajun Wu}.} \bibinfo{year}{2019}\natexlab{}.
\newblock \bibinfo{title}{Learning to Infer and Execute 3D Shape Programs}.
\newblock
\newblock
\urldef\tempurl%
\url{https://doi.org/10.48550/ARXIV.1901.02875}
\showDOI{\tempurl}


\bibitem[VanHattum et~al\mbox{.}(2021)]%
        {diospyros}
\bibfield{author}{\bibinfo{person}{Alexa VanHattum}, \bibinfo{person}{Rachit
  Nigam}, \bibinfo{person}{Vincent~T. Lee}, \bibinfo{person}{James Bornholt},
  {and} \bibinfo{person}{Adrian Sampson}.} \bibinfo{year}{2021}\natexlab{}.
\newblock \showarticletitle{Vectorization for Digital Signal Processors via
  Equality Saturation}. In \bibinfo{booktitle}{\emph{Proceedings of the 26th
  ACM International Conference on Architectural Support for Programming
  Languages and Operating Systems}} (Virtual, USA)
  \emph{(\bibinfo{series}{ASPLOS 2021})}. \bibinfo{publisher}{Association for
  Computing Machinery}, \bibinfo{address}{New York, NY, USA},
  \bibinfo{pages}{874--886}.
\newblock
\showISBNx{9781450383172}
\urldef\tempurl%
\url{https://doi.org/10.1145/3445814.3446707}
\showDOI{\tempurl}


\bibitem[Wang et~al\mbox{.}(2021)]%
        {cogsci-smiley}
\bibfield{author}{\bibinfo{person}{Haoliang Wang}, \bibinfo{person}{Nadia
  Polikarpova}, {and} \bibinfo{person}{Judith~E. Fan}.}
  \bibinfo{year}{2021}\natexlab{}.
\newblock \showarticletitle{Learning part-based abstractions for visual object
  concepts}. In \bibinfo{booktitle}{\emph{Proceedings of the Annual Meeting of
  the Cognitive Science Society}}, Vol.~\bibinfo{volume}{43}.
\newblock
\urldef\tempurl%
\url{https://escholarship.org/uc/item/9009w415}
\showURL{%
\tempurl}


\bibitem[Wang et~al\mbox{.}(2020)]%
        {spores}
\bibfield{author}{\bibinfo{person}{Yisu~Remy Wang}, \bibinfo{person}{Shana
  Hutchison}, \bibinfo{person}{Jonathan Leang}, \bibinfo{person}{Bill Howe},
  {and} \bibinfo{person}{Dan Suciu}.} \bibinfo{year}{2020}\natexlab{}.
\newblock \showarticletitle{SPORES: Sum-Product Optimization via Relational
  Equality Saturation for Large Scale Linear Algebra}.
\newblock \bibinfo{journal}{\emph{Proc. VLDB Endow.}} \bibinfo{volume}{13},
  \bibinfo{number}{12} (\bibinfo{date}{jul} \bibinfo{year}{2020}),
  \bibinfo{pages}{1919--1932}.
\newblock
\showISSN{2150-8097}
\urldef\tempurl%
\url{https://doi.org/10.14778/3407790.3407799}
\showDOI{\tempurl}


\bibitem[Willsey et~al\mbox{.}(2021)]%
        {willsey2021egg}
\bibfield{author}{\bibinfo{person}{Max Willsey}, \bibinfo{person}{Chandrakana
  Nandi}, \bibinfo{person}{Yisu~Remy Wang}, \bibinfo{person}{Oliver Flatt},
  \bibinfo{person}{Zachary Tatlock}, {and} \bibinfo{person}{Pavel Panchekha}.}
  \bibinfo{year}{2021}\natexlab{}.
\newblock \showarticletitle{Egg: Fast and Extensible Equality Saturation}.
\newblock \bibinfo{journal}{\emph{Proceedings of the ACM on Programming
  Languages}} \bibinfo{volume}{5}, \bibinfo{number}{POPL}
  (\bibinfo{year}{2021}), \bibinfo{pages}{1--29}.
\newblock


\bibitem[Wong et~al\mbox{.}(2021)]%
        {laps}
\bibfield{author}{\bibinfo{person}{Catherine Wong}, \bibinfo{person}{Kevin
  Ellis}, \bibinfo{person}{Joshua~B. Tenenbaum}, {and} \bibinfo{person}{Jacob
  Andreas}.} \bibinfo{year}{2021}\natexlab{}.
\newblock \bibinfo{title}{Leveraging Language to Learn Program Abstractions and
  Search Heuristics}.
\newblock
\newblock
\urldef\tempurl%
\url{https://doi.org/10.48550/ARXIV.2106.11053}
\showDOI{\tempurl}


\bibitem[Wong et~al\mbox{.}(2022)]%
        {cogsci-dataset}
\bibfield{author}{\bibinfo{person}{Catherine Wong}, \bibinfo{person}{William~P.
  McCarthy}, \bibinfo{person}{Gabriel Grand}, \bibinfo{person}{Yoni Friedman},
  \bibinfo{person}{Joshua~B. Tenenbaum}, \bibinfo{person}{Jacob Andreas},
  \bibinfo{person}{Robert~D. Hawkins}, {and} \bibinfo{person}{Judith~E. Fan}.}
  \bibinfo{year}{2022}\natexlab{}.
\newblock \showarticletitle{Identifying concept libraries from language about
  object structure}. In \bibinfo{booktitle}{\emph{Proceedings of the Annual
  Meeting of the Cognitive Science Society}}.
\newblock


\bibitem[Wu et~al\mbox{.}(2019)]%
        {cc1}
\bibfield{author}{\bibinfo{person}{Chenming Wu}, \bibinfo{person}{Haisen Zhao},
  \bibinfo{person}{Chandrakana Nandi}, \bibinfo{person}{Jeffrey~I. Lipton},
  \bibinfo{person}{Zachary Tatlock}, {and} \bibinfo{person}{Adriana Schulz}.}
  \bibinfo{year}{2019}\natexlab{}.
\newblock \showarticletitle{Carpentry Compiler}.
\newblock \bibinfo{journal}{\emph{ACM Trans. Graph.}} \bibinfo{volume}{38},
  \bibinfo{number}{6}, Article \bibinfo{articleno}{195} (\bibinfo{date}{nov}
  \bibinfo{year}{2019}), \bibinfo{numpages}{14}~pages.
\newblock
\showISSN{0730-0301}
\urldef\tempurl%
\url{https://doi.org/10.1145/3355089.3356518}
\showDOI{\tempurl}


\bibitem[Yang et~al\mbox{.}(2021)]%
        {tensat}
\bibfield{author}{\bibinfo{person}{Yichen Yang}, \bibinfo{person}{Phitchaya
  Phothilimthana}, \bibinfo{person}{Yisu Wang}, \bibinfo{person}{Max Willsey},
  \bibinfo{person}{Sudip Roy}, {and} \bibinfo{person}{Jacques Pienaar}.}
  \bibinfo{year}{2021}\natexlab{}.
\newblock \showarticletitle{Equality Saturation for Tensor Graph
  Superoptimization}. In \bibinfo{booktitle}{\emph{Proceedings of Machine
  Learning and Systems}}, \bibfield{editor}{\bibinfo{person}{A.~Smola},
  \bibinfo{person}{A.~Dimakis}, {and} \bibinfo{person}{I.~Stoica}} (Eds.),
  Vol.~\bibinfo{volume}{3}. \bibinfo{pages}{255--268}.
\newblock
\urldef\tempurl%
\url{https://proceedings.mlsys.org/paper/2021/file/65ded5353c5ee48d0b7d48c591b8f430-Paper.pdf}
\showURL{%
\tempurl}


\end{thebibliography}

% \iflong
\clearpage
\appendix
\section{Proofs}\label{app:proofs}

\begin{proposition}
  The language of compressed terms is strongly normalizing.
\end{proposition}

\begin{proof}
  Consider evaluating $\hat{t}$ in applicative order (leftmost innermost).
  In this case, any left-hand size of a $\beta$-reduction has no inner $\beta$-redexes.
  Since in our language a $\lambda$-abstraction can only be on the left-hand side of a $\beta$-redex,
  it means that expression being reduced has no inner $\lambda$-abstractions.
  For that reason, the number of $\lambda$-abstractions in an expression strictly decreases with every $\beta$-step
  (the sole one in the reduced redex disappears, and all other $\lambda$-abstractions are outside of the reduced redex, and hence unchanged).
\end{proof}

\begin{lemma}
Given a $\kappa$-step $\hat{t}_1 \rwsteps{\kappa(p')} \hat{t}_2$
  if a pattern $p$ has $N$ matches in $\hat{t}_2$, it also has at least $N$ matches in $\hat{t}_1$.
  \label{lem:pat-match}
\end{lemma}
\begin{proof}
  There are four cases for what $p$ can match in $\hat{t}_2$:
  \begin{itemize}
  \item a subterm that does not overlap with the newly introduced $\beta$-redex: in this case, the match is unaffected by the $\kappa$-step;
  \item a subterm inside an actual argument of the $\beta$-redex: this actual argument becomes a sub-term of $\hat{t}_1$;
  \item a sub-term inside the body of the $\beta$-redex: $\hat{t}_1$ has a sub-term that is more specific than the body, and hence still matches $p$;
  \item a sub-term that includes the entire $\beta$-redex inside: since $p$ is first order, it cannot mention the redex, so the redex must be contained entirely inside the substitution; 
  hence again the match is unaffected by the $\kappa$-step.
  \end{itemize}
\end{proof}

\begin{theorem}[Soundness and Completeness of Pattern-Based Library Learning]
  For any term $t \in \termsof{\Sigma}$ and compressed term $\hat{t} \in \absof{\Sigma}{\mathcal{X}}$:
  \begin{description}[font=\itshape]
    \item[(Soundness)] If $t$ compresses into $\hat{t}$, then $\hat{t}$ evaluates to $t$: $\forall \patset . t \rewrites{\kappa(\patset)} \hat{t} \Longrightarrow \hat{t} \evals t$.
    \item[(Completeness)] If $\hat{t}$ is a solution to the (global) library learning problem, 
    then $t$ compresses into $\hat{t}$ using only patterns that have a match in $t$:
    $\hat{t} \in \argmin_{\hat{t}' \evals t} \sz(\hat{t}') \Longrightarrow t \rewrites{\kappa(\patset)} \hat{t}$, 
    where $\patset = \{p \in \patsof{\Sigma}{\mathcal{X}} \mid t' \in \subterms(t) , t' \sub p\}$.
  \end{description}
\label{thm:pat-lib-learning-app}
\end{theorem} 

\begin{proof}
The soundness is trivial because inverting any $\kappa$-rewrite gives a valid $\beta$-reduction.

The other direction (completeness) is more involved.
First let us prove that $\hat{t}$ can be obtain by any compression, regardless of whether the patterns occur in $t$.
This is non-trivial
because inverting a $\beta$-reduction $(\lambda \many{X} \bnd \hat{t}_1)\ \many{\hat{t}_2}  \betared  \hat{t}'$
does not always correspond to a well-form $\kappa$-step,
for two reasons:
(a) $\hat{t}_1$ itself contains $\lambda$-abstractions (and hence does not correspond to a pattern);
(b) $\many{X} \neq \fv(\hat{t}_1)$.
To overcome point (a), note that in an applicative evaluation,
both $\hat{t}_1$ and $\hat{t}_2$ contain no $\beta$-redexes (and hence no $\lambda$-abstractions).
Hence  $\hat{t}_1$ is a valid pattern.

For point (b), there are two cases: either (1) $\fv(\hat{t}_1) \subsetneq \many{X}$
or (2) $\fv(\hat{t}_1) \supsetneq \many{X}$.
Consider case (1).
In this case, the compressed term $\hat{t}$ has $\lambda$-bindings that are not used in their bodies.
Such a term cannot possibly be minimal.
% because adding unused bindings does not change the body of the lambda, 
% it cannot create more sharing, only less.
%
Removing unused bindings makes each individual $\beta$-redex smaller and cannot remove sharing,
so it always makes the overall term smaller as well.

Now consider case (2).
This means that $\lambda \many{X} \bnd \hat{t}_1$ has free variables, 
which must be defined in outer $\lambda$-abstractions.
However, using such a $\lambda$-abstraction contradicts the assumption that we are only interested in global library learning.

Finally, let us prove that all patterns used in the compression have a match in $t$.
Consider a compression $t = \hat{t}_0 \rwsteps{\kappa(p_1)} \hat{t}_1 \dots \rwsteps{\kappa(p_n)}  \hat{t}_n = \hat{t}$.
Clearly, each pattern $p_i$ has a match in the (compressed) term $\hat{t}_{i-1}$, by definition of the $\kappa$-rewrite.
Hence, it also occurs in $t$, by induction using Lemma~\ref{lem:pat-match}.
\end{proof}

\newpage
\section{Additional Experiments}\label{app:eval}

\begin{figure}[h]
  \centering
  \begin{subfigure}{0.48\linewidth}
    \includegraphics[width=\linewidth]{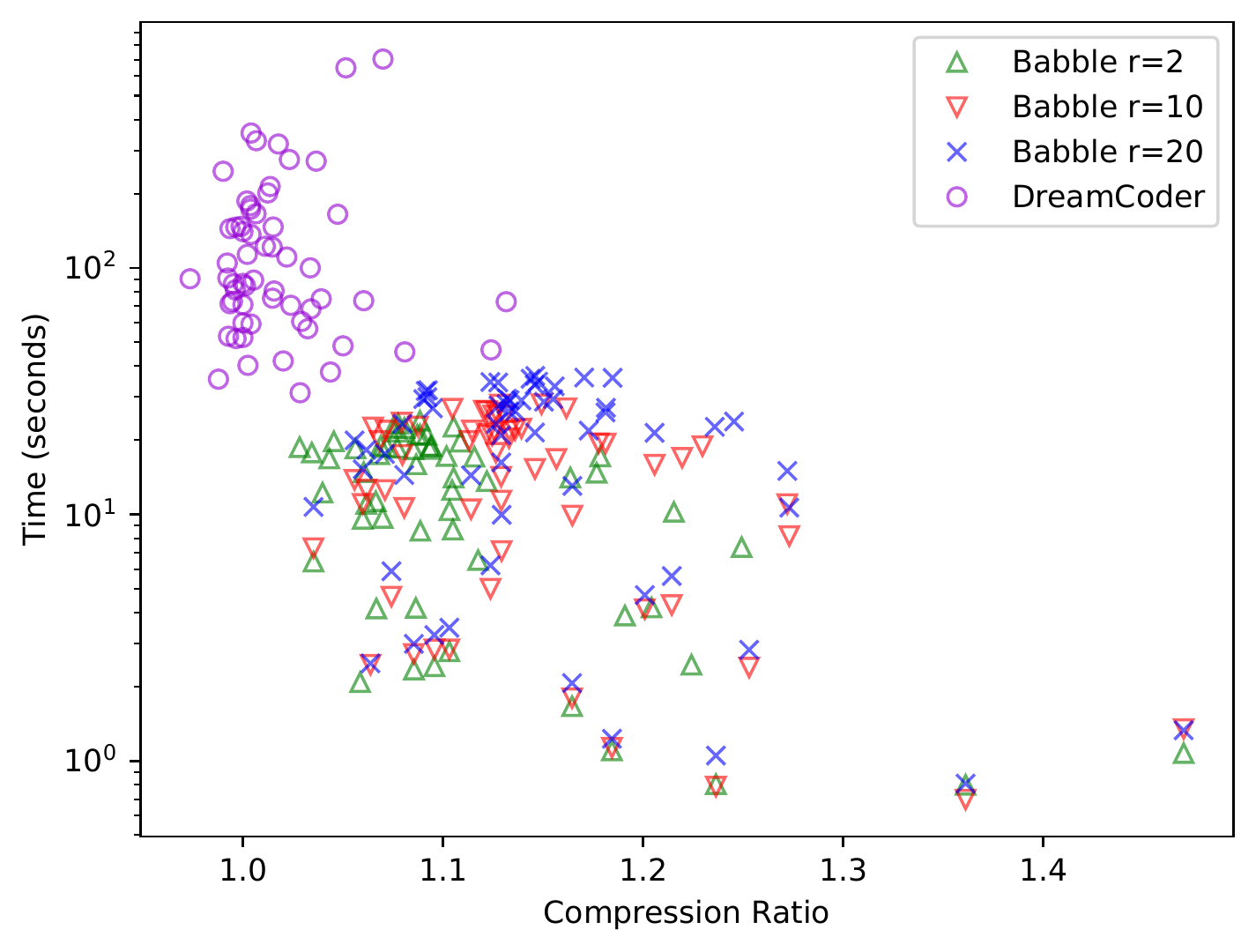}
    \caption{List domain}
  \end{subfigure}
  \hfill
  \begin{subfigure}{0.48\linewidth}
    \includegraphics[width=\linewidth]{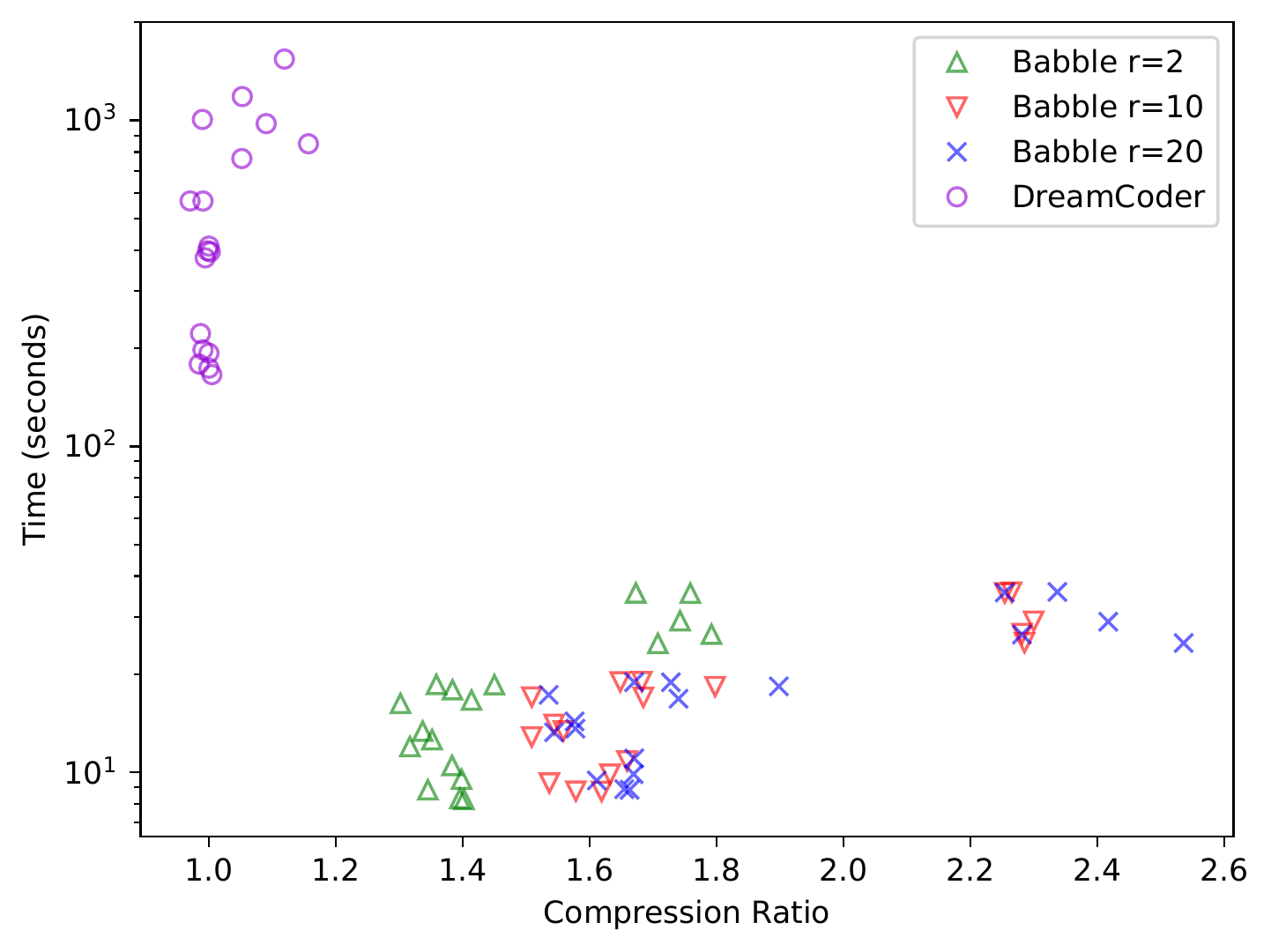}
    \caption{Physics domain}
  \end{subfigure}

  \begin{subfigure}{0.48\linewidth}
    \includegraphics[width=\linewidth]{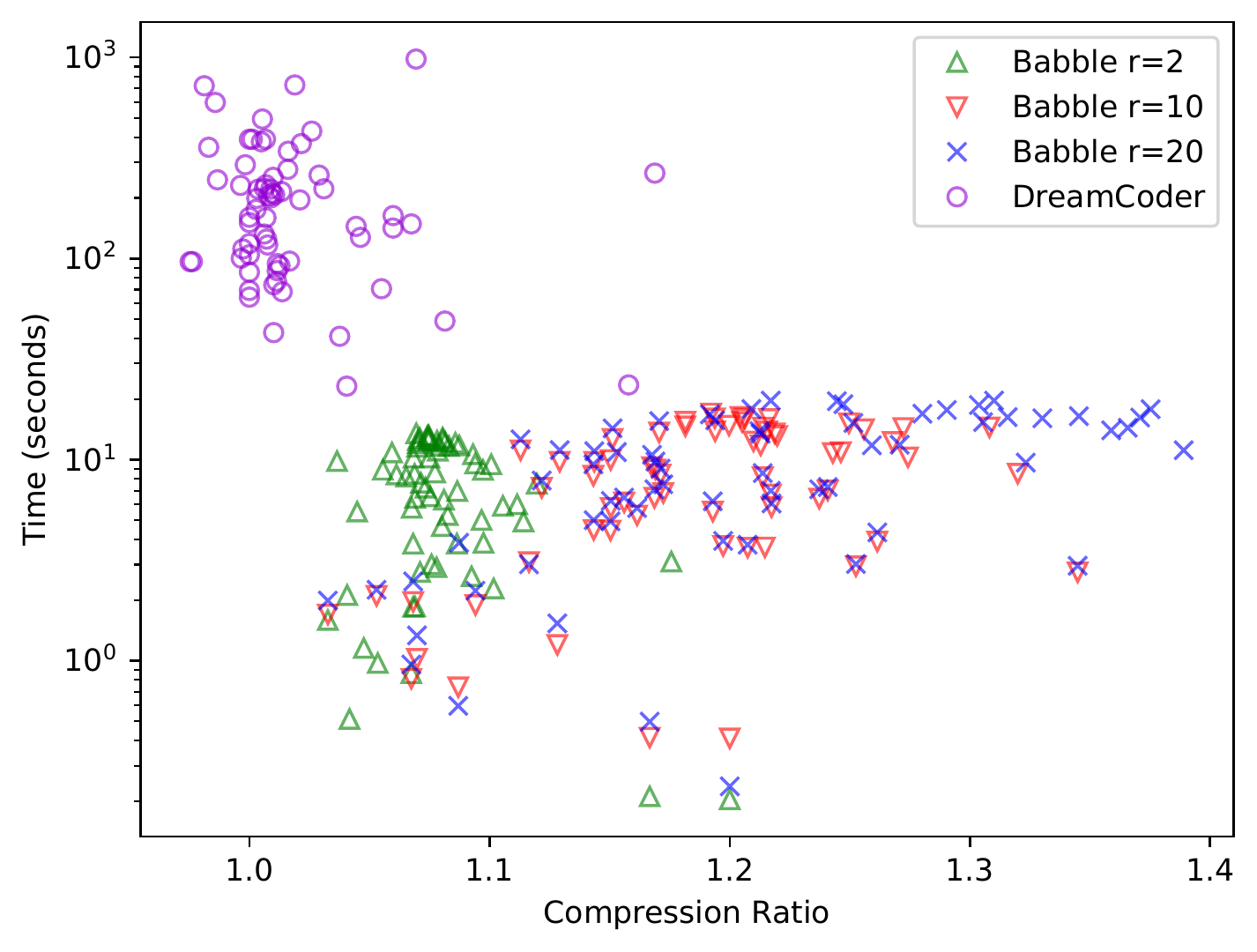}
    \caption{Text domain}
  \end{subfigure}
  \hfill
  \begin{subfigure}{0.48\linewidth}
    \includegraphics[width=\linewidth]{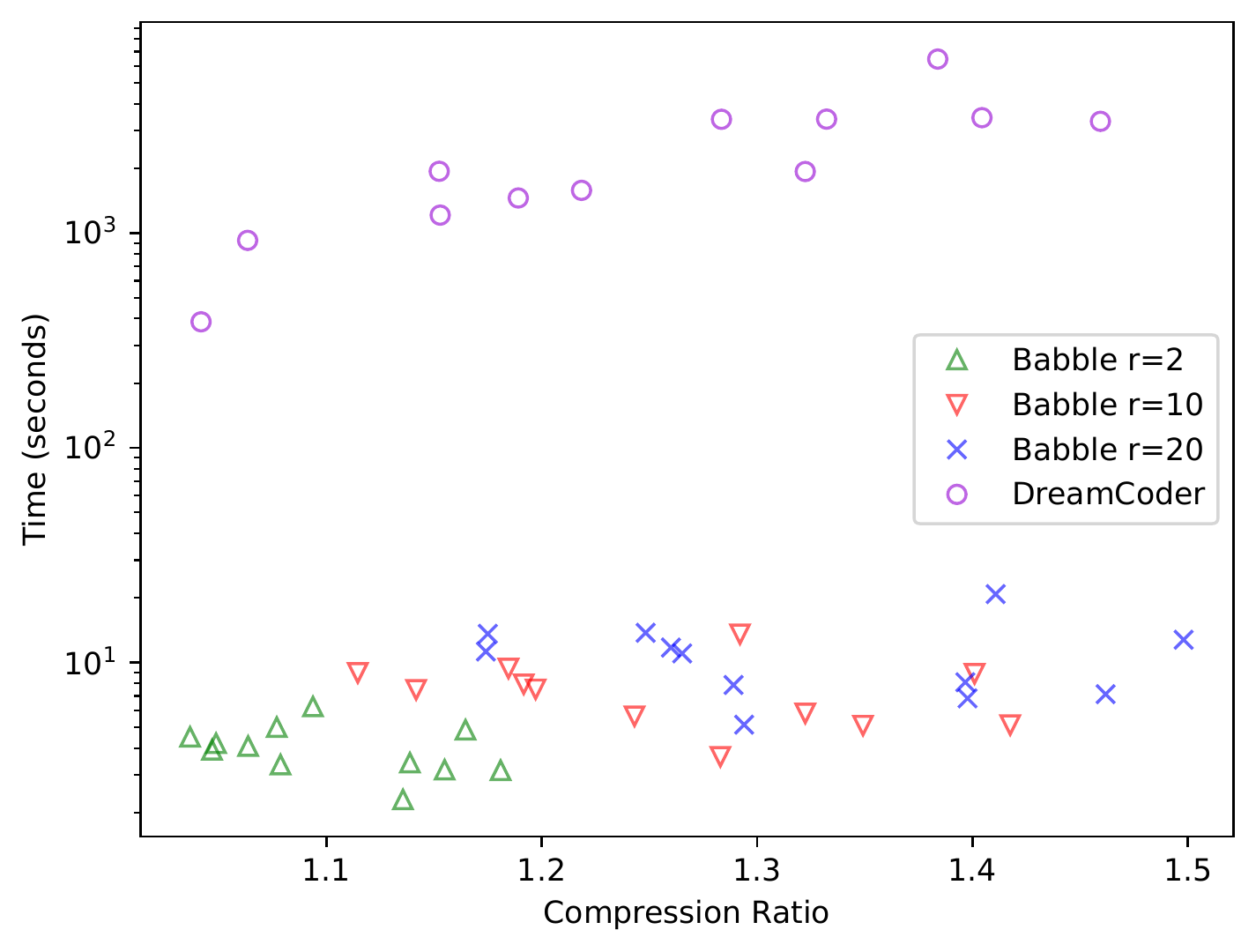}
    \caption{Logo domain}
  \end{subfigure}

  \begin{subfigure}{0.48\linewidth}
    \includegraphics[width=\linewidth]{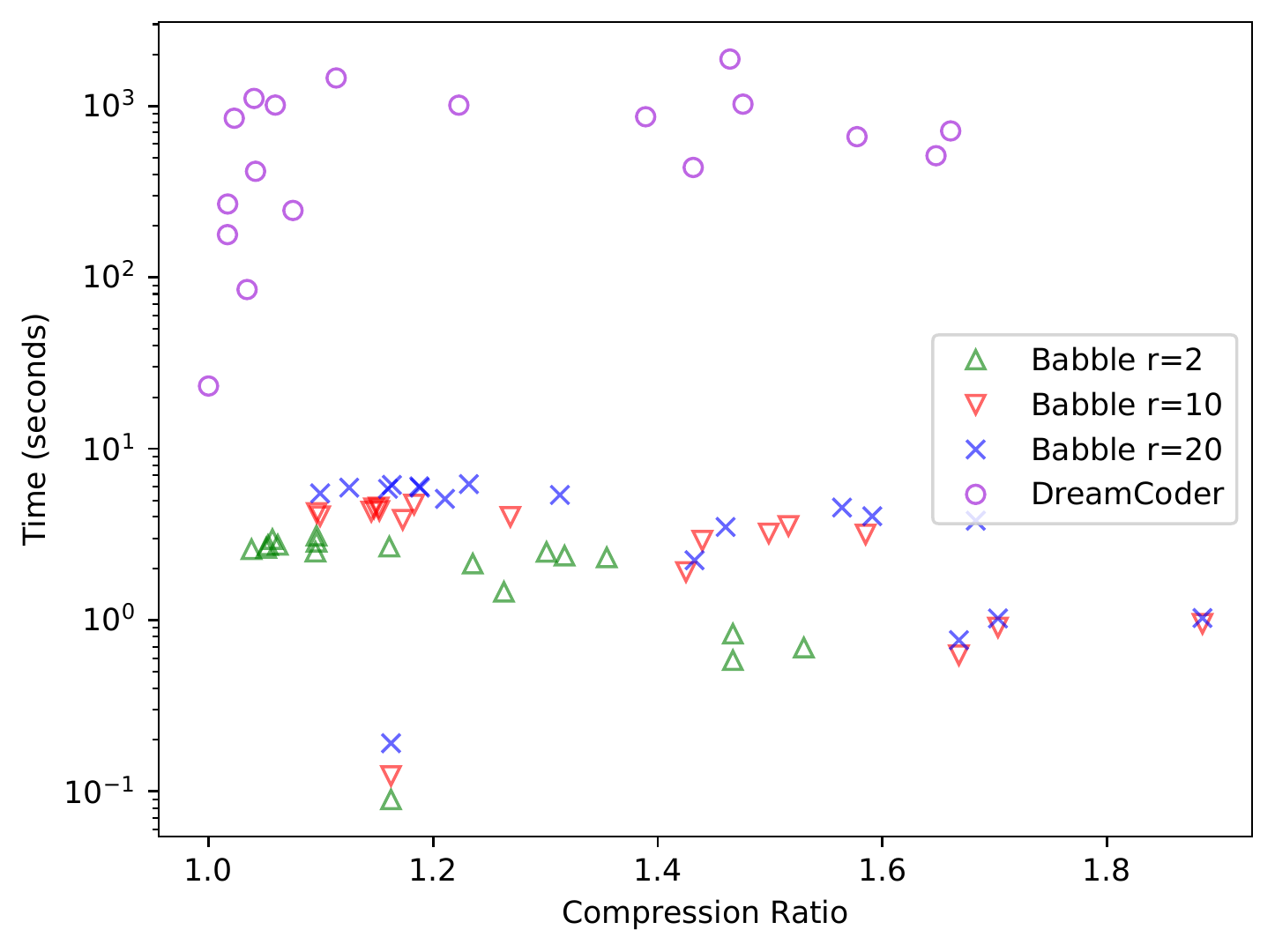}
    \caption{Towers domain}
  \end{subfigure}
  \caption{
    \tool's default configuration runs for 20 rounds (r=20)
    to learn additional library functions.
    This plot shows the data from \autoref{fig:dc-domains},
     but with additional \tool configuration running fewer rounds.
    With fewer rounds, \tool runs faster but compresses worse.
  }\label{fig:dc-domains-rounds}
\end{figure}
% \fi

\end{document}